\documentclass[11pt]{article}

\usepackage{amsthm,amsmath}
\RequirePackage[numbers]{natbib}
\RequirePackage[colorlinks,citecolor=blue,urlcolor=blue]{hyperref}

\usepackage{algorithm}
\usepackage{enumerate}
\usepackage{ifthen}
\usepackage[utf8]{inputenc}
\usepackage{framed}

\usepackage[utf8]{inputenc}
\usepackage{makeidx}
\usepackage{amsfonts}
\usepackage{amssymb}
\usepackage{booktabs}

\usepackage{gensymb}
\usepackage{tabularx}
\usepackage{float}
\usepackage{caption}
\usepackage[labelformat=simple]{subcaption}
\usepackage{bbm}
\usepackage{lipsum}
\usepackage{enumitem}
\usepackage{hyperref}
\usepackage[english]{babel} 
\usepackage{bm} 
\usepackage[T1]{fontenc}	 
\usepackage{lmodern}
\usepackage{amsthm} 
\usepackage{graphicx} 
\usepackage{booktabs} 
\usepackage{multicol} 
\usepackage{tocstyle} 
\usepackage[toc,page]{appendix}

\newcommand{\footremember}[2]{%
	\footnote{#2}
	\newcounter{#1}
	\setcounter{#1}{\value{footnote}}%
}
\newcommand{\footrecall}[1]{%
	\footnotemark[\value{#1}]%
}

\DeclareMathOperator*{\argmax}{arg\,max}
\DeclareMathOperator*{\argmin}{arg\,min}

\newcommand{\betab}{\bm{\beta}}

\newcommand{\epsilonb}{\bm{\varepsilon}}
\newcommand{\Xb}{\bm{X}}
\newcommand{\Yb}{\bm{Y}}
\newcommand{\yb}{\bm{Y}}
\newcommand{\zb}{\bm{z}}

\newcommand{\R}{\mathbb{R}}

\newcommand{\N}{\mathcal{N}}

\newcommand{\zerob}{\mathbf{0}}

\newcommand{\Ib}{\mathbf{I}}

\allowdisplaybreaks
\numberwithin{equation}{section}

\newtheorem{definition}{Definition}
\theoremstyle{plain}
\newtheorem{theorem}{Theorem}
\newtheorem{lemma}{Lemma}

\newtheorem{proposition}{Proposition}
\newtheorem{remark}{Remark}

\newcommand{\numbereqn}{\addtocounter{equation}{1}\tag{\theequation}} 

\allowdisplaybreaks

\begin{document}

\hypersetup{linkcolor=blue}

\date{\today}

\author{Ray Bai\footremember{USCStat}{Department of Statistics, University of South Carolina, Columbia, SC 29208.}\thanks{Co-first author. Email: \href{mailto:RBAI@mailbox.sc.edu}{\tt RBAI@mailbox.sc.edu} }, Gemma E. Moran\footremember{Wharton}{Data Science Institute, Columbia University, New York, NY 10027.}\thanks{Co-first author. Email: \href{gm2918@columbia.edu}{\tt gm2918@columbia.edu} } , Joseph L. Antonelli\footremember{UFStat}{Department of Statistics, University of Florida, Gainesville, FL 32611.}\thanks{Co-first author. Email: \href{mailto:jantonelli@ufl.edu}{\tt jantonelli@ufl.edu} } , \\
Yong Chen\footremember{DBEI}{Department of Biostatistics, Epidemiology, and Informatics, University of Pennsylvania, Philadelphia, PA 19104.}, Mary R. Boland\footrecall{DBEI}}

\title{Spike-and-Slab Group Lassos for Grouped Regression and Sparse Generalized Additive Models 
}

\maketitle

\begin{abstract}
We introduce the spike-and-slab group lasso (SSGL) for Bayesian estimation and variable selection in linear regression with grouped variables. We further extend the SSGL to sparse generalized additive models (GAMs), thereby introducing the first nonparametric variant of the spike-and-slab lasso methodology. Our model simultaneously performs group selection and estimation, while our fully Bayes treatment of the mixture proportion allows for model complexity control and automatic self-adaptivity to different levels of sparsity. We develop theory to uniquely characterize the global posterior mode under the SSGL and introduce a highly efficient block coordinate ascent algorithm for maximum a posteriori (MAP) estimation. We further employ de-biasing methods to provide uncertainty quantification of our estimates. Thus, implementation of our model avoids the computational intensiveness of Markov chain Monte Carlo (MCMC) in high dimensions. We derive posterior concentration rates for both grouped linear regression and sparse GAMs when the number of covariates grows at nearly exponential rate with sample size. Finally, we illustrate our methodology through extensive simulations and data analysis.
\end{abstract}

\section{Introduction} \label{intro}

\subsection{Regression with Grouped Variables} \label{groupedregression}

Group structure arises in many statistical applications. For example, in multifactor analysis of variance, multi-level categorical predictors are each represented by a group of dummy variables. In genomics, genes within the same pathway may form a group at the pathway or gene set level and act in tandem to regulate a biological system. In each of these scenarios, the response $\yb_{n \times 1}$ can be modeled as a linear regression problem with $G$ groups:
\begin{equation} \label{groupmodel}
\yb = \displaystyle \sum_{g=1}^{G} \Xb_g \betab_g + \bm{\varepsilon},
\end{equation}
where $\bm{\varepsilon} \sim \N_n ( \bm{0}, \sigma^2 \bm{I}_n)$, $\betab_g$ is a coefficients vector of length $m_g$, and $\bm{X}_g$ is an $n \times m_g$ covariate matrix corresponding to group $g = 1, \ldots G$.  Even in the absence of grouping information about the covariates, the model (\ref{groupmodel}) subsumes a wide class of important nonparametric regression models called \textit{generalized additive models} (GAMs). In GAMs, continuous covariates may be represented by groups of basis functions which have a nonlinear relationship with the response.  We defer further discussion of GAMs to Section \ref{NPSSLIntro}.

It is often of practical interest to select groups of variables that are most significantly associated with the response. To facilitate this group-level selection, \citet{YuanLin2006} introduced the group lasso, which solves the optimization problem,
\begin{equation} \label{grouplassoobj}
\displaystyle \argmin_{\betab} \frac{1}{2 } \lVert \yb - \displaystyle \sum_{g=1}^{G} \Xb_g \betab_g \rVert_2^2 + \lambda \displaystyle \sum_{g=1}^{G} \sqrt{m_g} \lVert \betab_g \rVert_2,
\end{equation}
where $|| \cdot ||_2$ is the $\ell_2$ norm. In the frequentist literature, many variants of model (\ref{grouplassoobj}) have been introduced, which use some combination of $\ell_1$ and $\ell_2$ penalties on the coefficients of interest (e.g., \cite{JacobObozinskiVert2009, LiNanZhu2015, SimonFriedmanHastieTibshirani2013}).  

In the Bayesian framework, selection of relevant groups under model (\ref{groupmodel}) is often done by placing spike-and-slab priors on each of the groups $\betab_g$ (e.g., \cite{XuGhosh2015, LiquetMengersenPettittSutton2017, YangNarisetty2019, NingGhosal2018}). These priors typically take the form,
\begin{align} \label{pointmassspikeandslab}
\begin{array}{l}
\pi(\bm{\beta} \vert \bm{\gamma}) = \displaystyle \prod_{g=1}^{G} [ (1-\gamma_g) \delta_0 (\betab_g) + \gamma_g \pi(\betab_g) ],\\
\pi(\bm{\gamma} | \theta) = \displaystyle \prod_{g=1}^{G} \theta^{\gamma_g} (1-\theta)^{1-\gamma_g}, \\
\theta \sim \pi(\theta),
\end{array}
\end{align}
where $\bm{\gamma}$ is a binary vector that indexes the $2^G$ possible models, $\theta \in (0, 1)$ is the mixing proportion, $\delta_{0}$ is a point mass at $\bm{0}_{m_g} \in \R^{m_g}$ (the ``spike''), and $\pi(\betab_g)$ is an appropriate ``slab'' density (typically a multivariate normal distribution or a scale-mixture multivariate normal density). With a well-chosen prior on $\theta$, this model will favor parsimonious models in very high dimensions, thus avoiding the curse of dimensionality.

\subsection{The Spike-and-Slab Lasso} \label{spikeandslablasso}
For Bayesian variable selection, point mass spike-and-slab priors (\ref{pointmassspikeandslab}) are interpretable, but they are computationally intractable in high dimensions, due in large part to the combinatorial complexity of updating the discrete indicators $\bm{\gamma}$. As an alternative, fully continuous variants of spike-and-slab models have been developed. For continuous spike-and-slab models, the point mass spike $\delta_{0}$ is replaced by a continuous density heavily concentrated around $\bm{0}_{m_g}$. This not only mimics the point mass but it \textit{also} facilitates more efficient computation, as we describe later.

In the context of sparse normal means estimation and univariate linear regression, \citet{Rockova2018} and \citet{RockovaGeorge2018} introduced the univariate spike-and-slab lasso (SSL). The SSL places a mixture prior of two Laplace densities on the individual coordinates $\beta_j$, i.e.
\begin{equation} \label{SSlasso}
\pi( \betab | \theta) = \displaystyle \prod_{j=1}^{p} [( 1 - \theta) \psi (\beta_j | \lambda_0 ) + \theta \psi (\beta_j | \lambda_1 )],
\end{equation}
where $\theta \in (0,1)$ is the mixing proportion and $\psi(\cdot | \lambda )$ denotes a univariate Laplace density indexed by hyperparameter $\lambda$, i.e. $\psi(\beta | \lambda ) = \frac{\lambda}{2} e^{-\lambda | \beta | }$.  Typically, we set $\lambda_0 \gg \lambda_1$ so that the spike is heavily concentrated about zero. Unlike \eqref{pointmassspikeandslab}, the SSL model \eqref{SSlasso} does not place any mass on exactly sparse vectors. Nevertheless, the global posterior mode under the SSL prior may be exactly sparse. Meanwhile, the slab stabilizes posterior estimates of the larger coefficients so they are not downward biased. Thus, the SSL posterior mode can be used to perform variable selection and estimation simultaneously.

The spike-and-slab lasso methodology has now been adopted for a wide number of statistical problems. Apart from univariate linear regression, it has been used for factor analysis \cite{RockovaGeorge2016, MoranRockovaGeorge2019}, multivariate regression \cite{DeshpandeRockovaGeorge2018}, covariance/precision matrix estimation \cite{DeshpandeRockovaGeorge2018, GanNarisettyLiang2018, LiMcCormickClark2019}, causal inference \cite{AntonelliParmigianiDominici2019}, generalized linear models (GLMs) \cite{TangShenZhangYi2017GLM, TangShenLiZhangWenQianZhuangShiYi2018}, and Cox proportional hazards models \cite{TangShenZhangYi2017Cox}. 

While the SSL (\ref{SSlasso}) induces sparsity on individual coefficients (through the posterior mode), it does not account for group structure of covariates. For inference with structured data in GLMs, \citet{TangShenLiZhangWenQianZhuangShiYi2018} utilized the univariate spike-and-slab lasso prior \eqref{SSlasso} for grouped data where each group had a group-specific sparsity-inducing parameter, $\theta_g$, instead of a single $\theta$ for all coefficients. However, this univariate SSL prior does not feature the ``all in, all out'' selection property of the original group lasso of \citet{YuanLin2006} or the \emph{grouped} and \emph{multivariate} SSL prior, which we develop in this work. 

In this paper, we introduce the \textit{spike-and-slab group lasso} (SSGL) for Bayesian grouped regression and variable selection. Under the SSGL prior, the global posterior mode is exactly sparse, thereby allowing the mode to automatically threshold out insignificant groups of coefficients. To widen the use of spike-and-slab lasso methodology for situations where the linear model is too inflexible, we extend the SSGL to sparse generalized additive models by introducing the \textit{nonparametric spike-and-slab lasso} (NPSSL). To our knowledge, our work is the first to apply the spike-and-slab lasso methodology outside of a parametric setting. Our contributions can be summarized as follows:
\begin{enumerate}
	\item We propose a new group spike-and-slab prior for estimation and variable selection in both parametric and nonparametric settings. Unlike frequentist methods which rely on separable penalties, our model has a \textit{non}-separable and self-adaptive penalty which allows us to automatically adapt to ensemble information about sparsity. 
	\item We introduce a highly efficient block coordinate ascent algorithm for global posterior mode estimation. This allows us to rapidly identify significant groups of coefficients, while thresholding out insignificant ones.
	\item We show that de-biasing techniques that have been used for the original lasso \citep{Tibshirani1996} can be extended to our SSGL model to provide valid inference on the estimated regression coefficients. 
	\item For both grouped regression and sparse additive models, we derive near-optimal posterior contraction rates for both the regression coefficients $\betab$ \textit{and} the unknown variance $\sigma^2$ under the SSGL prior. 
\end{enumerate}
The rest of the paper is structured as follows. In
Section \ref{SSGroupLassoIntro}, we introduce the spike-and-slab group lasso (SSGL). In Section \ref{optimizationSSGL}, we characterize the global posterior mode and introduce efficient algorithms for fast maximum \emph{a posteriori} (MAP) estimation and variable selection. In Section \ref{sec:inference}, we utilize ideas from the de-biased lasso to perform inference on the SSGL model. In Section \ref{NPSSLIntro}, we extend the SSGL to nonparametric settings by proposing the nonparametric spike-and-slab lasso (NPSSL). In Section \ref{asymptotictheory}, we present asymptotic theory for the SSGL and the NPSSL. Finally, in Sections \ref{Simulations} and \ref{dataanalysis}, we provide extensive simulation studies and use our models to analyze real data sets.

\subsection{Notation}
We use the following notations. For two nonnegative sequences $\{ a_n \}$ and $\{ b_n \}$, we write $a_n \asymp b_n$ to denote $0 < \lim \inf_{n \rightarrow \infty} a_n/b_n \leq \lim \sup_{n \rightarrow \infty} a_n/b_n < \infty$. If $\lim_{n \rightarrow \infty} a_n/b_n = 0$, we write $a_n = o(b_n)$ or $a_n \prec b_n$. We use $a_n \lesssim b_n$ or $a_n = O(b_n)$ to denote that for sufficiently large $n$, there exists a constant $C >0$ independent of $n$ such that $a_n \leq Cb_n$. For a vector $\bm{v} \in \mathbb{R}^p$, we let $\lVert \bm{v} \rVert_1 := \sum_{i=1}^p |v_i|$, $\lVert \bm{v} \rVert_2 := \sqrt{ \sum_{i=1}^p v_i^2}$, and $\lVert \bm{v} \rVert_{\infty} := \max_{1 \leq i \leq p} | v_i |$  denote the $\ell_1$, $\ell_2$, and $\ell_{\infty}$ norms respectively. For a symmetric matrix $\bm{A}$, we let $\lambda_{\min} (\bm{A})$ and $\lambda_{\max} (\bm{A})$ denote its minimum and maximum eigenvalues.

\section{The Spike-and-Slab Group Lasso} \label{SSGroupLassoIntro}

Let $\betab_g$ denote a real-valued vector of length $m_g$. We define the \textit{group lasso density} as 
\begin{equation} \label{grouplassoprior}
\bm{\Psi} ( \betab_g | \lambda ) = C_g \lambda^{m_g} \exp \left( - \lambda \lVert \betab_g \rVert_2 \right),
\end{equation}
where $C_g = 2^{-m_g} \pi^{-(m_g-1)/2} \left[ \Gamma \left( (m_g+1)/2 \right) \right]^{-1}$. This prior has been previously considered by \cite{KyungGillGhoshCasella2010, XuGhosh2015} for Bayesian inference in the grouped regression model (\ref{groupmodel}). \citet{KyungGillGhoshCasella2010} considered a single prior (\ref{grouplassoprior}) on each of the $\betab_g$'s, while \citet{XuGhosh2015} employed (\ref{grouplassoprior}) as the slab in the point-mass mixture (\ref{pointmassspikeandslab}). These authors implemented their models using MCMC.

In this manuscript, we introduce a \textit{continuous} spike-and-slab prior with the group lasso density (\ref{grouplassoprior}) for both the spike \textit{and} the slab. The continuous nature of our prior is critical in facilitating efficient coordinate ascent algorithms for MAP estimation that allow us to bypass the use of MCMC. Letting $\betab = ( \betab_1^T, \ldots, \betab_G^T)^T$ under model (\ref{groupmodel}), the \textit{spike-and-slab group lasso} (SSGL) is defined as:
\begin{equation} \label{ssgrouplasso}
\pi (\betab | \theta ) = \displaystyle \prod_{g=1}^{G} \left[ (1- \theta) \bm{\Psi} ( \betab_g | \lambda_0 ) + \theta \bm{\Psi} ( \betab_g | \lambda_1 ) \right],
\end{equation}
where $\bm{\Psi}( \cdot | \lambda)$ denotes the group lasso density (\ref{grouplassoprior}) indexed by hyperparameter $\lambda$, and $\theta \in (0, 1)$ is a mixing proportion. $\lambda_0$ corresponds to the spike which shrinks the entire vector $\bm{\beta}_g$ towards $\bm{0}_{m_g}$, while $\lambda_1$ corresponds to the slab. For shorthand notation, we denote $\bm{\Psi} ( \betab_g | \lambda_0 )$ as $\bm{\Psi}_0 (\betab_g)$ and $\bm{\Psi} ( \betab_g | \lambda_1 )$  as $\bm{\Psi}_1 (\betab_g)$ going forward.

Under the grouped regression model (\ref{groupedregression}), we place the SSGL prior (\ref{ssgrouplasso}) on $\betab$. In accordance with the recommendations of \cite{MoranRockovaGeorge2018}, we do not scale our prior by the unknown $\sigma$. Instead, we place an independent Jeffreys prior on $\sigma^2$, i.e. 
\begin{equation} \label{jeffreys}
\pi(\sigma^2) \propto \sigma^{-2}.
\end{equation}
The mixing proportion $\theta$ in (\ref{ssgrouplasso}) can either be fixed deterministically or endowed with a prior $\theta \sim \pi(\theta)$. We will discuss this in detail in Section \ref{optimizationSSGL}.

\section{Characterization and Computation of the Global Posterior Mode}\label{optimizationSSGL}

Throughout this section, we let $p$ denote the total number of covariates, i.e. $p = \sum_{g=1}^{G} m_g$. Our goal is to find the maximum \emph{a posteriori} estimates of the regression coefficients $\betab \in \mathbb{R}^p$.  This optimization problem is equivalent to a penalized likelihood method in which the logarithm of the prior \eqref{ssgrouplasso} may be reinterpreted as a penalty on the regression coefficients.  Similarly to \citet{RockovaGeorge2018}, we will leverage this connection between the Bayesian and frequentist paradigms and introduce the SSGL penalty.  This strategy combines the adaptivity of the Bayesian approach with the computational efficiency of existing algorithms in the frequentist literature. 

A key component of the SSGL model is  $\theta$, the prior expected proportion of groups with large coefficients. Ultimately, we will pursue a fully Bayes approach and place a prior on $\theta$, allowing the SSGL to adapt to the underlying sparsity of the data and perform an automatic multiplicity adjustment \cite{SB10}.  For ease of exposition, however, we will first consider the case where $\theta$ is fixed, echoing the development of \citet{RockovaGeorge2018}. In this situation, the regression coefficients $\betab_g$ are conditionally independent \emph{a priori}, resulting in a separable SSGL penalty.  Later we will consider the fully Bayes approach, which will yield the \emph{non-separable} SSGL penalty. 

\begin{definition}
	Given $\theta \in (0, 1)$, the separable SSGL penalty is defined as
	\begin{align} \label{penSbetag}
	pen_S(\betab|\theta) &=  \log\left[\frac{\pi(\betab|\theta)}{\pi(\zerob_p|\theta)}\right] = -\lambda_1\sum_{g =1}^G \lVert \betab_g\rVert_2 + \sum_{g =1}^G \log\left[\frac{p^*_{\theta}(\zerob_{m_g})}{p^*_{\theta}(\betab_g)}\right] \numbereqn
	\end{align}
	where 
	\begin{align}
	p_{\theta}^*(\betab_g) = \frac{\theta \bm{\Psi}_1(\betab_g)}{\theta\bm{\Psi}_1(\betab_g) + (1-\theta)\bm{\Psi}_0(\betab_g)}.
	\end{align}
\end{definition}
The separable SSGL penalty is almost the logarithm of the original prior \eqref{ssgrouplasso}; the only modification is an additive constant to ensure that $pen_S(\zerob_p|\theta) = 0$. The connection between the SSGL and penalized likelihood methods is made clearer when considering the derivative of the separable SSGL penalty, given in the following lemma.
\begin{lemma} \label{derivativeseparableSSGL}
	The derivative of the separable SSGL penalty satisfies
	\begin{align}
	\frac{\partial pen_S(\betab|\theta)}{\partial \lVert \betab_g\rVert_2} = -\lambda_{\theta}^*(\betab_g)
	\end{align}
	where 
	\begin{align}
	\lambda^*_{\theta}(\betab_g) = \lambda_1p_{\theta}^*(\betab_g) + \lambda_0[1-p_{\theta}^*(\betab_g)].
	\end{align}
\end{lemma}
Similarly to the SSL, the SSGL penalty is a weighted average of the two regularization parameters, $\lambda_1$ and $\lambda_0$. The weight $p^*_{\theta}(\betab_g)$ is the conditional probability that $\betab_g$ was drawn from the slab distribution rather than the spike. Hence, the SSGL features an adaptive regularization parameter which applies different amounts of shrinkage to each group, unlike the group lasso which applies the same shrinkage to each group. 

\subsection{The Global Posterior Mode} \label{globalmodetheorems}

Similarly to the group lasso \citep{YuanLin2006}, the separable nature of the penalty \eqref{penSbetag} lends itself naturally to a block coordinate ascent algorithm which cycles through the groups.  In this section, we first outline the group updates resulting from the Karush-Kuhn-Tucker (KKT) conditions. The KKT conditions provide necessary conditions for the global posterior mode. We then derive a more refined condition for the global mode to aid in optimization for multimodal posteriors. 

Following \citet{HBM12}, we assume that within each group, covariates are orthonormal, i.e. $\Xb_g^T\Xb_g = n\Ib_{m_g}$ for $g= 1,\dots, G$. If this assumption does not hold, then the $\Xb_g$ matrices can be orthonormalized before fitting the model. As noted by \citet{BrehenyHuang2015}, orthonormalization can be done without loss of generality since the resulting solution can be transformed back to the original scale.

\begin{proposition}
	The necessary conditions for $\widehat{\betab} = (\widehat{\betab}_1^T, \dots, \widehat{\betab}_G^T)^T$ to be a global mode are:
	\begin{align}
	\Xb^T_g(\Yb - \Xb\widehat{\betab}) = \sigma^2\lambda_{\theta}^*(\widehat{\betab}_g)\frac{\widehat{\betab}_g}{\lVert \betab_g\rVert_2} \quad &\text{for} \quad \widehat{\betab}_g \neq \zerob_{m_g},\\
	\lVert \Xb_g^T(\Yb - \Xb\widehat{\betab}) \rVert_2 \leq \sigma^2\lambda_{\theta}^*(\widehat{\betab}_g) \quad &\text{for}\quad \widehat{\betab}_g = \zerob_{m_g}.
	\end{align}
	Equivalently, 
	\begin{align}
	\widehat{\betab}_g = \frac{1}{n}\left(1-\frac{\sigma^2\lambda_{\theta}^*(\widehat{\betab}_g)}{\lVert \zb_g\rVert_2}\right)_+\zb_g \label{soft_thresh}
	\end{align}
	where $\zb_g = \Xb_g^T\left[\Yb - \sum_{l\neq g}\Xb_l\widehat{\betab}_l\right].$
\end{proposition} 
\begin{proof}
	Follows immediately from Lemma \ref{derivativeseparableSSGL} and subdifferential Calculus.
\end{proof}
The above characterization for the global mode is necessary, but not sufficient.  A more refined characterization may be obtained by considering the group-wise optimization problem, noting that the global mode is also a maximizer of the $g$th group, keeping all other groups fixed. 

\begin{proposition} \label{globalmodeseparable}
	The global mode $\widehat{\betab}_g = \zerob_{m_g}$ if and only if $\lVert \zb_g \rVert_2 \leq \Delta$, where
	\begin{align}
	\Delta = \inf_{\betab_g} \left\{ \frac{n\lVert \betab_g \rVert_2}{2} - \frac{\sigma^2pen_S(\betab|\theta)}{\lVert \betab_g \rVert_2} \right\}.
	\end{align}
\end{proposition}
The proof for Proposition \ref{globalmodeseparable} can be found in Appendix \ref{App:D2}. Unfortunately, the threshold $\Delta$ is difficult to compute. We instead find an approximation to this threshold. An upper bound is simply that of the soft-threshold solution \eqref{soft_thresh}, with $\Delta \leq \sigma^2\lambda^*(\betab_g)$. However, when $\lambda_0$ is large, this bound may be improved. Similarly to \citet{RockovaGeorge2018}, we provide improved bounds on the threshold in Theorem \ref{delta_bounds}. This result requires the  function  $h: \R^{m_g} \to \R$,  defined as:
\begin{align*}
h(\betab_g) = [\lambda_{\theta}^*(\betab_g) - \lambda_1]^2 + \frac{2n}{\sigma^2}\log p_{\theta}^*(\betab_g).
\end{align*}

\begin{theorem}\label{delta_bounds}
	When  $(\lambda_0 - \lambda_1) > 2\sqrt{n}/\sigma$ and $h(\zerob_{m_g})> 0$, the threshold $\Delta$ is bounded by:
	\begin{align}
	\Delta^L < \Delta < \Delta^U
	\end{align}
	where 
	\begin{align}
	\Delta^L &= \sqrt{2n \sigma^2\log[1/p_{\theta}^*(\zerob_{m_g})] - \sigma^4d} + \sigma^2\lambda_1, \\
	\Delta^U &= \sqrt{2n \sigma^2\log[1/p_{\theta}^*(\zerob_{m_g})] } + \sigma^2\lambda_1,\label{delta_u}
	\end{align}
	and
	\begin{align}
	0 < d < \frac{2n}{\sigma^2} - \left(\frac{n}{\sigma^2(\lambda_0-\lambda_1)} - \frac{\sqrt{2n}}{\sigma}\right)^2
	\end{align}
\end{theorem}
When $\lambda_0$ is large, $d\to 0$ and the lower bound on the threshold approaches the upper bound, yielding the approximation $\Delta = \Delta^U$. 
We will ultimately use this approximation in our block coordinate ascent algorithm.

\subsection{The Non-Separable SSGL penalty}\label{NS_SSGL}

As discussed earlier, a key reason for adopting a  Bayesian strategy is that it allows the model to borrow information across groups and self-adapt to the true underlying sparsity in the data. This is achieved by placing a prior on $\theta$, the proportion of groups with non-zero coefficients. We now outline this fully Bayes strategy and the resulting \emph{non-separable} SSGL penalty. With the inclusion of the prior $\theta \sim \pi(\theta)$,  the marginal prior for the regression coefficients has the following form:
\begin{align}
\pi(\betab) &= \int_0^1 \prod_{g=1}^G[\theta \bm{\Psi}_1(\betab_g) + (1-\theta) \bm{\Psi}_0(\betab_g)] d\pi(\theta) \\
&= \left( \prod_{g=1}^{G} C_g \lambda_1^{m_g} \right) e^{-\lambda_1\sum_{g=1}^G\lVert\betab_g\rVert_2} \int_0^1 \frac{\theta^G}{\prod_{g=1}^G p_{\theta}^*(\betab_g)} d\pi(\theta), \label{marginal_ns}
\end{align}
The non-separable SSGL penalty is then defined similarly to the separable penalty, where again we have centered the penalty to ensure $pen_{NS}(\zerob_p) = 0$. 
\begin{definition}
	The non-separable SSGL (NS-SSGL) penalty with $\theta \sim \pi(\theta)$ is defined as
	\begin{align} \label{penNSbeta}
	pen_{NS}(\betab) &= \log \left[ \frac{\pi(\betab)}{\pi(\zerob_p)}\right] =  -\lambda_1\sum_{g=1}^G \lVert \betab_g\rVert_2 + \log\left[ \frac{\int_0^1 \theta^G/\prod_{g=1}^G p_{\theta}^*(\betab_g) d\pi(\theta)}{\int_0^1 \theta^G/\prod_{g=1}^G p_{\theta}^*(\zerob_{m_g}) d\pi(\theta)}\right]. \numbereqn
	\end{align}
\end{definition}
Although the penalty \eqref{marginal_ns} appears intractable, intuition is again obtained by considering the derivative. Following the same line of argument as \citet{RockovaGeorge2018}, the derivative of \eqref{marginal_ns} is given in the following lemma.
\begin{lemma}
	
	\begin{align}
	\frac{\partial pen_{NS}(\betab)}{\partial \lVert \betab_g \rVert_2} \equiv \lambda^*(\betab_g; \betab_{\backslash g}),
	\end{align}
	where 
	\begin{align}
	\lambda^*(\betab_g; \betab_{\backslash g}) = p^*(\betab_g;\betab_{\backslash g})\lambda_1 + [1-p^*(\betab_g;\betab_{\backslash g})]\lambda_0
	\end{align}
	and
	\begin{align}
	p^*(\betab_g;\betab_{\backslash g}) \equiv p^*_{\theta_g}(\betab_g), \quad \text{with} \quad \theta_g = \mathbb{E}[\theta|\betab_{\backslash g}].
	\end{align}
\end{lemma}
That is, the marginal prior from \eqref{marginal_ns} is rendered tractable by considering each group of regression coefficients separately, conditional on the remaining coefficients. Such a conditional strategy is motivated by the group-wise updates for the separable penalty considered in the previous section.  Thus, our optimization strategy for the non-separable penalty will be very similar to the separable case, except instead of a fixed value for $\theta$, we will impute the mean of $\theta$ conditioned on the remaining regression coefficients. 

We now consider the form of the conditional mean, $\mathbb{E}[\theta|\widehat{\betab}_{\backslash g}]$. As noted by \citet{RockovaGeorge2018}, when the number of groups is large, this conditional mean can be replaced by $\mathbb{E}[\theta|\widehat{\betab}]$; we will proceed with the same approximation. For the prior on $\theta$, we will use the standard beta prior $\theta\sim \mathcal{B}(a, b)$.  With the choices $a = 1$ and $b = G$ for these hyperparameters, this prior results in an automatic multiplicity adjustment for the regression coefficients \cite{SB10}. 

We now examine the conditional distribution $\pi(\theta|\widehat{\betab})$. Suppose that the number of groups with non-zero coefficients is $\widehat{q}$, and assume without loss of generality that the first $\widehat{q}$ groups have non-zero coefficients. Then,
\begin{align}
\pi(\theta|\widehat{\betab}) \propto \theta^{a-1}(1-\theta)^{b-1}(1-\theta z)^{G-\widehat{q}}\prod_{g=1}^{\widehat{q}}(1-\theta x_g), \label{theta_posterior}
\end{align}
with $z = 1-\frac{\lambda_1}{\lambda_0}$ and $x_g = (1-\frac{\lambda_1}{\lambda_0}e^{\lVert \widehat{\betab}_g\rVert_2(\lambda_0-\lambda_1)})$.  Similarly to \citet{RockovaGeorge2018}, this distribution is a generalization of the Gauss hypergeometric distribution. Consequently, the expectation may be written as
\begin{align}
\mathbb{E}[\theta |\widehat{\betab}] = \frac{\int_0^1 \theta^a (1-\theta)^{b-1}(1-\theta z)^{G-\widehat{q} }\prod_{g=1}^{\widehat{q}} (1-\theta x_g)d\theta}{ \int_0^1 \theta^{a-1} (1-\theta)^{b-1}(1-\theta z)^{G-\widehat{q} }\prod_{g=1}^{\widehat{q}} (1-\theta x_g) d\theta}. \label{theta_expectation}
\end{align}
While the above expression \eqref{theta_expectation} appears laborious to compute, it admits a much simpler form when $\lambda_0$ is very large. Using a slight modification to the arguments of \cite{RG16_abel}, we obtain this simpler form in Lemma \ref{theta_mean_lemma}.

\begin{lemma} \label{condexpectationlemma}
	Assume $\pi(\theta|\widehat{\betab})$ is distributed according to \eqref{theta_posterior}. Let $\widehat{q}$ be the number of groups with non-zero coefficients. Then as $\lambda_0 \to \infty$,
	\begin{align}
	\mathbb{E}[\theta|\widehat{\betab}] = \frac{a + \widehat{q} }{a + b + G}.\label{theta_mean}
	\end{align}
	\label{theta_mean_lemma}
\end{lemma}
The proof for Lemma \ref{condexpectationlemma} is in Appendix \ref{App:D2}. We note that the expression \eqref{theta_mean} is essentially the usual posterior mean of $\theta$ under a beta prior. Intuitively, as $\lambda_0$ diverges, the weights $p_{\theta}^*(\betab_g)$ concentrate at zero and one, yielding the familiar form for $\mathbb{E}[\theta|\widehat{\betab}]$.  With this in hand, we are now in a position to outline the block coordinate ascent algorithm for the non-separable SSGL. 

\subsection{Optimization} \label{optimizationalgorithm}

The KKT conditions for the non-separable SSGL penalty yield the following necessary condition for the global mode:
\begin{align}
\widehat{\betab}_g \leftarrow \frac{1}{n} \left(1-\frac{\sigma^2\lambda_{\widehat{\theta}}^*(\widehat{\betab}_g)}{\lVert \zb_g\rVert_2}\right)_+\zb_g, \label{soft_thresh_ns}
\end{align}
where $\zb_g = \Xb_g^T\left[\Yb - \sum_{l\neq g}\Xb_l\widehat{\betab}_l\right]$ and $\widehat{\theta}$ is the mean \eqref{theta_mean}, conditioned on the previous value of $\betab$.  As before, \eqref{soft_thresh_ns} is sufficient for a local mode, but not the global mode.  When $p \gg n$ and $\lambda_0$ is large,  the posterior will be highly multimodal. As in the separable case, we require a refined thresholding scheme that will eliminate some of these suboptimal local modes from consideration. In approximating the group-wise conditional mean $\mathbb{E}[\theta|\widehat{\betab}_{\backslash g}] $ with $\mathbb{E}[\theta|\widehat{\betab}]$, we do not require group-specific thresholds. Instead, we can use the threshold given in Proposition \ref{globalmodeseparable} and Theorem \ref{delta_bounds} where $\theta$ is replaced with the current update \eqref{theta_mean}.  In particular, we shall use the upper bound $\Delta^U$ in our block coordinate ascent algorithm.

Similarly to \citet{RockovaGeorge2018}, we combine the refined threshold, $\Delta^U$ with the soft thresholding operation \eqref{soft_thresh_ns}, to yield the following update for $\widehat{\betab}_g$ at iteration $k$:
\begin{align}
{\betab}_{g}^{(k)} \leftarrow  \frac{1}{n}\left(1-\frac{\sigma^{2(k)} \lambda^*({\betab}_{g}^{(k - 1)}; {\theta}^{(k)} )}{\lVert \zb_g\rVert_2}\right)_+\zb_g \ \mathbb{I}(\lVert \zb_g\rVert_2 > \Delta^U) 
\end{align}
where $\theta^{(k)}= \mathbb{E}[\theta|\betab^{(k-1)}]$. Technically, $\theta$ should be updated after each group $\betab_g$ is updated. In practice, however, there will be little change after one group is updated and so we will update both $\theta$ and $\Delta^U$ after every $M$ iterations with a default value of $M = 10$.

With the Jeffreys prior $\pi(\sigma^2) \propto \sigma^{-2}$, the error variance $\sigma^2$ also has a closed form update:
\begin{align}
\sigma^{2(k)} \leftarrow \frac{\lVert \Yb - \Xb\betab^{(k-1)}\rVert_2^2}{n + 2}.
\end{align}
The complete optimization algorithm is given in Algorithm \ref{algorithm} of Appendix \ref{completealgorithm}. The computational complexity of this algorithm is $\mathcal{O}(np)$ per iteration, where $p = \sum_{g=1}^{G} m_g$. It takes $\mathcal{O}(n m_g)$ operations to compute the partial residual $\bm{z}_g$ for the $g$th group, for a total cost of $\mathcal{O} (n \sum_{g=1}^{G} m_g) = \mathcal{O}(np)$. Similarly, it takes $\mathcal{O}(np)$ cost to compute the sum of squared residuals $\lVert \bm{Y} - \bm{X} \widehat{\bm{\beta}} \rVert_2^2$ to update the variance parameter $\sigma^2$. The computational complexity of our algorithm matches that of the usual gradient descent algorithms for lasso and group lasso \citep{FriedmanHastieTibshirani2010}.

As a non-convex method, it is not guaranteed that SSGL will find the global posterior mode, only a local mode. However, the refined thresholding scheme (Theorem \ref{delta_bounds}) and a warm start initialization strategy (described in detail in Appendix \ref{AddlComputationalDetails}) enable SSGL to eliminate a number sub-optimal local modes from consideration in a similar manner to \citet{RockovaGeorge2018}. To briefly summarize the initialization strategy, we tune $\lambda_0$ from an increasing sequence of values, and we further scale $\lambda_0$ by $\sqrt{m_g}$ for each $g$th group to ensure that the amount of penalization is on the same scale for groups of potentially different sizes \citep{HBM12}. Meanwhile, we keep $\lambda_1$ fixed at a small value so that selected groups have minimal shrinkage. See Appendix \ref{AddlComputationalDetails} for detailed discussion of choosing $(\lambda_0, \lambda_1)$.

\section{Approaches to Inference} \label{sec:inference}

While the above procedure allows us to find the posterior mode of $\boldsymbol{\beta}$, providing a measure of uncertainty around our estimate is a challenging task. One possible solution is to run MCMC where the algorithm is initialized at the posterior mode. By starting the MCMC chain at the mode, the algorithm should converge faster. However, this is still not ideal, as it can be computationally burdensome in high dimensions. Instead, we will adopt ideas from a recent line of research (\cite{van2014asymptotically, javanmard2018debiasing}) based on de-biasing estimates from high-dimensional regression. These ideas were derived in the context of lasso regression, and we will explore the extent to which they work for the SSGL penalty. Define $\widehat{\boldsymbol{\Sigma}} = \boldsymbol{X}^T \boldsymbol{X}/n$ and let $\widehat{\boldsymbol{\Theta}}$ be an approximate inverse of $\widehat{\boldsymbol{\Sigma}}$. We define
\begin{equation}
\widehat{\boldsymbol{\beta}}_d = \widehat{\boldsymbol{\beta}} + \widehat{\boldsymbol{\Theta}} \boldsymbol{X}^T (\boldsymbol{Y} - \boldsymbol{X} \widehat{\boldsymbol{\beta}})/n.
\end{equation}
where $\widehat{\bm{\beta}}$ is the MAP estimator of $\bm{\beta}$ under the SSGL model. By \cite{van2014asymptotically}, this quantity $\widehat{\bm{\beta}}_d$ has the following asymptotic distribution: 
\begin{equation} \label{asymptoticdist}
\sqrt{n}(\widehat{\boldsymbol{\beta}}_d - \boldsymbol{\beta}) \sim \mathcal{N}(\boldsymbol{0}, \sigma^2 \widehat{\boldsymbol{\Theta}} \widehat{\boldsymbol{\Sigma}} \widehat{\boldsymbol{\Theta}}^T).
\end{equation}
For our inference procedure, we replace the population variance $\sigma^2$ in (\ref{asymptoticdist}) with the modal estimate $\widehat{\sigma}^2$ from the SSGL model. To estimate $\widehat{\boldsymbol{\Theta}}$, we utilize the nodewise regression approach developed in \cite{meinshausen2006high, van2014asymptotically}. We describe this estimation procedure for $\widehat{\boldsymbol{\Theta}}$ in Appendix \ref{ThetaEstimateDebiasing}.

Let $\widehat{\beta}_{dj}$ denote the $j$th coordinate of $\widehat{\betab}_d$. We have from (\ref{asymptoticdist}) that the $100(1-\alpha)  \%$ asymptotic pointwise confidence intervals for $\beta_{j}, j = 1, \ldots, p$, are
\begin{align} \label{confidenceintervals}
[ \widehat{\beta}_{dj} - c(\alpha, n, \widehat{\sigma}^2), \widehat{\beta}_{dj} + c(\alpha, n, \widehat{\sigma}^2) ],
\end{align}
where $c(\alpha, n, \widehat{\sigma}^2) := \Phi^{-1} (1-\alpha/2) \sqrt{ \widehat{\sigma}^2 ( \widehat{\boldsymbol{\Theta}} \widehat{\boldsymbol{\Sigma}} \widehat{\boldsymbol{\Theta}}^T )_{jj} / n}$ and $\Phi(\cdot)$ denotes the cdf of $\mathcal{N}(0,1)$. It should be noted that our posterior mode estimates should have less bias than existing estimates such as the group lasso. Therefore, the goal of the de-biasing procedure is less about de-biasing the posterior mode estimates, and more about providing an estimator with an asymptotic normal distribution from which we can perform inference.  

To assess the ability of this procedure to obtain accurate confidence intervals (\ref{confidenceintervals}) with $\alpha=0.05$, we run a small simulation study with $n=100$, $G=100$ or $n=300, G=300$, and each of the $G$ groups having $m=2$ covariates. We generate the covariates from a multivariate normal distribution with mean $\boldsymbol{0}$ and an AR(1) covariance structure with correlation $\rho$. The two covariates from each group are the linear and squared term from the original covariates. We set the first seven elements of $\boldsymbol{\beta}$ equal to $(0, 0.5, 0.25, 0.1, 0, 0, 0.7)$ and the remaining elements equal to zero. Lastly, we try $\rho = 0$ and $\rho = 0.7$. Table \ref{tab:debiasing} shows the coverage probabilities across 1000 simulations for all scenarios looked at. We see that important covariates, i.e. covariates with a nonzero corresponding $\beta_j$, have coverage near 0.85 when $n=100$ under either correlation structure, though this increases to nearly the nominal rate when $n=300$. The remaining covariates (null covariates) achieve the nominal level regardless of the sample size or correlation present. 

\begin{table}[t]
	\centering
	\begin{tabular}{|l|rrr|}
		\hline
		& $\rho$ & Important covariates & Null covariates\\ 
		\hline
		$n=100, G=100$ & 0.0 & 0.83 & 0.93 \\ 
		&  0.7 & 0.85 & 0.94 \\ 
		\hline
		$n=300, G = 300$ & 0.0 & 0.93 & 0.95 \\ 
		&  0.7 & 0.92 & 0.95 \\ 
		\hline
	\end{tabular}
	\caption{Coverage probabilities for de-biasing simulation. }
	\label{tab:debiasing}
\end{table}

\section{Nonparametric Spike-and-Slab Lasso} \label{NPSSLIntro}

We now introduce the nonparametric spike-and-slab lasso (NPSSL).  The NPSSL allows for flexible modeling of a response surface with minimal assumptions regarding its functional form. We consider two cases for the NPSSL: (i) a main effects only model, and (ii) a model with both main and interaction effects.

\subsection{Main Effects} \label{NPSSLMainEffects}

We first consider the main effects NPSSL model.  Here, we assume that the response surface may be decomposed into the sum of univariate functions of each of the $p$ covariates. That is, we have the following model:
\begin{align}
y_i = \sum_{j=1}^p f_j(X_{ij}) + \varepsilon_i, \quad \varepsilon_i \sim \mathcal{N}(0, \sigma^2). \label{NPSSL_main_effects}
\end{align}
Following \citet{RavikumarLaffertyLiuWasserman2009}, we assume that each $f_j$, $j = 1,\dots, p$, may be approximated by a linear combination of basis functions $\mathcal{B}_j = \{g_{j1}, \dots, g_{jd}\}$, i.e.,
\begin{align}
f_j(X_{ij}) \approx \sum_{k = 1}^{d} g_{jk}(X_{ij}) \beta_{jk} 
\end{align}
where $\betab_j = (\beta_{j1}, \dots, \beta_{jd})^T$ are the unknown weights.  Let $\widetilde{\Xb}_j$ denote the $n\times d$ matrix with the $(i, k)$th entry $\widetilde{\Xb}_j(i, k) = g_{jk}(X_{ij})$. Then, \eqref{NPSSL_main_effects} may be represented in matrix form as
\begin{align}
\yb - \bm{\delta} = \sum_{j=1}^p\widetilde{\Xb}_j\betab_j + \epsilonb, \quad \epsilonb\sim \mathcal{N}_n(\zerob, \sigma^2 \bm{I}_n), \label{matrix_main_effects}
\end{align}
where $\bm{\delta}$ is a vector of the lower-order truncation bias. Note that we assume the response $\yb$ has been centered and so we do not include a grand mean $\bm{\mu}$ in (\ref{matrix_main_effects}). Thus, we do not require the main effects to integrate to zero as in \cite{WeiReichHoppinGhosal2018}. We do, however, require the matrices $\widetilde{\Xb}_j, j = 1, \ldots, p$, to be orthogonal, as discussed in Section \ref{optimizationSSGL}.  Note that the entire design matrix does not need to be orthogonal; only the group-specific matrices need to be. We can enforce this in practice by either using orthonormal basis functions or by orthornormalizing the $\widetilde{\Xb}_j$ matrices before fitting the model. 

We assume that $\yb$ depends on only a small number of the $p$ covariates so that many of the $f_j$'s have a negligible contribution to \eqref{NPSSL_main_effects}. This is equivalent to assuming that most of the weight vectors $\betab_j$ have all zero elements.  If the $j$th covariate is determined to be predictive of $\yb$, then $f_j$ has a non-negligible contribution to \eqref{NPSSL_main_effects}. In this case, we want to include the \textit{entire} basis function approximation to $f_j$ in the model.

The above situation is a natural fit for the SSGL. We have $p$ groups where each group is either included as a whole or not included in the model. The design matrices for each group are exactly the matrices of basis functions, $\widetilde{\Xb}_j, j=1, \ldots, p$. We will utilize the non-separable SSGL penalty developed in Section \ref{NS_SSGL} to enforce this group-sparsity behavior in the model \eqref{matrix_main_effects}. More specifically, we seek to maximize the objective function with respect to $\betab = (\betab_1^T,\dots, \betab_p^T)^T \in \R^{pd}$ and $\sigma^2$:
\begin{align}
L(\betab, \sigma^2) = -\frac{1}{2\sigma^2}\lVert \yb - \sum_{j=1}^p \widetilde{\Xb}_j\betab_j\rVert_2^2 - (n+2)\log\sigma+ pen_{NS}(\betab).\label{NPSSL_main_effects_obj}
\end{align}
To find the estimators of $\betab$ and $\sigma^2$, we use Algorithm \ref{algorithm} in Appendix \ref{completealgorithm}. Similar additive models have been proposed by a number of authors including \citet{RavikumarLaffertyLiuWasserman2009} and \citet{WeiReichHoppinGhosal2018}. However, our proposed NPSSL method has a number of advantages. First, we allow the noise variance $\sigma^2$ to be unknown, unlike \citet{RavikumarLaffertyLiuWasserman2009}. Accurate estimates of $\sigma^2$ are important to avoid overfitting the noise beyond the signal. Secondly, we use a block-descent algorithm to quickly target the modes of the posterior, whereas \citet{WeiReichHoppinGhosal2018} utilize MCMC. Finally, our SSGL algorithm automatically thresholds negligible groups to zero, negating the need for a post-processing thresholding step.

\subsection{Main and Interaction Effects} \label{NPSSLMainInteractionEffects}

The main effects model \eqref{NPSSL_main_effects} allows for each covariate to have a nonlinear contribution to the model, but assumes a linear relationship \emph{between} the covariates. In some applications, this assumption may be too restrictive. For example, in the environmental exposures data which we analyze in Section \ref{NHANES}, we may expect high levels of two toxins to have an even more adverse effect on a person's health than high levels of either of the two toxins. Such an effect may be modeled by including interaction effects between the covariates. 

Here, we extend the NPSSL to include interaction effects. We consider only second-order interactions between the covariates, but our model can easily be extended to include even higher-order interactions. We assume that the interaction effects may be decomposed into the sum of bivariate functions of each pair of covariates, yielding the model:
\begin{align}
y_i = \sum_{j=1}^p f_j(X_{ij}) + \sum_{k=1}^{p-1}\sum_{l=k+1}^p f_{kl}(X_{ik}, X_{il}) + \varepsilon_i, \quad \varepsilon_i \sim \mathcal{N}(0, \sigma^2).\label{NPSSL_interactions}
\end{align}

For the interaction terms, we follow \citet{WeiReichHoppinGhosal2018} and approximate $f_{kl}$ using the outer product of the basis functions of the interacting covariates:
\begin{align}
f_{kl}(X_{ik}, X_{il}) \approx \sum_{s = 1}^{d^*}\sum_{r=1}^{d^*} g_{ks}(X_{ik})g_{lr}(X_{il}) \beta_{klsr}
\end{align}
where $\betab_{kl} = (\beta_{kl11}, \dots, \beta_{kl1d^*}, \beta_{kl21}, \dots, \beta_{kld^*d^*})^T \in \R^{d^{*2}}$ is the vector of unknown weights. We let $\widetilde{\Xb}_{kl}$ denote the $n\times d^{*2}$ matrix with rows $$\widetilde{\Xb}_{kl}(i, \cdot) = \text{vec}(\bm{g}_k(X_{ik})\bm{g}_l(X_{il})^T),$$ where $\bm{g}_k(X_{ik}) = (g_{k1}(X_{ik}), \dots, g_{kd^*}(X_{ik}))^T$. Then, \eqref{NPSSL_interactions} may be represented in matrix form as
\begin{align} 
\yb - \bm{\delta} = \sum_{j=1}^p \widetilde{\Xb}_j\betab_j + \sum_{k=1}^{p-1}\sum_{l=k+1}^p \widetilde{\Xb}_{kl}\betab_{kl} + \epsilonb, \quad \epsilonb\sim \mathcal{N}_n(\zerob, \sigma^2\bm{I}_n),\label{matrix_NPSSL_interactions}
\end{align}
where $\bm{\delta}$ is a vector of the lower-order truncation bias. We again assume $\yb$ has been centered and so do not include a grand mean in \eqref{matrix_NPSSL_interactions}. We do not constrain   $f_{kl}$ to integrate to zero as in \citet{WeiReichHoppinGhosal2018}. However, we do ensure that the main effects are not in the linear span of the interaction functions. That is, we require the ``main effect'' matrices $\widetilde{\Xb}_l$ and  $\widetilde{\Xb}_k$ to be orthogonal to the ``interaction'' matrix $\widetilde{\Xb}_{kl}$. This condition is needed to maintain identifiability for both the main and interaction effects in the model. In practice, we enforce this condition by setting the interaction design matrix to be the residuals of the regression of $\widetilde{\Xb}_k \circ \widetilde{\Xb}_l$ on $\widetilde{\Xb}_k$ and $\widetilde{\Xb}_l$.

Note that the current representation does not enforce strong hierarchy. That is, interaction terms can be included even if their corresponding main effects are removed from the model. However, the NPSSL model can be easily modified to accommodate strong hierarchy. If hierarchy is desired, the ``interaction'' matrices can be augmented to contain both main and interaction effects, as in \citet{lim2015learning}, i.e. the ``interaction'' matrices in \eqref{matrix_NPSSL_interactions} would be $\widetilde{\bm{X}}_{kl}^{\textrm{aug}} = [ \widetilde{\bm{X}}_k, \widetilde{\bm{X}}_l, \widetilde{\bm{X}}_{kl}]$, instead of simply $\widetilde{\bm{X}}_{kl}$. This augmented model is overparameterized since the main effects still have their own separate design matrices as well (to ensure that main effects can still be selected even if $\bm{\beta}_{kl}^{\textrm{aug}} = \bm{0}$). However, this ensures that interaction effects are only selected if the corresponding main effects are also in the model.

In the interaction model, we either include $\betab_{kl}$ in the model \eqref{matrix_NPSSL_interactions} if there is a non-negligible interaction between the $k$th and $l$th covariates, or we estimate $\widehat{\betab}_{kl} = \zerob_{d^{*2}}$ if such an interaction is negligible.  With the non-separable SSGL penalty, the objective function is:
\begin{align*} \label{obj_NPSSL_interaction}
L(\betab, \sigma^2) &= -\frac{1}{2\sigma^2}\lVert \yb - \sum_{j=1}^p \widetilde{\Xb}_j\betab_j - \sum_{k=1}^{p-1} \sum_{l=k+1}^p \widetilde{\Xb}_{kl}\betab_{kl}\rVert_2^2 + pen_{NS}(\betab) \\
& \quad  -(n+2)\log\sigma, \numbereqn
\end{align*}
where $\betab = (\betab_1^T, \dots, \betab_p^T, \betab_{12}^T, \dots \betab_{(p-1)p}^T)^T \in \R^{pd + p(p-1)d^{*2}/2}.$ We can again use Algorithm \ref{algorithm} in Appendix \ref{completealgorithm} to find the modal estimates of $\betab$ and $\sigma^2$.

\section{Asymptotic Theory for the SSGL and NPSSL} \label{asymptotictheory}

In this section, we derive asymptotic properties for the separable SSGL and NPSSL models. We first note some differences between our theory and the theory in \citet{RockovaGeorge2018}. First, we prove \textit{joint} consistency in estimation of both the unknown $\bm{\beta}$ \textit{and} the unknown $\sigma^2$, whereas \cite{RockovaGeorge2018} proved their result only for $\bm{\beta}$, assuming known variance $\sigma^2 = 1$. Secondly, \citet{RockovaGeorge2018} established convergence rates for the global posterior mode and the full posterior separately, whereas we establish a contraction rate $\epsilon_n$ for the full posterior only. Our rate $\epsilon_n$ satisfies $\epsilon_n \rightarrow 0$ as $n \rightarrow \infty$ (i.e. the full posterior collapses to the true $(\bm{\beta}, \sigma^2)$ almost surely as $n \rightarrow \infty$), and hence, it automatically follows that the posterior mode is a consistent estimator of $(\bm{\beta}, \sigma^2)$. Finally, we also derive a posterior contraction rate for nonparametric additive regression, not just linear regression. All proofs for the theorems in this section can be found in Appendix \ref{App:D3}.

\subsection{Grouped Linear Regression}
We work under the frequentist assumption that there is a true model,
\begin{equation} \label{truemodel}
\yb = \displaystyle \sum_{g=1}^{G} \Xb_g \betab_{0g} + \bm{\varepsilon}, \hspace{.5cm} \bm{\varepsilon} \sim \N_n ( \zerob, \sigma_0^2 \Ib_n ),
\end{equation}
where $\betab_0 = ( \betab_{01}^T, \ldots, \betab_{0G}^T )^T$ and $\sigma_0^2 \in (0, \infty)$. Denote $\Xb = [ \Xb_1, \ldots, \Xb_G ]$ and $\betab = ( \betab_1^T, \ldots, \betab_G^T)^T$. Suppose we endow $(\betab, \sigma^2)$ under model (\ref{truemodel})  with the following prior:
\begin{equation} \label{hiermodel}
\begin{array}{rl}
\pi (\bm{\beta} | \theta ) \sim & \displaystyle \prod_{g=1}^{G} \left[ (1- \theta) \bm{\Psi} ( \betab_g | \lambda_0 ) + \theta \bm{\Psi} ( \betab_g | \lambda_1 ) \right], \\
\theta \sim & \mathcal{B}(a, b), \\
\sigma^2 \sim & \mathcal{IG} (c_0, d_0), 
\end{array}
\end{equation}
where $c_0 > 0$ and $d_0 > 0$ are fixed constants and the hyperparameters $(a,b)$ in the prior on $\theta$ are to be chosen later.

\begin{remark}
	In our implementation of the SSGL model, we endowed $\sigma^2$ with an improper prior, $\pi(\sigma^2) \propto \sigma^{-2}$. This can be viewed as a limiting case of the $\mathcal{IG}(c_0,d_0)$ prior with $c_0 \rightarrow 0, d_0 \rightarrow 0$. This improper prior is fine for implementation since it leads to a proper posterior, but for our theoretical investigation, we require the priors on $(\bm{\beta}, \sigma^2)$ to be proper. 
\end{remark}

\subsubsection{Posterior Contraction Rates}

Let $m_{\max} = \max_{1 \leq j \leq G} m_g$ and let  $p = \sum_{g=1}^{G} m_g$. Let $S_0$ be the set containing the indices of the true nonzero groups, where $S_0 \subseteq \{1, \ldots, G \}$ with cardinality $s_0 = \lvert S_0 \rvert$.  We make the following assumptions:
\begin{enumerate}[label=(A\arabic*)]
	\item Assume that $G \gg n $, $\log(G) = o(n)$, and $m_{\max} = O( \log G / \log n)$. \label{A1} 
	\item The true number of nonzero groups satisfies $s_0 = o(n / \log G)$. \label{A2}
	\item There exists a constant $k>0$ so that $\lambda_{\max} ( \Xb^T \Xb ) \leq k n^{\alpha}$, for some $\alpha \in [1, \infty)$. \label{A3}
	\item Let $\xi \subset \{1, \ldots, G \}$, and let $\bm{X}_{\xi}$ denote the submatrix of $\Xb$ that contains the submatrices with groups indexed by $\xi$. There exist constants $\nu_1 > 0$, $\nu_2 > 0$, and an integer $\bar{p}$ satisfying $s_0 = o(\bar{p})$ and $\bar{p} = o(  s_0 \log n )$, so that $n \nu_1 \leq \lambda_{\min}  ( \Xb_{\xi}^T \Xb_{\xi} ) \leq \lambda_{\max} ( \Xb_{\xi}^T \Xb_{\xi}) \leq n \nu_2$ for any model of size $\lvert \xi \rvert \leq \bar{p}$.   \label{A4}
	\item $\lVert \betab_0 \rVert_{\infty} = O( \log G ).$	\label{A5}
\end{enumerate}

Assumption \ref{A1} allows the number of groups $G$ and total number of covariates $p$ to grow at nearly exponential rate with sample size $n$. The size of each individual group may also grow as $n$ grows, but should grow at a slower rate than $n / \log n$. Assumption \ref{A2} specifies the growth rate for the true model size $s_0$. Assumption \ref{A3} bounds the eigenvalues of $\Xb^T \Xb$ from above and is less stringent than requiring all the eigenvalues of the Gram matrix ($\Xb^T \Xb / n$) to be bounded away from infinity. Assumption \ref{A4} ensures that $\Xb^T \Xb$ is locally invertible over sparse sets. In general, conditions \ref{A3}-\ref{A4} are difficult to verify, but they can be shown to hold with high probability for certain classes of matrices where the rows of $\Xb$ are independent and sub-Gaussian \cite{MendelsonPajor2006, RaskuttiWainwrightYu2010}. Finally, Assumption \ref{A5} places a restriction on the growth rate of the maximum signal size for the true $\betab_0$. 

We now state our main theorem on the posterior contraction rates for the SSGL prior (\ref{hiermodel}) under model (\ref{truemodel}). Let $\mathbb{P}_0$ denote the probability measure underlying the truth (\ref{truemodel}) and $\Pi( \cdot | \yb)$ denote the posterior distribution under the prior (\ref{hiermodel}) for $(\betab, \sigma^2)$.

\begin{theorem}[posterior contraction rates] \label{posteriorcontractiongroupedregression}
	Let $\epsilon_n =  \sqrt{ s_0 \log G / n }$, and suppose that Assumptions \ref{A1}-\ref{A5} hold.  Under model (\ref{truemodel}), suppose that we endow $(\betab, \sigma^2)$ with the prior (\ref{hiermodel}). For the hyperparameters in the $\mathcal{B}(a,b)$ prior on $\theta$, we choose $a=1, b=G^{c}$, $c > 2$. Further, we set $\lambda_0 = (1-\theta)/\theta$ and $\lambda_1 \asymp 1/n$ in the SSGL prior. Then
	\begin{equation} \label{l2contraction}
	\Pi \left( \betab: \lVert \betab - \betab_0 \rVert_2 \geq M_1 \sigma_0 \epsilon_n \vert \yb \right) \rightarrow 0 \textrm{ a.s. } \mathbb{P}_0 \textrm{ as } n, G \rightarrow \infty,
	\end{equation} 
	\begin{equation} \label{predictioncontraction}
	\Pi \left( \bm{\beta}: \lVert \Xb \betab - \Xb \betab_0 \rVert_2 \geq M_2 \sigma_0 \sqrt{n} \epsilon_n \vert \yb \right) \rightarrow 0 \textrm{ a.s. } \mathbb{P}_0 \textrm{ as } n, G \rightarrow \infty, 
	\end{equation} 
	\begin{equation} \label{varianceconsistency}
	\Pi \left( \sigma^2: \lvert \sigma^2 - \sigma_0^2 \rvert \geq 4 \sigma_0^2 \epsilon_n  \vert \yb \right) \rightarrow 0 \textrm{ as } n \rightarrow \infty, \textrm{ a.s. } \mathbb{P}_0 \textrm{ as } n, G \rightarrow \infty,
	\end{equation}
	for some $M_1 > 0, M_2 > 0$.
\end{theorem}
\begin{remark}
	In the case where $G = p$ and $m_1 = \ldots = m_G = 1$, the $\ell_2$ and prediction error rates in (\ref{l2contraction})-(\ref{predictioncontraction}) reduce to the familiar optimal rates of $\sqrt{s_0 \log p / n}$ and $\sqrt{s_0 \log p}$ respectively. 
\end{remark}
\begin{remark}
	Eq. (\ref{varianceconsistency}) demonstrates that our model also consistently estimates the unknown variance $\sigma^2$, therefore providing further theoretical justification for placing an independent prior on $\sigma^2$, as advocated by \citet{MoranRockovaGeorge2018}.
\end{remark}

\subsubsection{Dimensionality Recovery}

Although the posterior mode is exactly sparse, the SSGL prior is absolutely continuous so it assigns zero mass to exactly sparse vectors. To approximate the model size under the SSGL model, we use the following generalized notion of sparsity \citep{BhattacharyaPatiPillaiDunson2015}. For $\omega_g > 0$, we define the generalized inclusion indicator and generalized dimensionality, respectively, as
\begin{equation} \label{generalizeddimensionality}
\gamma_{\omega_g} (\betab_g) = I( \lVert \betab_g \rVert_2 > \omega_g) \textrm{ and } \lvert \bm{\gamma} (\betab) \rvert = \displaystyle \sum_{g=1}^{G} \gamma_{\omega_g} (\betab_g ).
\end{equation} 
In contrast to \cite{BhattacharyaPatiPillaiDunson2015, RockovaGeorge2018}, we allow the threshold $\omega_g$ to be different for each group, owing to the fact that the group sizes $m_g$ may not necessarily all be the same. However, the $\omega_g$'s, $g= 1, \ldots, G$, should still tend towards zero as $n$ increases, so that $| \bm{\gamma} (\bm{\beta}) |$ provides a good approximation to $\# \{g: \betab_g \neq \zerob_{m_g} \}$. 

Consider as the threshold,
\begin{equation} \label{omegathreshold}
\omega_g \equiv \omega_g(\lambda_0, \lambda_1, \theta) = \frac{1}{\lambda_0 - \lambda_1} \log \left[ \frac{1-\theta}{\theta} \frac{\lambda_0^{m_g}}{\lambda_1^{m_g}} \right]
\end{equation}
Note that for large $\lambda_0$, this threshold rapidly approaches zero. Analogous to \cite{Rockova2018, RockovaGeorge2018}, any vectors $\betab_g$ that satisfy $\lVert \betab_g \rVert_2 = \omega_g$ correspond to the intersection points between the two group lasso densities in the separable SSGL prior (\ref{ssgrouplasso}), or when the second derivative $\partial^2 pen_S(\betab|\theta) / \partial \lVert \betab_g\rVert_2^2 = 0.5$. The value $\omega_g$ represents the turning point where the slab has dominated the spike, and thus, the sharper the spike (when $\lambda_0$ is large), the smaller the threshold. 

Using the notion of generalized dimensionality (\ref{generalizeddimensionality}) with (\ref{omegathreshold}) as the threshold, we have the following theorem.

\begin{theorem}[dimensionality] \label{dimensionalitygroupedregression}
	Suppose that the same conditions as those in Theorem \ref{posteriorcontractiongroupedregression} hold. Then under \eqref{truemodel}, for sufficiently large $M_3 > 0$,
	\begin{equation} \label{posteriorcompressibility}
	\displaystyle \sup_{\betab_0} \mathbb{E}_{\betab_0} \Pi \left( \betab: \lvert \bm{\gamma} (\betab ) \rvert > M_3 s_0 \vert \yb \right) \rightarrow 0 \textrm{ as } n, G \rightarrow \infty.
	\end{equation}
\end{theorem}
Theorem \ref{dimensionalitygroupedregression} shows that the expected posterior probability that the generalized dimension is a constant multiple larger than the true model size $s_0$ is asymptotically vanishing. In other words, the SSGL posterior concentrates on sparse sets.

\subsection{Sparse Generalized Additive Models (GAMs)}
Assume there is a true model,
\begin{equation} \label{truemodelGAM}
y_i = \displaystyle \sum_{j=1}^{p} f_{0j}(X_{ij}) + \varepsilon_i, \hspace{.5cm} \varepsilon_i \sim \mathcal{N} (0, \sigma_0^2).
\end{equation}
where $\sigma_0^2 \in (0, \infty)$. Throughout this section, we assume that all the covariates $\bm{X}_i = (X_{i1}, \ldots, X_{ip})^T$ have been standardized to lie in $[0,1]^p$ and that $f_{0j} \in \mathcal{C}^{\kappa}[0,1], j=1, \ldots, p$. That is, the true functions are all at least $\kappa$-times continuously differentiable over $[0,1]$, for some $\kappa \in \mathbb{N}$. Suppose that each $f_{0j}$ can be approximated by a linear combination of basis functions, $\{g_{j1}, \ldots, g_{jd} \}$. In matrix notation, (\ref{truemodelGAM}) can then be written as
\begin{equation} \label{truemodelGAMmatrix}
\yb = \displaystyle \sum_{j=1}^{p} \widetilde{\Xb}_j \betab_{0j} + \bm{\delta} + \bm{\varepsilon}, \hspace{.5cm} \bm{\varepsilon} \sim \mathcal{N}_n ( \bm{0}, \sigma_0^2 \Ib_n),
\end{equation} 
where $\widetilde{\Xb}_j$ denotes an $n \times d$ matrix where the $(i,k)$th entry is $\widetilde{\Xb}_j(i,k) = g_{jk} (X_{ij})$, the  $\betab_{0j}$'s are $d \times 1$ vectors of basis coefficients, and $\bm{\delta}$ denotes an $n \times 1$ vector of lower-order bias.

Denote $\widetilde{\bm{X}} = [ \widetilde{\Xb}_1, \ldots, \widetilde{\Xb}_p]$ and $\bm{\beta} = ( \betab_1^T, \ldots, \betab_p^T )^T$. Under (\ref{truemodelGAM}), suppose that we endow $(\betab, \sigma^2)$ in (\ref{truemodelGAMmatrix}) with the prior (\ref{hiermodel}). We have the following assumptions:
\begin{enumerate}[label=(B\arabic*)]
	\item Assume that $p \gg n$, $\log p = o(n)$, and $d \asymp n^{1 / (2 \kappa + 1)}$. \label{B1}
	\item The number of true nonzero functions satisfies 
	\begin{align*}
	    s_0 = o( \max \{ n / \log p, n^{2 \kappa / (2 \kappa + 1)} \} ).
	\end{align*} \label{B2}
	\item There exists a constant $k_1 > 0$ so that for all $n$, $\lambda_{\max} ( \widetilde{\Xb}^T \widetilde{\Xb} ) \leq k_1 n$. \label{B3}
	\item Let $\xi \subset \{1, \ldots, p \}$, and let $\widetilde{\Xb}_{\xi}$ denote the submatrix of $\widetilde{\Xb}$ that contains the submatrices indexed by $\xi$. There exists a constant $\nu_1 > 0$ and an integer $\bar{p}$ satisfying $s_0 = o(\bar{p})$ and $ \bar{p} = o( s_0 \log n )$, so that $\lambda_{\min}  ( \widetilde{\Xb}_{\xi}^T \widetilde{\Xb}_{\xi}  ) \geq n \nu_1$ for any model of size $\lvert \xi \rvert \leq \bar{p}$.   \label{B4}
	\item $\lVert \betab_0 \rVert_{\infty} = O( \log p ).$		\label{B5}
	\item The bias $\bm{\delta}$ satisfies $\lVert \bm{\delta} \rVert_{2} \lesssim \sqrt{s_0 n} d^{- \kappa}$. \label{B6}
\end{enumerate}
\noindent Assumptions \ref{B1}-\ref{B5} are analogous to assumptions \ref{A1}-\ref{A5}. Assumptions \ref{B3}-\ref{B4} are difficult to verify but can be shown to hold if appropriate basis functions for the $g_{jk}$'s are used, e.g. cubic B-splines \cite{YooGhosal2016, WeiReichHoppinGhosal2018}. Finally, Assumption \ref{B6} bounds the approximation error incurred by truncating the basis expansions to be of size $d$. This assumption is satisfied, for example, by B-spline basis expansions \cite{ZhouShenWolfe1998, WeiReichHoppinGhosal2018}.

Let $\widetilde{\mathbb{P}}_0$ denote the probability measure underlying the truth (\ref{truemodelGAM}) and $\Pi ( \cdot | \yb)$ denote the posterior distribution under NPSSL model with the prior (\ref{hiermodel}) for $(\betab, \sigma^2)$ in (\ref{truemodelGAMmatrix}). Further, let $f (\bm{X}_i) = \sum_{j=1}^{p} f_j ( X_{ij})$ and $f_0 (\bm{X}_i) = \sum_{j=1}^{p} f_{0j} (X_{ij})$, and define the empirical norm $\lVert \cdot \rVert_n$ as 
\begin{align*}
\lVert f - f_0 \rVert_n^2 = \frac{1}{n} \sum_{i=1}^{n} \left[ f (\bm{X}_i) - f_0 (\bm{X}_i) \right]^2.
\end{align*}
Let $\mathcal{F}$ denote the infinite-dimensional set of all possible additive functions $f = \sum_{j=1}^{p} f_j$, where each $f_j$ can be represented by a $d$-dimensional basis expansion. In \citet{RaskuttiWainwrightYu2012}, it was shown that the minimax estimation rate for $f_0 = \sum_{j=1}^{p} f_{0j}$ under squared $\ell_2$ error loss is $\epsilon_n^2 \asymp s_0 \log p / n + s_0 n^{-2 \kappa / (2 \kappa + 1)}$. The next theorem establishes that the NPSSL model achieves this minimax posterior contraction rate.

\begin{theorem}[posterior contraction rates] \label{contractionGAMs}
	Let $\epsilon_n^2 = s_0 \log p / n + s_0 n^{-2 \kappa / (2 \kappa + 1)}$. Suppose that Assumptions \ref{B1}-\ref{B6} hold. Under model (\ref{truemodelGAMmatrix}), suppose that we endow $(\bm{\beta}, \sigma^2)$  with the prior (\ref{hiermodel}) (replacing $G$ with $p$). For the hyperparameters in the $\mathcal{B}(a,b)$ prior on $\theta$, we choose $a=1, b=p^{c}$, $c > 2$. Further, we set $\lambda_0 = (1-\theta)/\theta$ and $\lambda_1 \asymp 1/n$ in the SSGL prior. Then
	\begin{equation} \label{empiricalcontractionGAM}
	\Pi \left( f \in \mathcal{F}: \lVert f - f_0 \rVert_n \geq \widetilde{M}_1 \epsilon_n | \yb \right) \rightarrow 0 \textrm{ a.s. } \widetilde{\mathbb{P}}_0 \textrm{ as } n, p \rightarrow \infty,
	\end{equation}
	\begin{equation} \label{GAMvarianceconsistency}
	\Pi \left( \sigma^2: \lvert \sigma^2 - \sigma_0^2 \rvert \geq 4 \sigma_0^2 \epsilon_n  \vert \yb \right) \rightarrow 0 \textrm{ as } n \rightarrow \infty, \textrm{ a.s. } \widetilde{\mathbb{P}}_0 \textrm{ as } n, p \rightarrow \infty,
	\end{equation}
	for some $\widetilde{M}_1 > 0$.
\end{theorem}

Let the generalized dimensionality $ | \bm{\gamma} (\bm{\beta}) |$ be defined as before in  (\ref{generalizeddimensionality}) (replacing $G$ with $p$), with $\omega_g$ from (\ref{omegathreshold}) as the threshold (replacing $m_g$ with $d$). The next theorem shows that under the NPSSL, the expected posterior probability that the generalized dimension size is a constant multiple larger than the true model size $s_0$ asymptotically vanishes.
\begin{theorem}[dimensionality] \label{dimensionalityGAM}
	Suppose that the same conditions as those in Theorem \ref{contractionGAMs} hold. Then under \eqref{truemodelGAMmatrix}, for sufficiently large $\widetilde{M}_2 > 0$,
	\begin{equation} \label{posteriorcompressibilityGAM}
	\displaystyle \sup_{\bm{\beta}_0} \widetilde{\mathbb{E}}_{\bm{\beta}_0} \Pi \left( \bm{\beta}: \lvert \bm{\gamma} (\bm{\beta} ) \rvert > \widetilde{M}_2 s_0 \vert \yb \right) \rightarrow 0 \textrm{ as } n, p \rightarrow \infty.
	\end{equation}
\end{theorem}

\section{Simulation Studies} \label{Simulations}

In this section, we will evaluate our method in a number of settings. For the SSGL approach, we fix $\lambda_1 = 1$ and use cross-validation to choose from $\lambda_0 \in \{1, 2, \ldots, 100 \}$. For the prior $\theta \sim \mathcal{B}(a,b)$, we set $a=1, b=G$ so that $\theta$ is small with high probability.  We will compare our SSGL approach with the following methods:

\begin{enumerate}
	\item GroupLasso: the group lasso \citep{YuanLin2006}
	\item BSGS: Bayesian sparse group selection \citep{chen2016bayesian}
	\item SoftBart: soft Bayesian additive regression tree (BART) \citep{linero2018bayesian}
	\item RandomForest: random forests \citep{breiman2001random}
	\item SuperLearner: super learner \citep{van2007super}
	\item GroupSpike: point-mass spike-and-slab priors \eqref{pointmassspikeandslab} placed on groups of coefficients\footnote{Code to implement GroupSpike is included in the Supplementary data. Due to the discontinuous prior, GroupSpike is not amenable to a MAP finding algorithm and has to be implemented using MCMC.}
\end{enumerate}
In our simulations, we will look at the mean squared error (MSE) for estimating $f(\boldsymbol{X}_{\textrm{new}})$ averaged over a new sample of data $\boldsymbol{X}_{\textrm{new}}$. We will also evaluate the variable selection properties of the different methods using precision and recall, where $\text{precision} = \text{TP}/(\text{TP}+\text{FP})$, $\text{recall} = \text{TP}/(\text{TP}+\text{FN})$, and TP, FP, and FN denote the number of true positives, false positives, and false negatives respectively. Note that we will not show precision or recall for the SuperLearner, which averages over different models and different variable selection procedures and therefore does not have one set of variables that are deemed significant.

\subsection{Sparse Semiparametric Regression} \label{simsparse}

Here, we will evaluate the use of our proposed SSGL procedure in sparse semiparametric regression with $p$ continuous covariates. Namely, we implement the NPSSL main effects model described in Section \ref{NPSSLMainEffects}. In Appendix \ref{App:B}, we include more simulation studies of the SSGL approach under both sparse and dense settings, as well as a simulation study showing that we are accurately estimating the residual variance $\sigma^2$.

We let $n=100, p=300$. We generate independent covariates from a standard uniform distribution, and we let the true regression surface take the following form:
\begin{align*}
\mathbb{E} (Y \vert \boldsymbol{X}) = 5 \text{sin}(\pi X_1) + 2.5 (X_3^2 - 0.5) + e^{X_4} + 3 X_5,
\end{align*}
with variance $\sigma^2 = 1$.  

\begin{figure}[t!]
	\centering
	\includegraphics[width=0.32\linewidth]{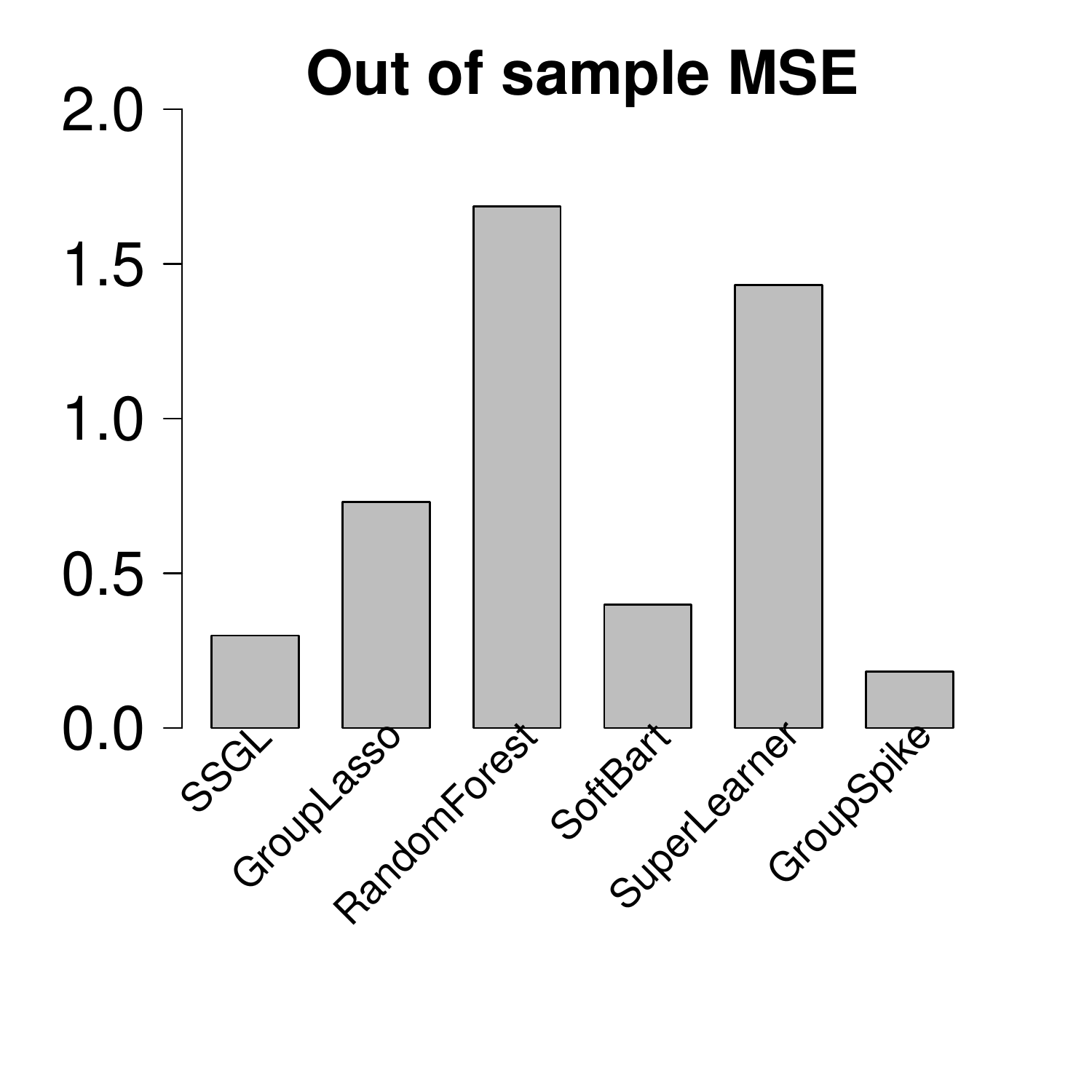}
	\includegraphics[width=0.32\linewidth]{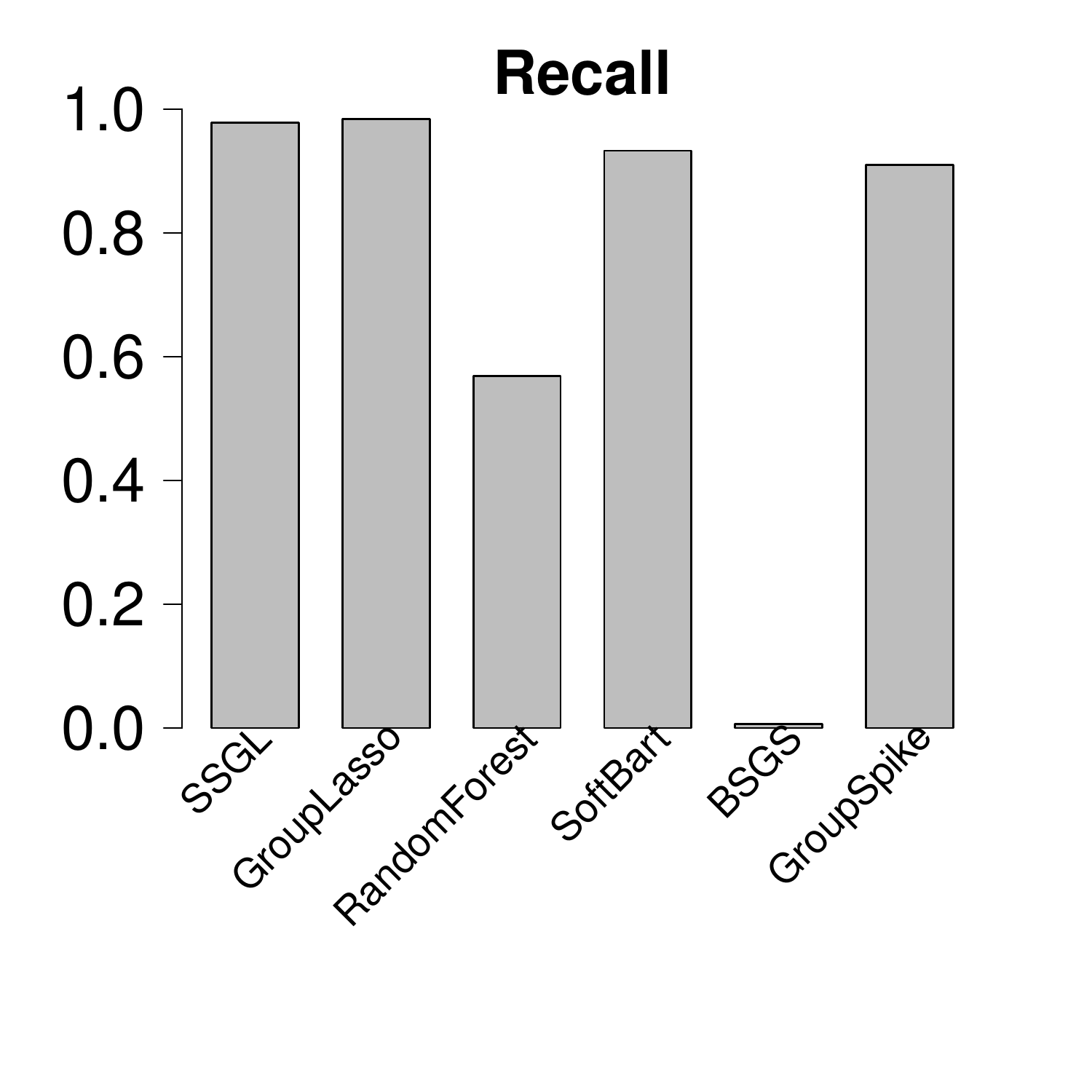} \\
	\includegraphics[width=0.32\linewidth]{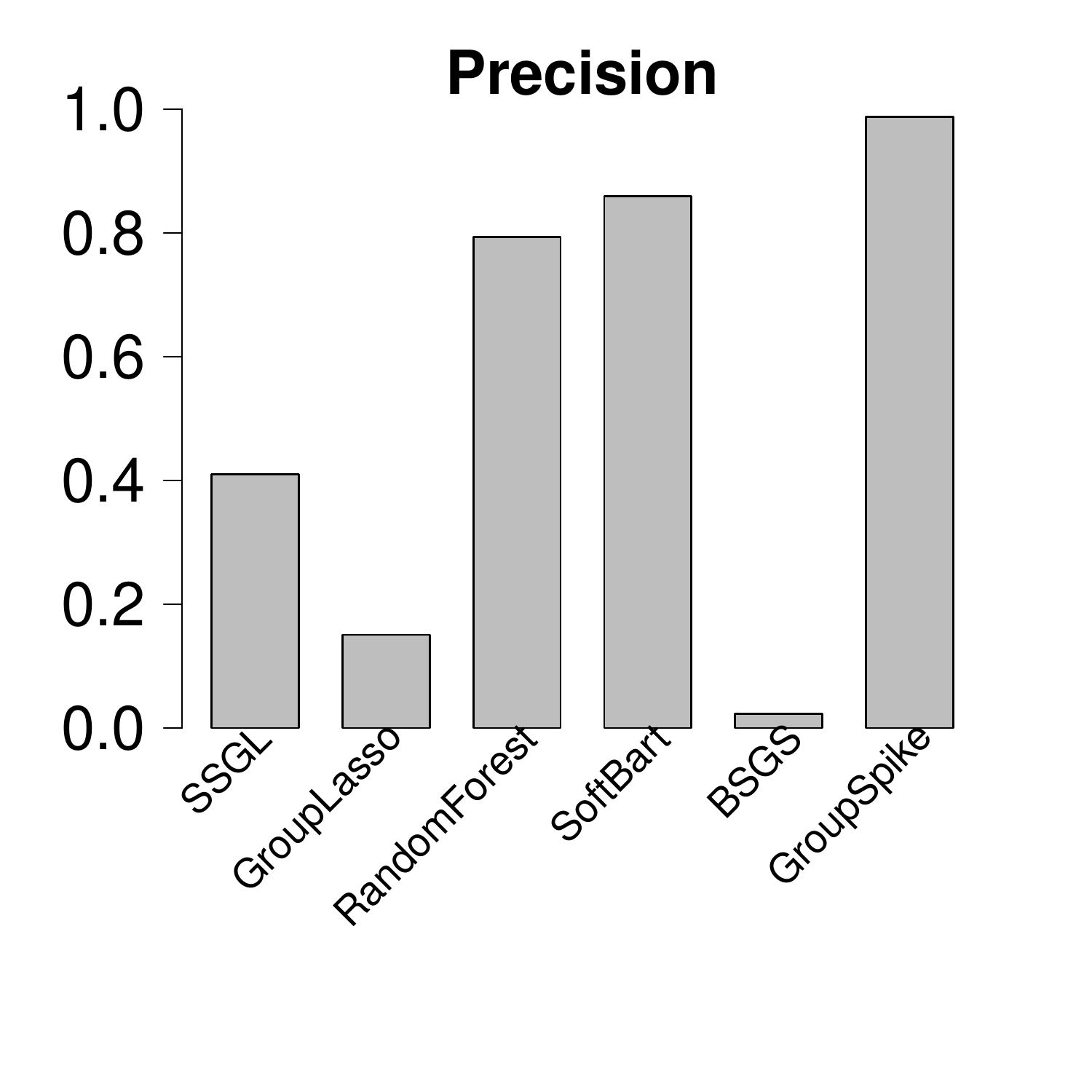}
	\includegraphics[width=0.32\linewidth]{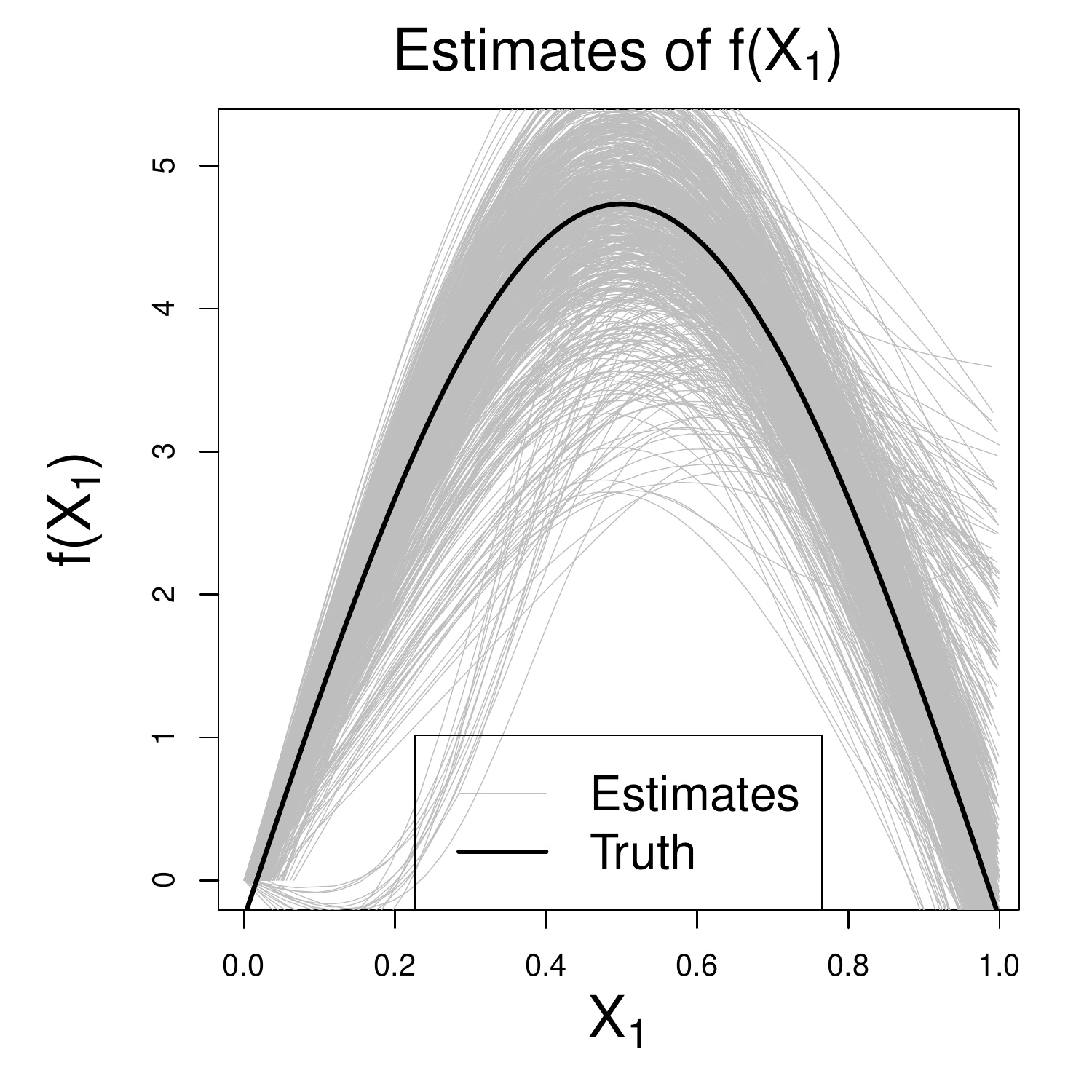}
	\caption{Simulation results for semiparametric regression. The top left panel presents the out-of-sample mean squared error, the top right panel shows the recall score to evaluate variable selection, the bottom left panel shows the precision score, and the bottom right panel shows the estimates from each simulation of $f_1(X_1)$ for SSGL. The MSE for BSGS is not displayed as it lies outside of the plot area.}
	\label{fig:simSparse}
\end{figure}

To implement the SSGL approach, we estimate the mean response as
\begin{align*}
\mathbb{E} ( \bm{Y} \vert \boldsymbol{X}) = \widetilde{\boldsymbol{X}}_1 \boldsymbol{\beta}_1 + \dots + \widetilde{\boldsymbol{X}}_p \boldsymbol{\beta}_p,
\end{align*}
where $\widetilde{\boldsymbol{X}}_j$ is a design matrix of basis functions used to capture the possibly nonlinear effect of $X_j$ on $Y$. For the basis functions in $\widetilde{\boldsymbol{X}}_j, j = 1, \ldots, p$, we use natural splines with degrees of freedom $d$ chosen from  $d \in \{2, 3, 4\}$ using cross-validation. Thus, we are estimating a total of between 600 and 1200 unknown basis coefficients. 

We run 1000 simulations and average all of the metrics considered over each simulated data set.  Figure \ref{fig:simSparse} shows the results from this simulation study. The GroupSpike approach has the best performance in terms of MSE, followed closely by SSGL, with the next best approach being SoftBart. In terms of recall, the SSGL and GroupLasso approaches perform the best, indicating the highest power in detecting the significant groups. This comes with a loss of precision as the GroupSpike and SoftBart approaches have the best precision among all methods.

Although the GroupSpike method performed best in this scenario, the SSGL method was much faster. As we show in Appendix \ref{App:B5}, when $p=4000$, fitting the SSGL model with a sufficiently large $\lambda_0$ takes around three seconds to run. This is almost 50 times faster than running 100 MCMC iterations of the GroupSpike method (never mind the total time it takes for the GroupSpike model to converge). Our experiments demonstrate that the SSGL model gives comparable performance to the ``theoretically ideal'' point mass spike-and-slab in a fraction of the computational time. 

\subsection{Interaction Detection}

We now explore the ability of the SSGL approach to identify important interaction terms in a nonparametric regression model. To this end, we implement the NPSSL model with interactions from Section \ref{NPSSLMainInteractionEffects}. We generate 25 independent covariates from a standard uniform distribution with a sample size of 300. Data is generated from the model:
\begin{align*}
\mathbb{E} (Y \vert \boldsymbol{X}) = 2.5\text{sin}(\pi X_1 X_2) + 2\text{cos}(\pi (X_3 + X_5)) + 2(X_6 - 0.5) + 2.5X_7,
\end{align*}
with variance $\sigma^2=1$. While this may not seem like a high-dimensional problem, we will consider all two-way interactions, and there are 300 such interactions. The important two-way interactions are between $X_1$ and $X_2$ and between $X_3$ and $X_5$. We evaluate the performance of each method and examine the ability of SSGL to identify important interactions while excluding all of the remaining interactions. Figure \ref{fig:simInt} shows the results for this simulation setting. The SSGL, GL, GroupSpike, and SoftBart approaches all perform well in terms of out-of-sample mean squared error, with GroupSpike slightly outperforming the competitors. The SSGL also does a very good job at identifying the two important interactions. The $(X_1, X_2)$ interaction is included in 97\% of simulations, while the $(X_3, X_5)$ interaction is included 100\% of the time. All other interactions are included in only a small fraction of simulated data sets. 

\begin{figure}[t!]
	\centering
	\includegraphics[width=0.32\linewidth]{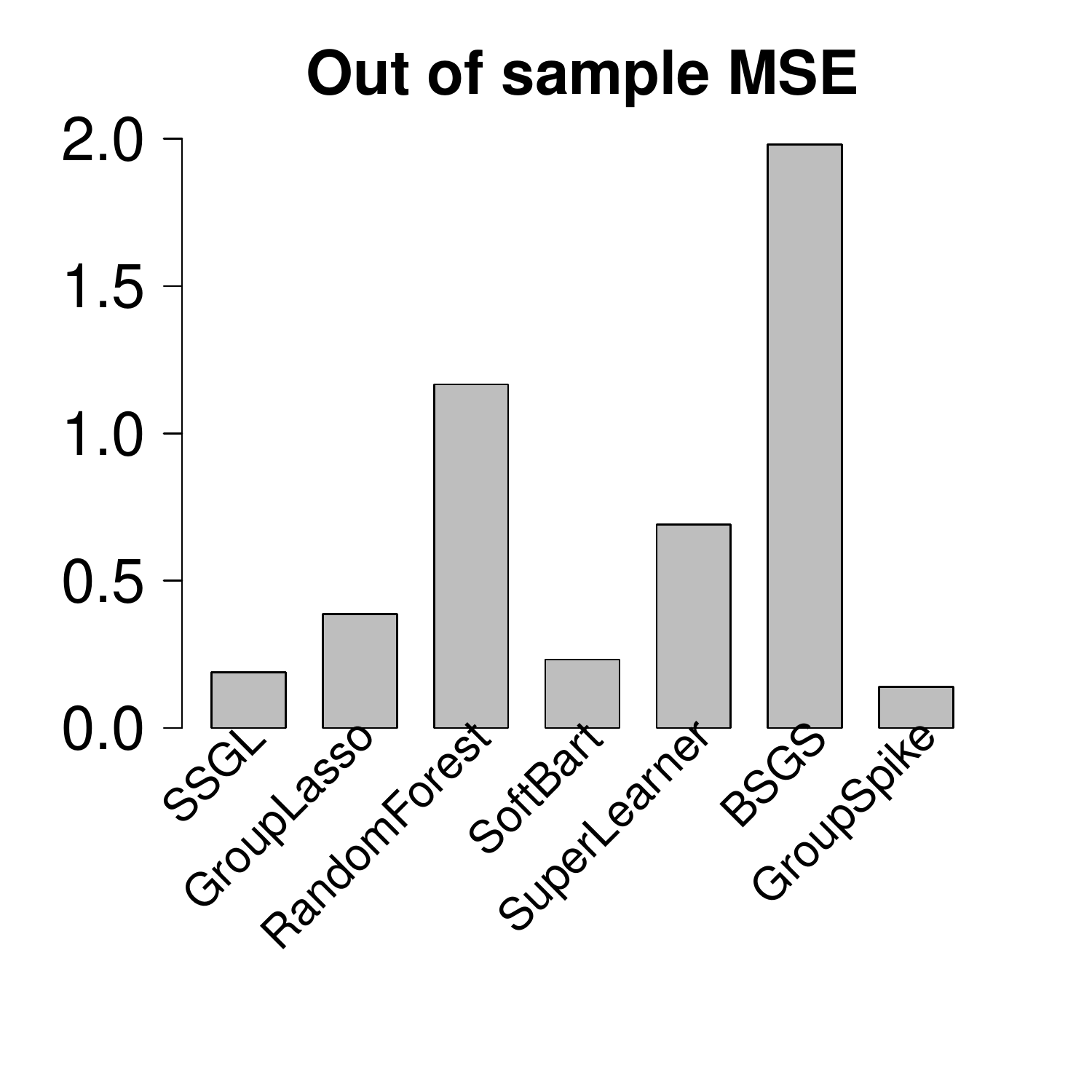}
	\includegraphics[width=0.35\linewidth]{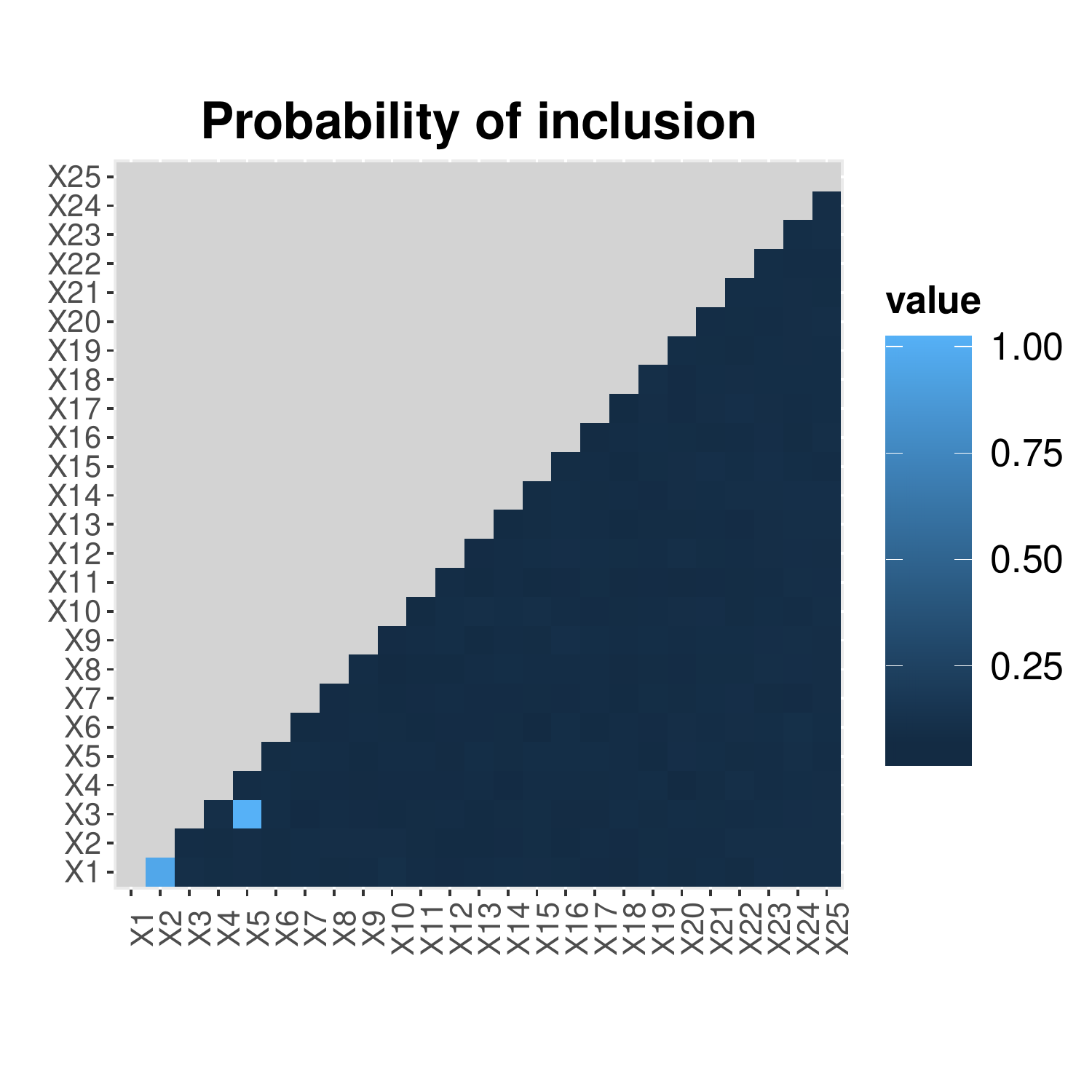}
	\caption{Simulation results from the interaction setting. The left panel shows out-of-sample MSE for each approach, while the right panel shows the probability of a two-way interaction being included into the SSGL model for all pairs of covariates.}
	\label{fig:simInt}
\end{figure}

\section{Real Data Analysis} \label{dataanalysis}

Here, we will illustrate the SSGL procedure in two distinct settings: 1) evaluating the SSGL's performance on a data set where $n=120$ and $p=15,000$, and 2) identifying important (nonlinear) main effects and interactions of environmental exposures. In Appendix \ref{App:C}, we evaluate the predictive performance of our approach on benchmark data sets where $p < n$, compared to several other state-of-the-other methods. Our results show that in both the $p \gg n$ and $p<n$ settings, the SSGL maintains good predictive accuracy.

\subsection{Bardet-Biedl Syndrome Gene Expression Study} \label{bb_subsection}

We now analyze a microarray data set consisting of gene expression measurements from the eye tissue of 120 laboratory rats\footnote{Data accessed from the Gene Expression Omnibus \url{ www.ncbi.nlm.nih.gov/geo} (accession no. GSE5680).}. The data was originally studied by \citet{Scheetz06} to investigate mammalian eye disease, and later analyzed by \citet{BrehenyHuang2015} to demonstrate the performance of their group variable selection algorithm. In this data, the goal is to identify genes which are associated with the gene TRIM32. TRIM32 has previously been shown to cause Bardet-Biedl syndrome \citep{chiang06}, a disease affecting multiple organs including the retina. 

The original data consists of 31,099 probe sets. Following \citet{BrehenyHuang2015}, we included only the 5,000 probe sets with the largest variances in expression (on the log scale). For these probe sets, we considered a three-term natural cubic spline basis expansion, resulting in a grouped regression problem with $n = 120$ and $p = 15,000$.  We implemented SSGL with regularization parameter values $\lambda_1=1$ and $\lambda_0$ ranging on an equally spaced grid from 1 to 500. We compared SSGL with the group lasso \citep{YuanLin2006}, implemented using the R package \texttt{gglasso} \citep{Yang2015}.

As shown in Table \ref{BB_gene_table}, SSGL selected much fewer groups than the group lasso. Namely, SSGL selected 12 probe sets, while the group lasso selected 83 probe sets.  Moreover, SSGL achieved a smaller 10-fold cross-validation error than the group lasso, albeit within range of random variability (Table \ref{BB_gene_table}).  These results demonstrate that the SSGL achieves strong predictive accuracy, while \textit{also} achieving the most parsimony. The groups selected by both SSGL and the group lasso are displayed in Table \ref{bb_gene_table} of Appendix \ref{App:C}. Interestingly, only four of the 12 probes selected by SSGL were also selected by the group lasso. 

\begin{table}[t!]
	\centering
	\begin{tabular}{lcc}
		\hline
		& SSGL & Group Lasso \\ 
		\hline
		\# groups selected & 12 & 83\\
		10-fold CV error &  0.012 (0.003) &  0.017 (0.008)  \\
		\hline
	\end{tabular}
	\caption{Results for SSGL and Group Lasso on the Bardet-Biedl syndrome gene expression data set. In parentheses, we report the standard errors for the CV prediction error. }\label{BB_gene_table}
\end{table}

We next conducted gene ontology enrichment analysis on the group of genes found by each of the methods using the R package  \texttt{clusterProfiler} \citep{YWY12}. This software determines whether subsets of genes known to act in a biological process are overrepresented in a group of genes, relative to chance. If such a subset is significant, the group of genes is said to be ``enriched'' for that biological process. With a false discovery rate of 0.01, SSGL had five enriched terms, while the group lasso had none. The terms for which SSGL was enriched included RNA binding, a biological process with which the response gene TRIM32 is associated.\footnote{https://www.genecards.org/cgi-bin/carddisp.pl?gene=TRIM32 (accessed 03/01/20)} These findings show the ability of SSGL to find biologically meaningful signal in the data. Additional details for our gene ontology enrichment analysis can be found in Appendix \ref{App:C}.

\subsection{Environmental Exposures in the NHANES Data}\label{NHANES} 
Here, we analyze data from the 2001-2002 cycle of the National Health and Nutrition Examination Survey (NHANES), which was previously analyzed by \citet{antonelli2017estimating}. We aim to identify which organic pollutants are associated with changes in leukocyte telomere length (LTL) levels. Telomeres are segments of DNA that help to protect chromosomes, and LTL levels are commonly used as a proxy for overall telomere length. LTL levels have previously been shown to be associated with adverse health effects \citep{haycock2014leucocyte}, and recent studies within the NHANES data have found that organic pollutants can be associated with telomere length \citep{mitro2015cross}. 

\begin{figure}[t!]
		\centering
		\includegraphics[width=0.77\linewidth]{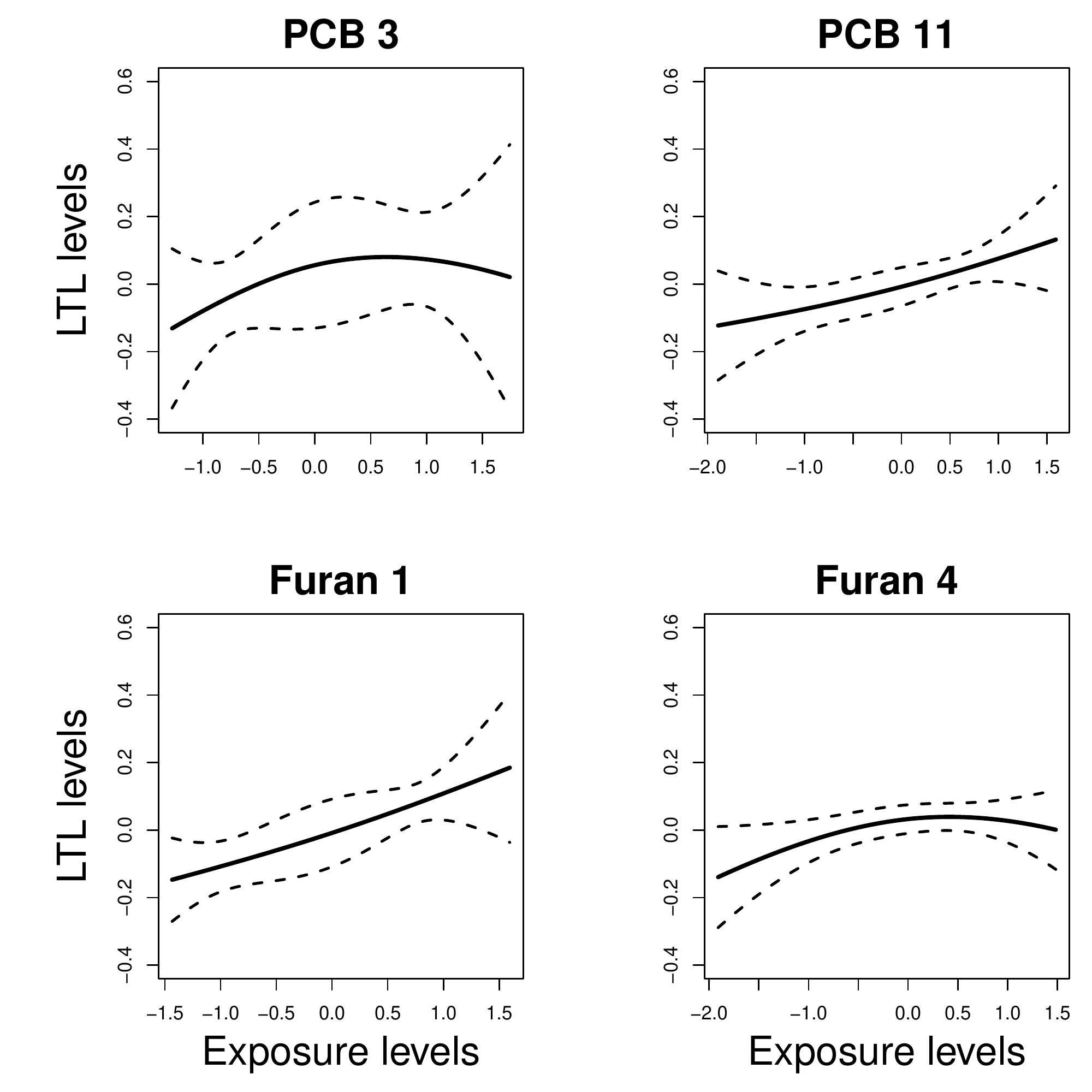}
		\caption{Exposure response curves for each of the four exposures with significant main effects identified by the model. }
		\label{fig:MainEffectNHANES}
\end{figure}

We use the SSGL approach to evaluate whether any of 18 organic pollutants are associated with LTL length and whether there are any significant interactions among the pollutants also associated with LTL length. In addition to the 18 exposures, there are 18 additional demographic variables which we adjust for in our model. We model the effects of the 18 exposures on LTL length using spline basis functions with two degrees of freedom. For the interaction terms, this leads to four terms for each pair of interactions, and we orthogonalize these terms with respect to the main effects. In total, this leads to a data set with $n=1003$ and $p=666$. 

Our model selects four significant main effects and six significant interaction terms. In particular, PCB 3, PCB 11, Furan 1, and Furan 4 are identified as the important main effects in the model. Figure \ref{fig:MainEffectNHANES} plots the exposure response curves for these exposures. We see that each of these four exposures has a positive association with LTL length, which agrees with results seen in \cite{mitro2015cross} that saw positive relationships between persistent organic pollutants and telomere length. Further, our model identifies more main effects and more interactions than previous analyses of these data, e.g. \cite{antonelli2017estimating}, which could lead to more targeted future research in understanding how these pollutants affect telomere length. Additional discussion and analysis of the NHANES data set can be found in Appendix \ref{App:C}.

\section{Discussion} \label{discussion}
We have introduced the spike-and-slab group lasso (SSGL) model for variable selection and linear regression with grouped variables. We also extended the SSGL model to generalized additive models with the nonparametric spike-and-slab lasso (NPSSL). The NPSSL can efficiently identify both nonlinear main effects \textit{and} higher-order nonlinear interaction terms. Moreover, our prior performs an automatic multiplicity adjustment and self-adapts to the true sparsity pattern of the data through a \textit{non}-separable penalty. For computation, we introduced highly efficient coordinate ascent algorithms for MAP estimation and employed de-biasing methods for uncertainty quantification. An \textsf{R} package implementing the SSGL model can be found at \url{https://github.com/jantonelli111/SSGL}.

Although our model performs group selection, it does so in an ``all-in-all-out'' manner, similar to the original group lasso \citep{YuanLin2006}. Future work will be to extend our model to perform both group selection and within-group selection of individual coordinates. We are currently working to extend the SSGL to perform bilevel selection.

We are also working to extend the nonparametric spike-and-slab lasso so it can adapt to even more flexible regression surfaces than the generalized additive model. Under the NPSSL model, we used cross-validation to tune a single value for the degrees of freedom. In reality, different functions can have vastly differing degrees of smoothness, and it will be desirable to model anisotropic regression surfaces while avoiding the computational burden of tuning the individual degrees of freedom over a $p$-dimensional grid.

\section*{Acknowledgments}
Dr. Ray Bai, Dr. Gemma Moran, and Dr. Joseph Antonelli contributed equally and wrote this manuscript together, with input and suggestions from all other listed co-authors. The bulk of this work was done when the first listed author was a postdoc at the Perelman School of Medicine, University of Pennsylvania, under the mentorship of the last two authors. The authors are grateful to three anonymous reviewers, the Associate Editor, and the Editor whose thoughtful comments and suggestions helped to improve this manuscript. The authors would also like to thank Ruoyang Zhang, Peter B{\"u}hlmann, and Edward George for helpful discussions. 

\section*{Funding}
 Dr. Ray Bai and Dr. Mary Boland were funded in part by generous funding from the Perelman School of Medicine, University of Pennsylvania. Dr. Ray Bai and Dr. Yong Chen were funded by NIH grants 1R01AI130460 and 1R01LM012607.

\bibliographystyle{chicago}
\bibliography{SSGLReferences}

\begin{appendix}
	\section{Additional Computational Details} \label{App:A}
	
	\subsection{SSGL Block-Coordinate Ascent Algorithm} \label{completealgorithm}

	\begin{algorithm}[H]
		  \scriptsize \begin{flushleft}
			Input: grid of increasing $\lambda_0$ values $I = \{\lambda_0^{1},\dots, \lambda_0^{L}\}$, update frequency $M$\\[4pt]
			Initialize: $\betab^* = \zerob_p$, $\theta^* = 0.5$, $\sigma^{*2}$ as described in Section \ref{AddlComputationalDetails},  $\Delta^*$ according to \eqref{delta_u} in the main manuscript   \\[4pt]
			For $l = 1, \dots, L$:
			\begin{enumerate}
				\item Set iteration counter $k_l = 0$
				\item Initialize: $\widehat{\betab}^{(k_l)} = \betab^*$, $\theta^{(k_l)} = \theta^*$, $\sigma^{(k_l)2} = \sigma^{*2}$, $\Delta^U = \Delta^*$
				\item While \textsf{diff} $ > \varepsilon$
				\begin{enumerate}
					\item Increment $k_l $
					\item For $ g= 1, \dots,G$:
					\begin{enumerate}
						\item Update 
						\begin{equation*}
						{\betab}_{g}^{(k_l)} \leftarrow \frac{1}{n} \left(1-\frac{\sigma^{(k_l)2} \lambda^*({\betab}_{g}^{(k_l - 1)} ; \theta^{(k_l)} )}{\lVert \zb_g\rVert_2}\right)_+\zb_g \ \mathbb{I}(\lVert \zb_g\rVert_2 > \Delta^U) 
						\end{equation*}
						\item Update
						\begin{equation*}
						\widehat{Z}_g = 
						\begin{cases}
						1 &\text{if } {\betab}_{g}^{(k_l)} \neq \zerob_{m_g} \\
						0 &\text{otherwise}
						\end{cases}
						\end{equation*}
						\item If $g \equiv 0 \mod M$:
						\begin{enumerate}
							\item Update 
							$${\theta}^{(k_l)} \leftarrow \frac{a + \sum_{g=1}^G \widehat{Z}_g}{ a + b +G}$$
							\item If $k_{l-1}< 100$:
							\begin{equation*}
							\text{Update } {\sigma}^{(k_l)2} \leftarrow \frac{\lVert \Yb - \Xb{\betab}^{(k_l)}\rVert_2^2}{n+2}
							\end{equation*}
							\item Update 
							\begin{equation*}
							\Delta^U \leftarrow
							\begin{cases}
							\sqrt{2n\sigma^{(k_l)2}\log[1/p^*(\zerob_{m_g};\theta^{(k_l)})]} +\sigma^{(k_l)2}\lambda_1 &\text{if } h(\zerob_{m_g};{\theta}^{(k_l)} ) >0\\
							{\sigma}^{(k_l)2}\lambda^*(\zerob_{m_g};{\theta}^{(k_l)}) &\text{otherwise}
							\end{cases}
							\end{equation*}
						\end{enumerate}
						\item \textsf{diff} $= \lVert {\betab}^{(k_l)} - \betab^{(k_l - 1)}\rVert_2$
					\end{enumerate}
				\end{enumerate}
				\item Assign $\betab^* = \betab^{(k_l)}$, $\theta^* = \theta^{(k_l)}$, $\sigma^{*2}= \sigma^{2(k_l)}$, $\Delta^* = \Delta^U$
			\end{enumerate}
		\end{flushleft}
		\caption{Spike-and-Slab Group Lasso} \label{algorithm}
	\end{algorithm}
	
	\subsection{Tuning Hyperparameters, Initializing Values, and Updating the Variance in Algorithm 1} \label{AddlComputationalDetails}
	
	We keep the slab hyperparameter $\lambda_1$ fixed at a small value. We have found that our results are not very sensitive to the choice of $\lambda_1$. This parameter controls the variance of the slab component of the prior, and the variance must simply be large enough to avoid overshrinkage of important covariates. For the default implementation, we recommend fixing $\lambda_1 = 1$. This applies minimal shrinkage to the significant groups of coefficients and affords these groups the ability to escape the pull of the spike. 
	
	Meanwhile, we choose the spike parameter $\lambda_0$ from an increasing ladder of values. We recommend selecting $\lambda_0 \in \{ 1,2,...,100 \}$, which represents a range from hardly any penalization to very strong penalization. Below, we describe precisely how to tune $\lambda_0$.  To account for potentially different group sizes, we use the same $\lambda_0$ for all groups but multiply $\lambda_0$ by $\sqrt{m_g}$ for each $g$th group, $g=1, \ldots, G$. As discussed in \cite{HBM12}, further scaling of the penalty by group size is necessary in order to ensure that the same degree of penalization is applied to potentially different sized groups. Otherwise, larger groups may be erroneously selected simply because they are larger (and thus have larger $\ell_2$ norm), not because they contain significant entries.
	
	When the spike parameter $\lambda_0$ is very large, the continuous spike density approximates the point-mass spike. Consequently, we face the computational challenge of navigating a highly multimodal posterior. To ameliorate this problem for the spike-and-slab lasso, \citet{RockovaGeorge2018} recommend a ``dynamic posterior exploration'' strategy in which the slab parameter $\lambda_1$ is held fixed at a small value and $\lambda_0$ is gradually increased along a grid of values. Using the solution from a previous $\lambda_0$ as a ``warm start'' allows the procedure to more easily find optimal modes. In particular, when $(\lambda_1 - \lambda_0)^2 \leq 4$, the posterior is convex. 
	
	\citet{MoranRockovaGeorge2018} modify this strategy for the unknown $\sigma^2$ case. This is because the posterior is always non-convex when $\sigma^2$ is unknown. Namely, when $p\gg n$ and $\lambda_0 \approx \lambda_1$, the model can become saturated, causing the residual variance to go to zero. To avoid this suboptimal mode at $\sigma^2 = 0$, \citet{MoranRockovaGeorge2018} recommend fixing $\sigma^2$ until the $\lambda_0$ value at which the algorithm starts to converge in less than 100 iterations. Then, $\betab$ and $\sigma^2$ are simultaneously updated for the next largest $\lambda_0$ in the sequence. The intuition behind this strategy is we first find a solution to the convex problem (in which $\sigma^2$ is fixed) and then use this solution as a warm start for the non-convex problem (in which $\sigma^2$ can vary). 
	
	We pursue a similar ``dynamic posterior exploration'' strategy with the modification for the unknown variance case for the SSGL in Algorithm \ref{algorithm} of Section \ref{completealgorithm}.  A key aspect of this algorithm is how to choose the maximum value of $\lambda_0$. \citet{RockovaGeorge2018} recommend this maximum to be the $\lambda_0$ value at which the estimated coefficients stabilize. An alternative approach is to choose the maximum $\lambda_0$ using cross-validation, a strategy which is made computationally feasible by the speed of our block coordinate ascent algorithm. In our experience, the dynamic posterior exploration strategy favors more parsimonious models than cross-validation. In the simulation studies in Section \ref{Simulations}, we utilize cross-validation to choose $\lambda_0$, as there, our primary goal is predictive accuracy rather than parsimony.
	
	Following \cite{MoranRockovaGeorge2018}, we initialize $\bm{\beta}^{*} = \bm{0}_p$ and $\theta^{*} = 0.5$. We also initialize $\sigma^{*2}$ to be the mode of a scaled inverse chi-squared distribution with degrees of freedom $\nu=3$ and scale parameter chosen such that the sample variance of $\Yb$ corresponds to the 90th quantile of the prior. We have found this initialization to be quite effective in practice at ensuring that Algorithm \ref{algorithm} converges in less than 100 iterations for sufficiently large $\lambda_0$.	
	
	\subsection{Additional Details for the Inference Procedure} \label{ThetaEstimateDebiasing}
	Here, we describe the nodewise regression procedure for estimating $\widehat{\boldsymbol{\Theta}}$ in Section \ref{sec:inference}. This approach for estimating the inverse of the covariance matrix $\widehat{\boldsymbol{\Sigma}} = \Xb^T \Xb / n$ was originally proposed and studied theoretically in \cite{meinshausen2006high} and \cite{van2014asymptotically}.
	
	For each $j=1,\dots,p$, let $\boldsymbol{X}_j$ denote the $j$th column of $\boldsymbol{X}$ and $\boldsymbol{X}_{-j}$ denote the submatrix of $\boldsymbol{X}$ with the $j$th column removed. Define $\widehat{\boldsymbol{\gamma}}_j$ as
	\begin{align*}
	\widehat{\boldsymbol{\gamma}}_j = \displaystyle \argmin_{\gamma}(|| \boldsymbol{X}_{j} - \boldsymbol{X}_{-j} \boldsymbol{\gamma}||^2_2 / n + 2 \lambda_j ||\boldsymbol{\gamma}||_1).
	\end{align*}
	
	\noindent Now we can define the components of $\widehat{\boldsymbol{\gamma}}_j$ as $\widehat{\boldsymbol{\gamma}}_{j,k}$ for $k = 1, \dots, p$ and $k \neq p$, and create the following matrix:
	$$
	\widehat{\bm{C}} =
	\begin{pmatrix}
	1&-\widehat{\boldsymbol{\gamma}}_{1,2}&\dots&-\widehat{\boldsymbol{\gamma}}_{1,p}\\
	-\widehat{\boldsymbol{\gamma}}_{2,1}&1&\dots&-\widehat{\boldsymbol{\gamma}}_{2,p}\\
	\vdots&\vdots&\ddots&\vdots\\
	-\widehat{\boldsymbol{\gamma}}_{p,1}&-\widehat{\boldsymbol{\gamma}}_{p,2}&\dots&1
	\end{pmatrix}.
	$$
	
	\noindent Lastly, let $\widehat{\bm{T}}^2 = \text{diag}(\widehat{\tau}_1^2, \widehat{\tau}_2^2, \dots, \widehat{\tau}_p^2)$, where 
	
	$$\widehat{\tau}_j = || \boldsymbol{X}_{j} - \boldsymbol{X}_{-j} \widehat{\boldsymbol{\gamma}}_j||^2_2 / n + \lambda_j ||\widehat{\boldsymbol{\gamma}}_j||_1.$$
	
	\noindent We can proceed with $\widehat{\boldsymbol{\Theta}} = \widehat{\bm{T}}^{-2} \widehat{\bm{C}}$. This choice is used because it puts an upper bound on $|| \widehat{\boldsymbol{\Sigma}} \widehat{\boldsymbol{\Theta}}_j^T - \bm{e}_j||_{\infty}$. Other regression models such as the original spike-and-slab lasso \citep{RockovaGeorge2018} could be used instead of the lasso \citep{Tibshirani1996} regressions for each covariate. However, we will proceed with this choice, as it has already been studied theoretically and shown to have the required properties to be able to perform inference for $\boldsymbol{\beta}$. 
	
	\section{Additional Simulation Results} \label{App:B}
	
	Here, we present additional results which include different sample sizes than those seen in the manuscript, assessment of the SSGL procedure under dense settings, estimates of $\sigma^2$, timing comparisons, and additional figures.
	
	\subsection{Increased Sample Size for Sparse Simulation} \label{App:B1}
	
	Here, we present the same sparse simulation setup as that seen in Section \ref{simsparse}, though we will increase $n$ from 100 to 300. Figure \ref{fig:AppendixSimSparse} shows the results and we see that they are very similar to those from the manuscript, except that the mean squared error (MSE) for the SSGL approach is now nearly as low as the MSE for the GroupSpike approach, and the precision score has improved substantially. 
	
	\begin{figure}[t!]
		\centering
		\includegraphics[width=0.3\linewidth]{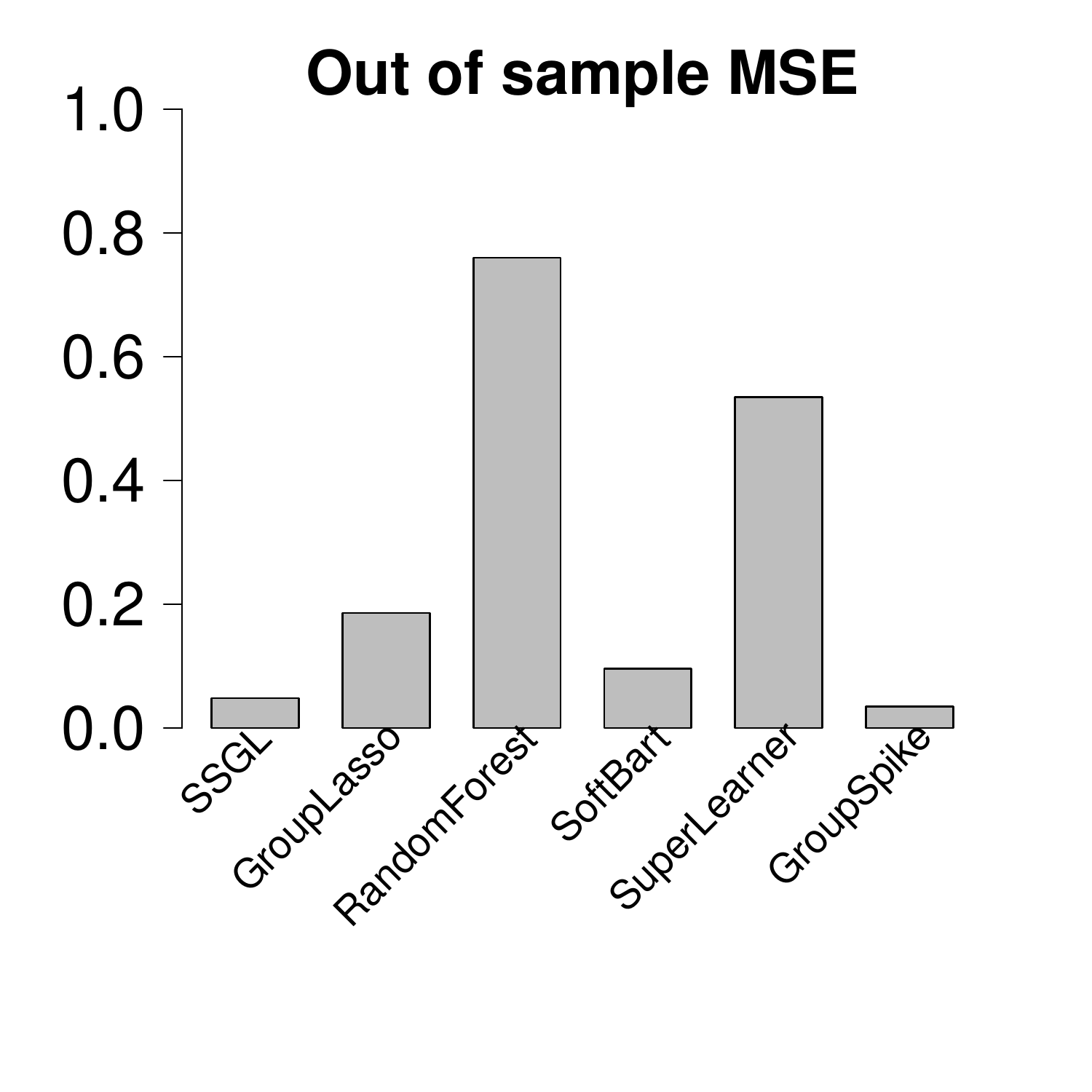}
		\includegraphics[width=0.3\linewidth]{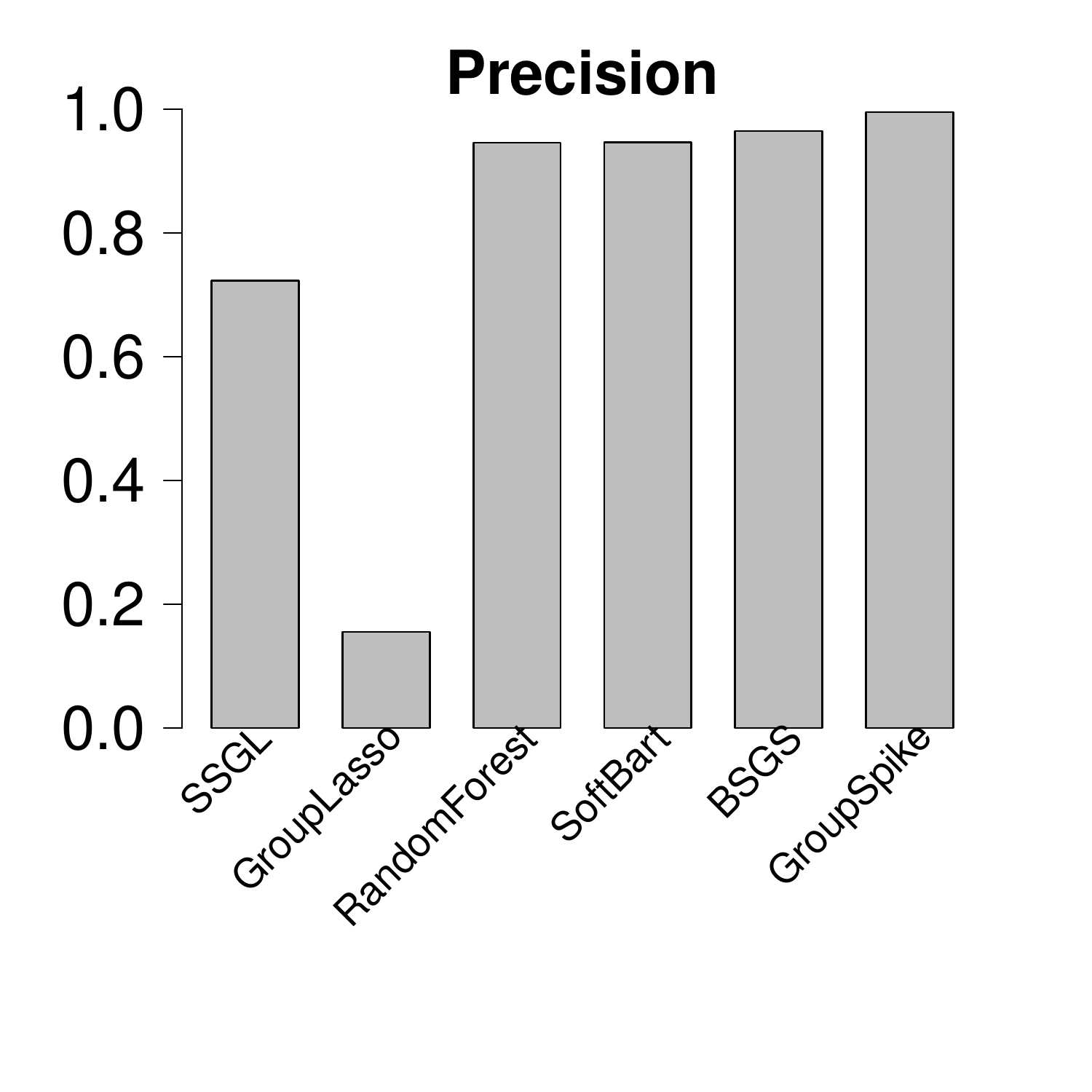}
		\includegraphics[width=0.3\linewidth]{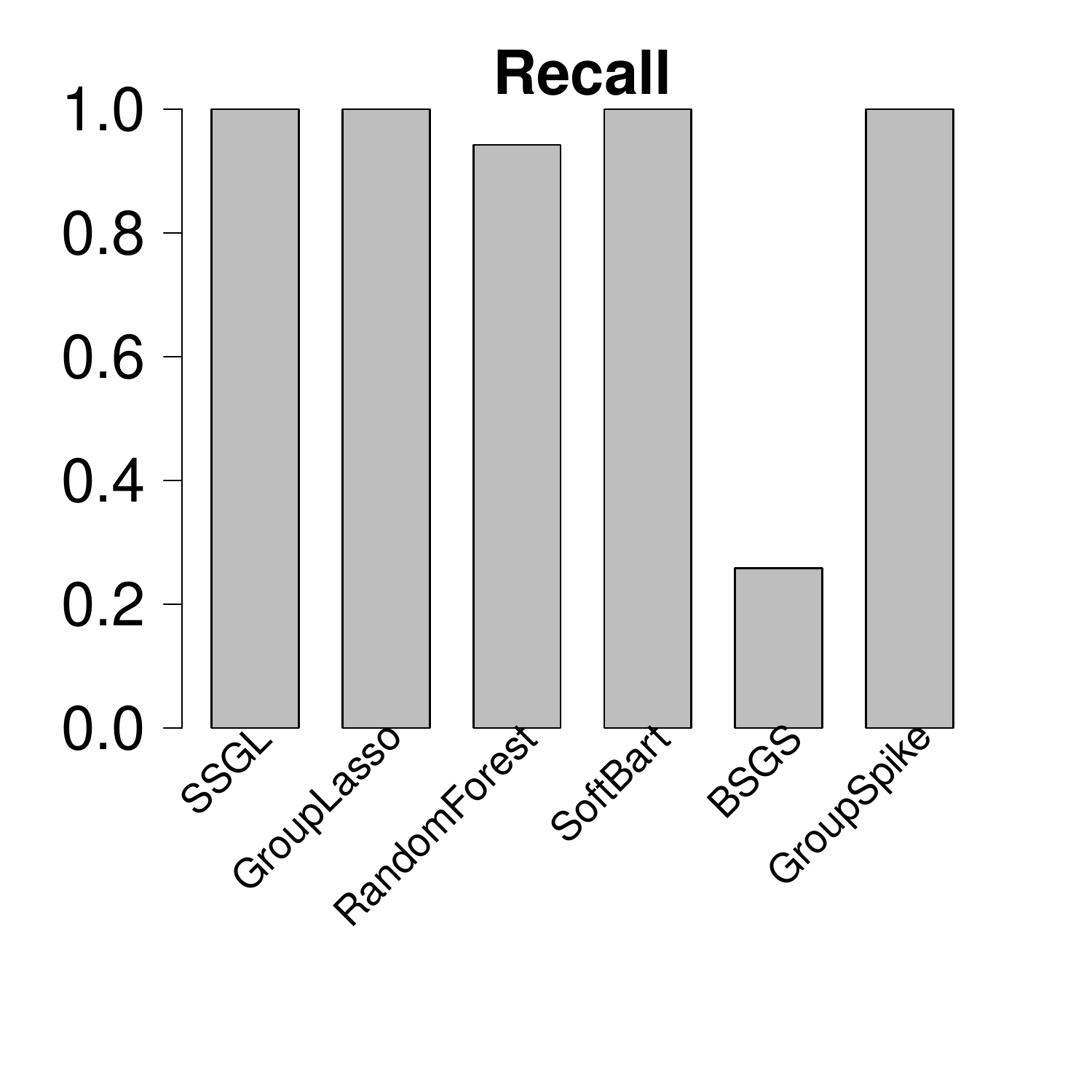}\\
		\caption{Simulation results from the sparse setting with $n=300$. The left panel presents the out-of-sample mean squared error, the middle panel shows the precision score, and the right panel shows the recall score. The MSE for BSGS is not displayed as it lies outside of the plot area.}
		\label{fig:AppendixSimSparse}
	\end{figure}
	
	\begin{figure}[t!]
		\centering
		\includegraphics[width=0.3\linewidth]{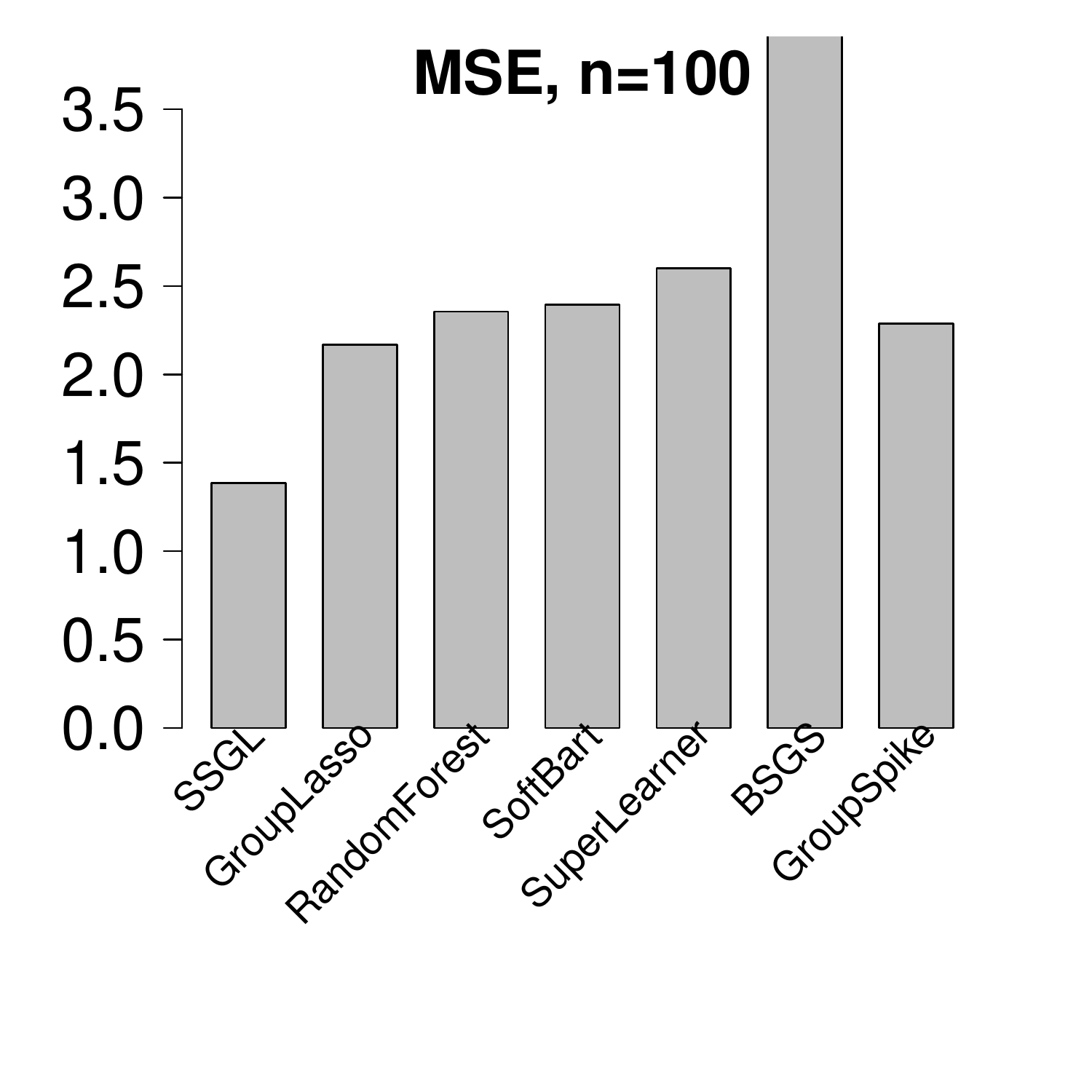}
		\includegraphics[width=0.3\linewidth]{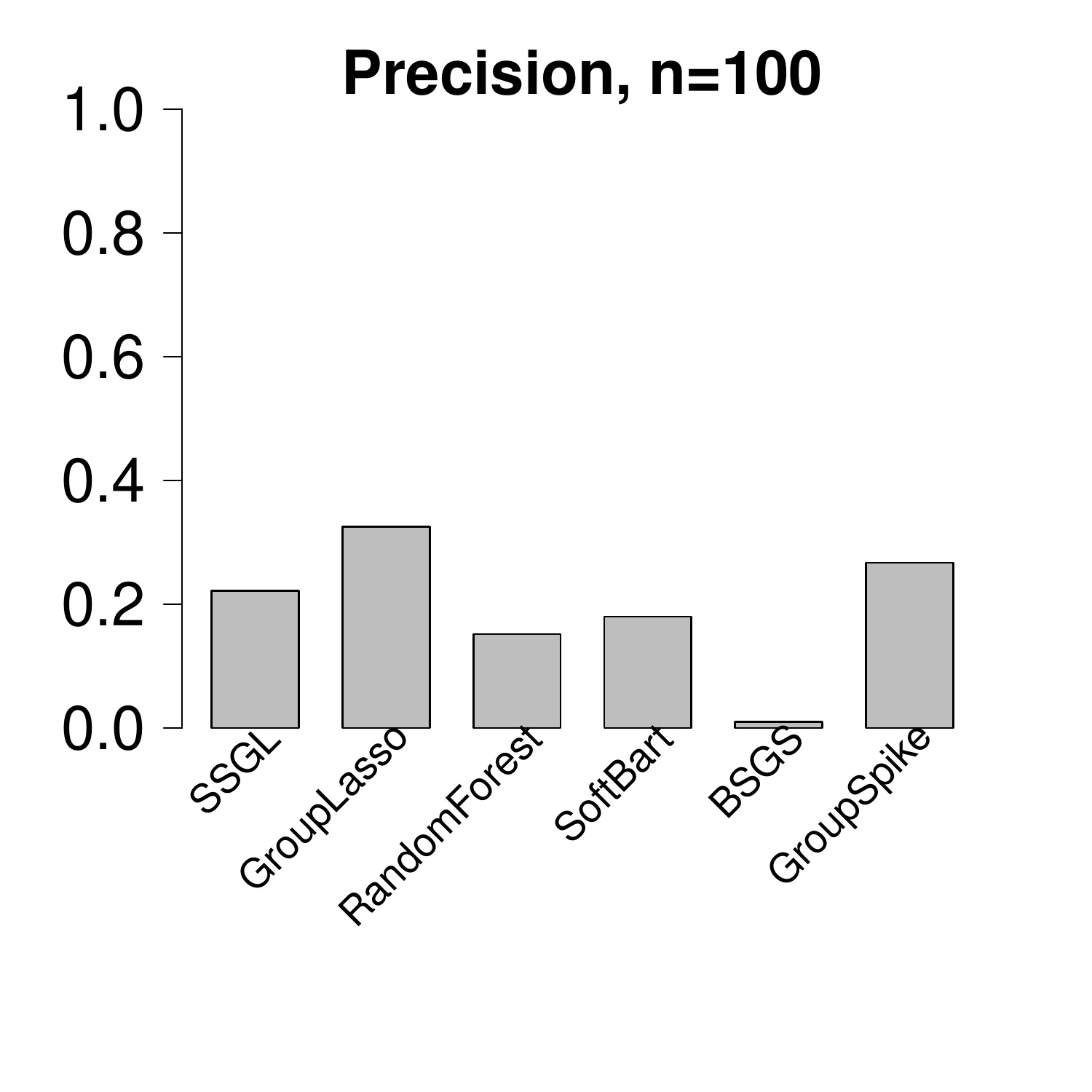}
		\includegraphics[width=0.3\linewidth]{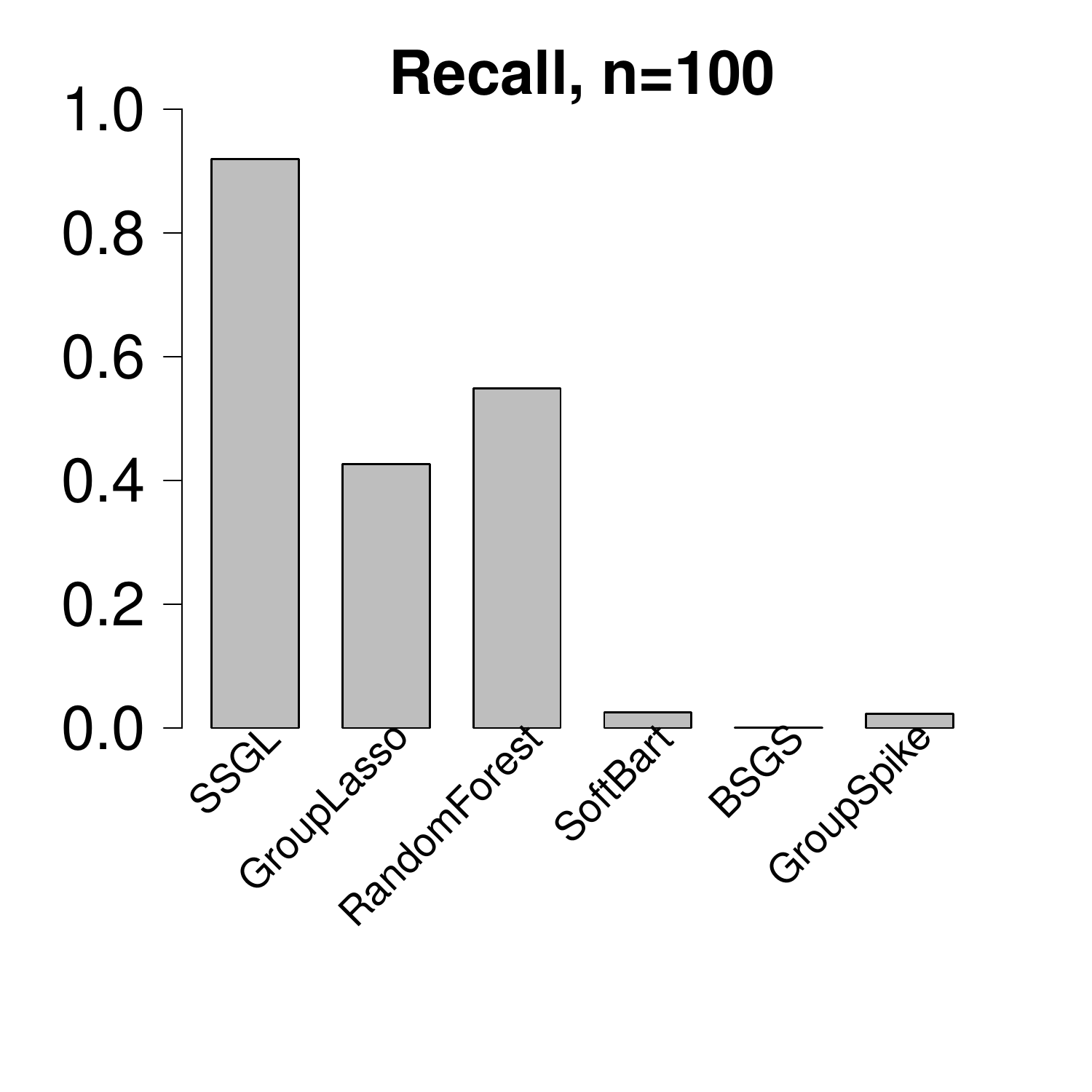}\\
		\includegraphics[width=0.3\linewidth]{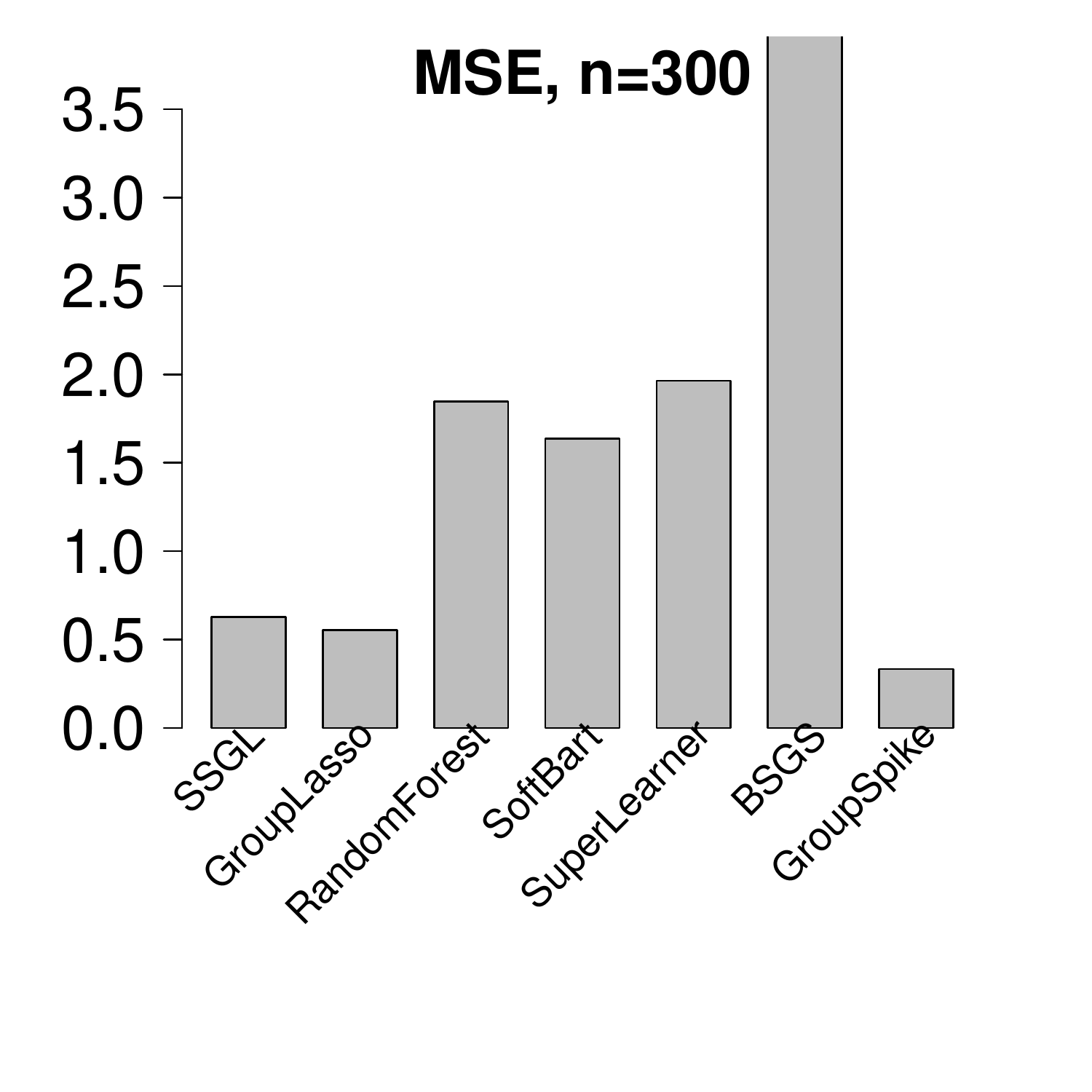}
		\includegraphics[width=0.3\linewidth]{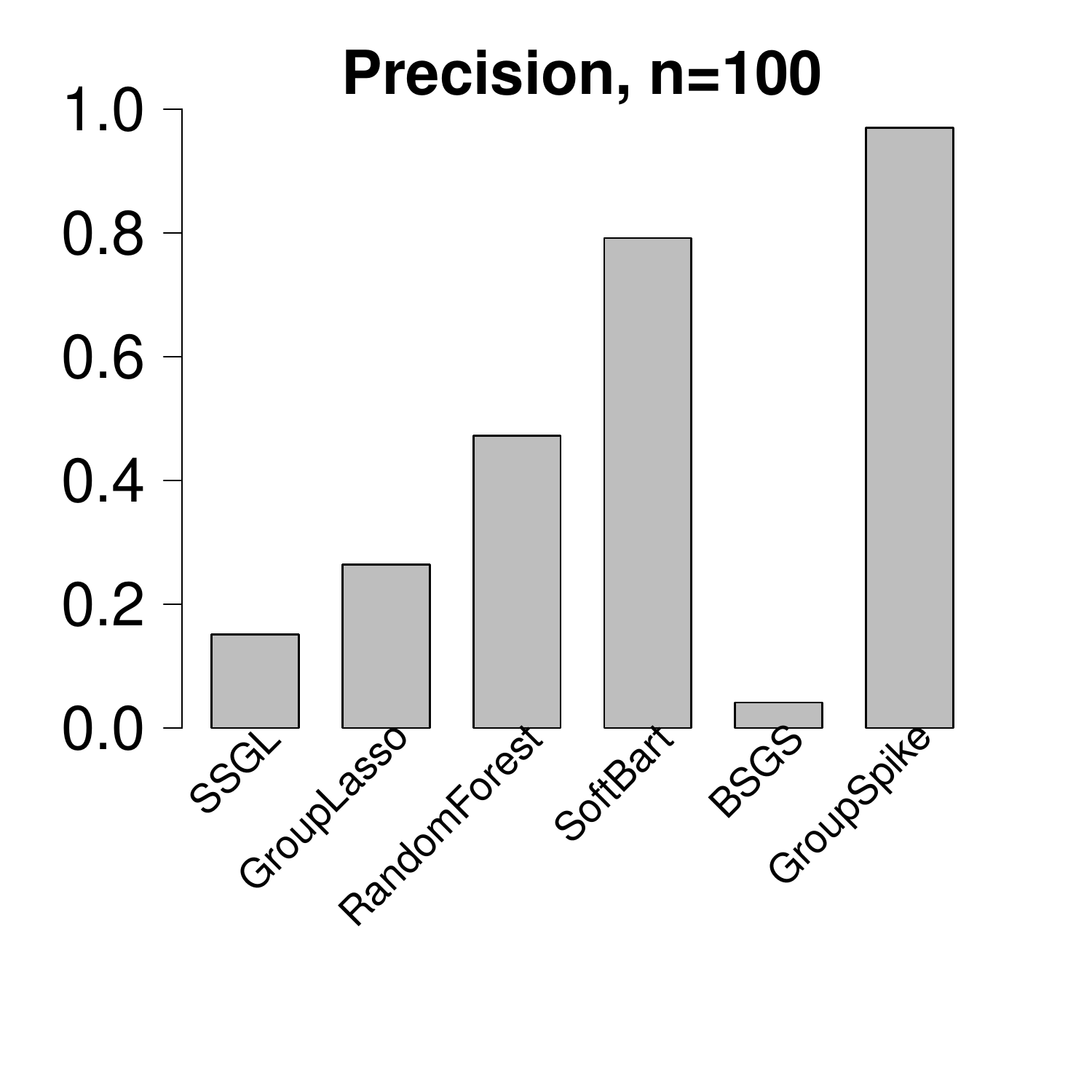}
		\includegraphics[width=0.3\linewidth]{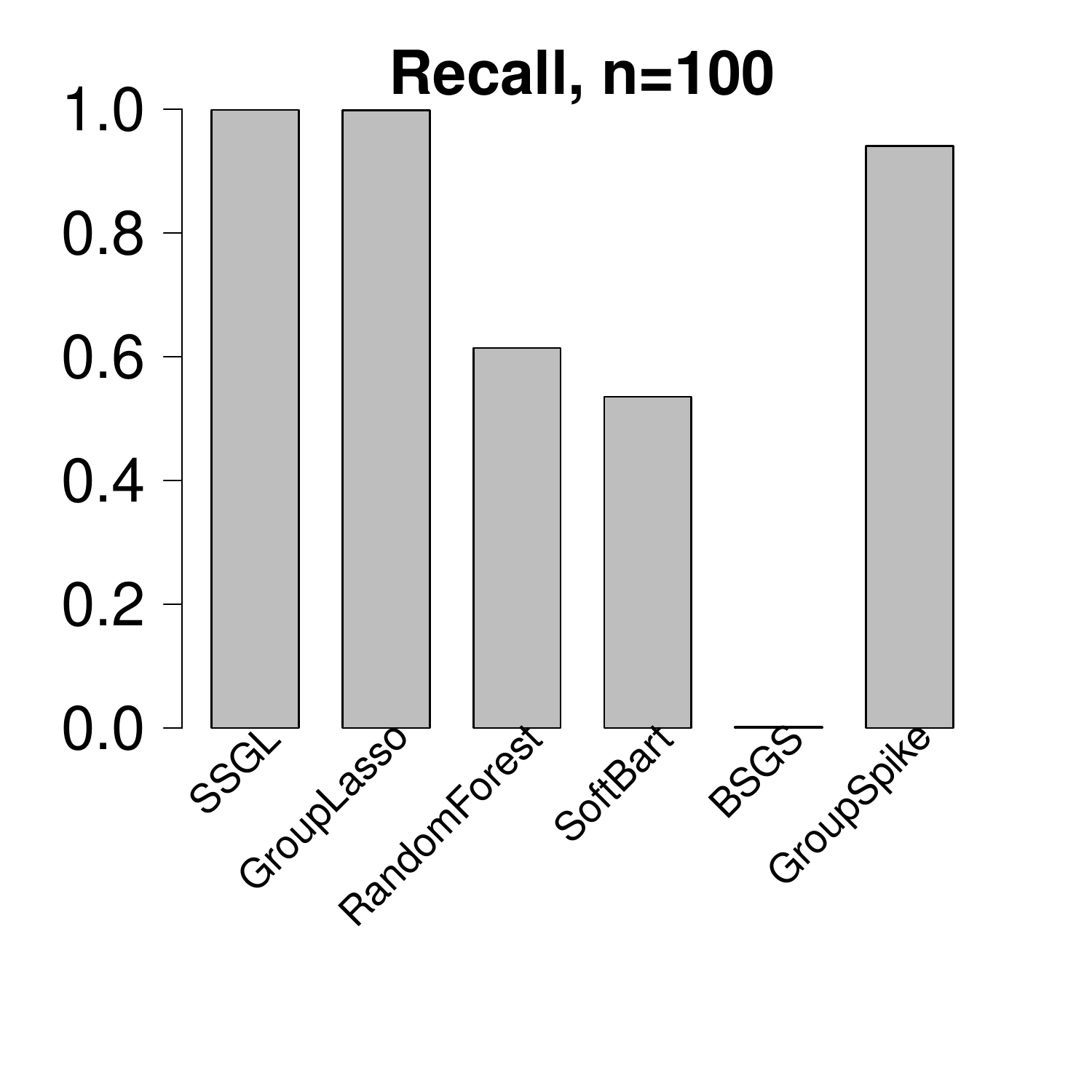}
		\caption{Simulation results from the less sparse setting with $n=100$ and $n=300$. The left column shows out-of-sample MSE, the middle panel shows the precision score, and the right column shows the recall score.}
		\label{fig:simDense}
	\end{figure}
	
	\subsection{Dense Model} \label{App:B2}
	
	Here, we generate independent covariates from a standard normal distribution, and we let the true regression surface take the following form
	\begin{align*}
	\mathbb{E} (Y \vert \boldsymbol{X}) = \sum_{j=1}^{20} 0.2 X_j + 0.2 X_j^2,
	\end{align*}
	with variance $\sigma^2=1$. In this model, there are no strong predictors of the outcome, but rather a large number of predictors which have small impacts on the outcome. Here, we display results for both $n=100$ and $p=300$, as well as $n=300$ and $p = 300$, as the qualitative results change across the different sample sizes. Our simulation results can be seen in Figure \ref{fig:simDense}. When the sample size is 100, the SSGL procedure performs the best in terms of both MSE and recall score, while all approaches do poorly with the precision score. When the sample size increases to 300, the SSGL approach still performs quite well in terms of MSE and recall, though the GroupLasso and GroupSpike approaches are slightly better in terms of MSE. The SSGL approach still maintains a low precision score while the GroupSpike approach has a very high precision once the sample size is large enough. 
	
	\subsection{Estimation of $\sigma^2$} \label{App:B3}
		
	To evaluate our ability to estimate $\sigma^2$ and confirm our theoretical results that the posterior of $\sigma^2$ contracts around the true parameter, we ran a simulation study using the following data generating model:
	\begin{align*}
	\mathbb{E} (Y \vert \boldsymbol{X}) = 0.5X_1 + 0.3X_2 + 0.6X_{10}^2 - 0.2X_{20},
	\end{align*} 
	with $\sigma^2 = 1$. We vary $n \in \{50, 100, 500, 1000, 2000\}$ and we set $G = n$ to confirm that the estimates are centering around the truth as both the sample size and covariate dimension grows. We use groups of size two that contain both the linear and quadratic term for each covariate. Note that in this setting, the total number of regression coefficients actually \textit{exceeds} the sample size since each group has two terms, leading to a total of $p=2G$ coefficients in the model.
	
		\begin{figure}[t!]
		\centering
		\includegraphics[width=0.47\linewidth]{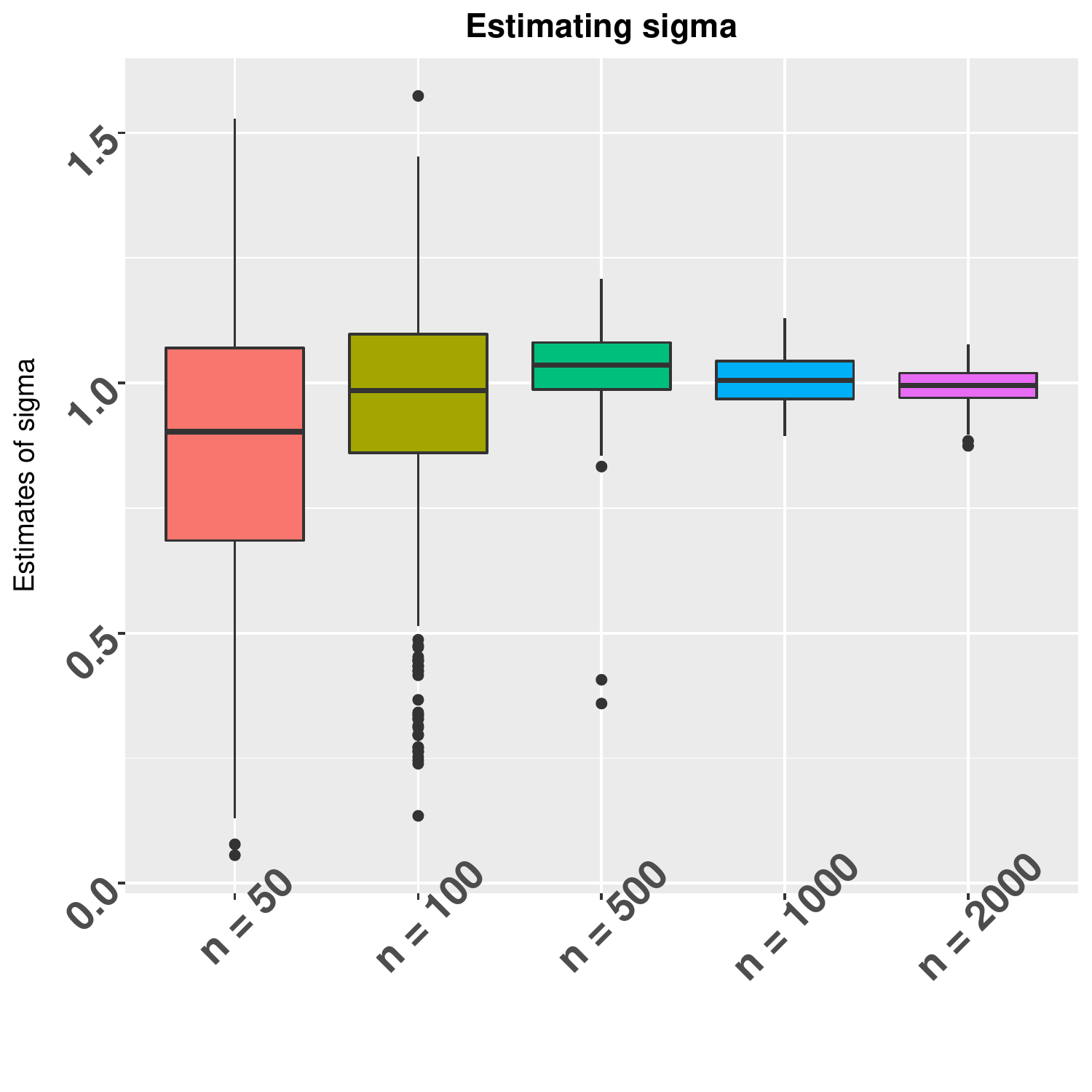}
		\caption{Boxplots of the estimates of $\sigma^2$ from the SSGL model for a range of sample sizes. Note that $n=G$ in each scenario.}
		\label{fig:SigmaTest}
	\end{figure}
	
	Figure \ref{fig:SigmaTest} shows box plots of the estimates for $\sigma^2$ across all simulations for each sample size and covariate dimension. We see that for small sample sizes there are some estimates well above 1 or far smaller than 1. This is because either some important variables are excluded (so the sum of squared residuals gets inflated), or too many variables are included and the model is overfitted (leading to small $\widehat{\sigma}^2$). These problems disappear as the sample size grows to 500 or larger, where we observe that the estimates are closely centering around the true $\sigma^2 = 1$. Figure \ref{fig:SigmaTest} confirms our theoretical results in Theorem \ref{posteriorcontractiongroupedregression} and Theorem \ref{contractionGAMs}, which state that as $n, G \rightarrow \infty$, the posterior $\pi(\sigma^2|\bm{Y})$ contracts around the true $\sigma^2$.
	
	\subsection{Large Number of Groups} \label{App:B4}
	
		\begin{figure}[t]
		\centering
		\includegraphics[width=0.32\linewidth]{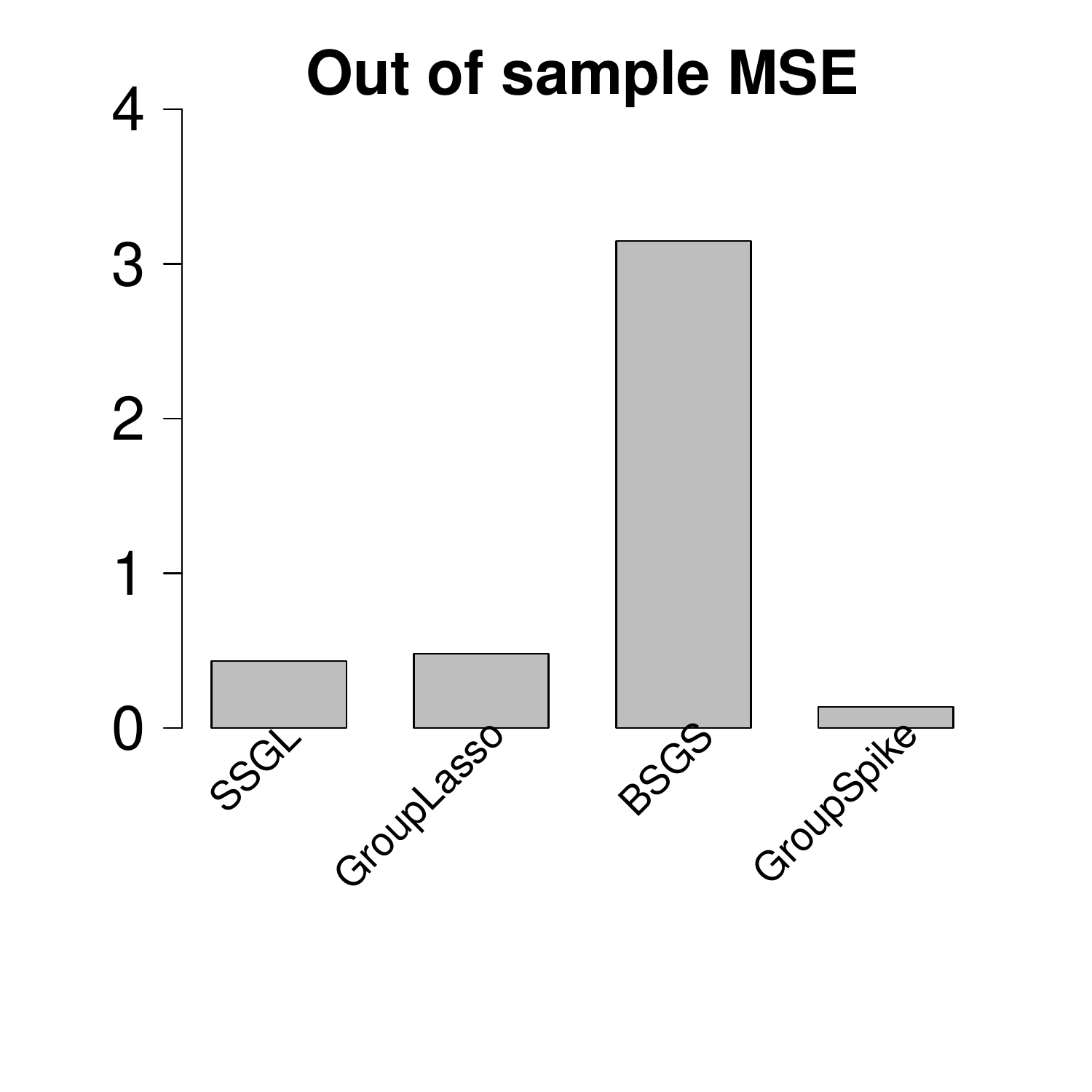}
		\includegraphics[width=0.32\linewidth]{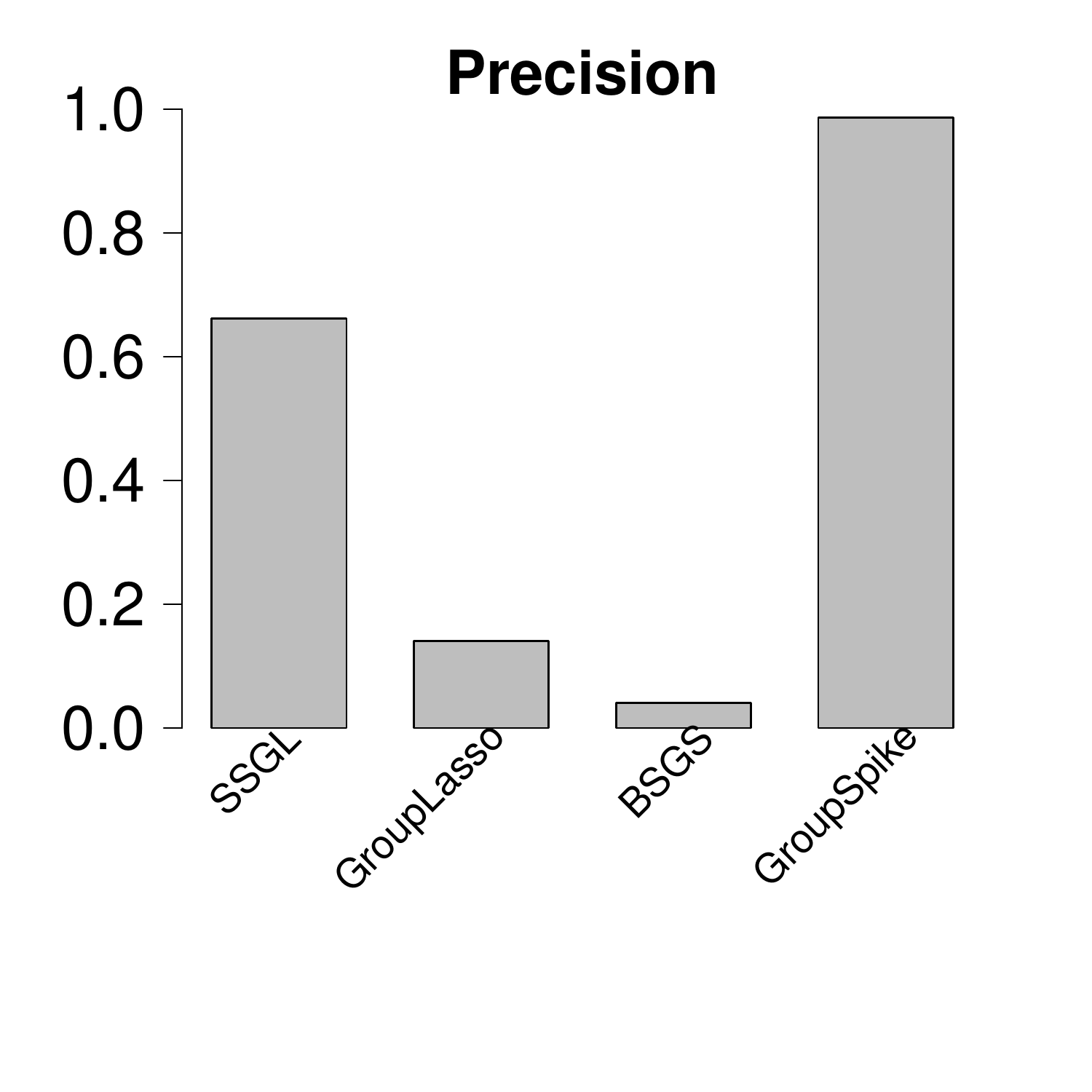}
		\includegraphics[width=0.32\linewidth]{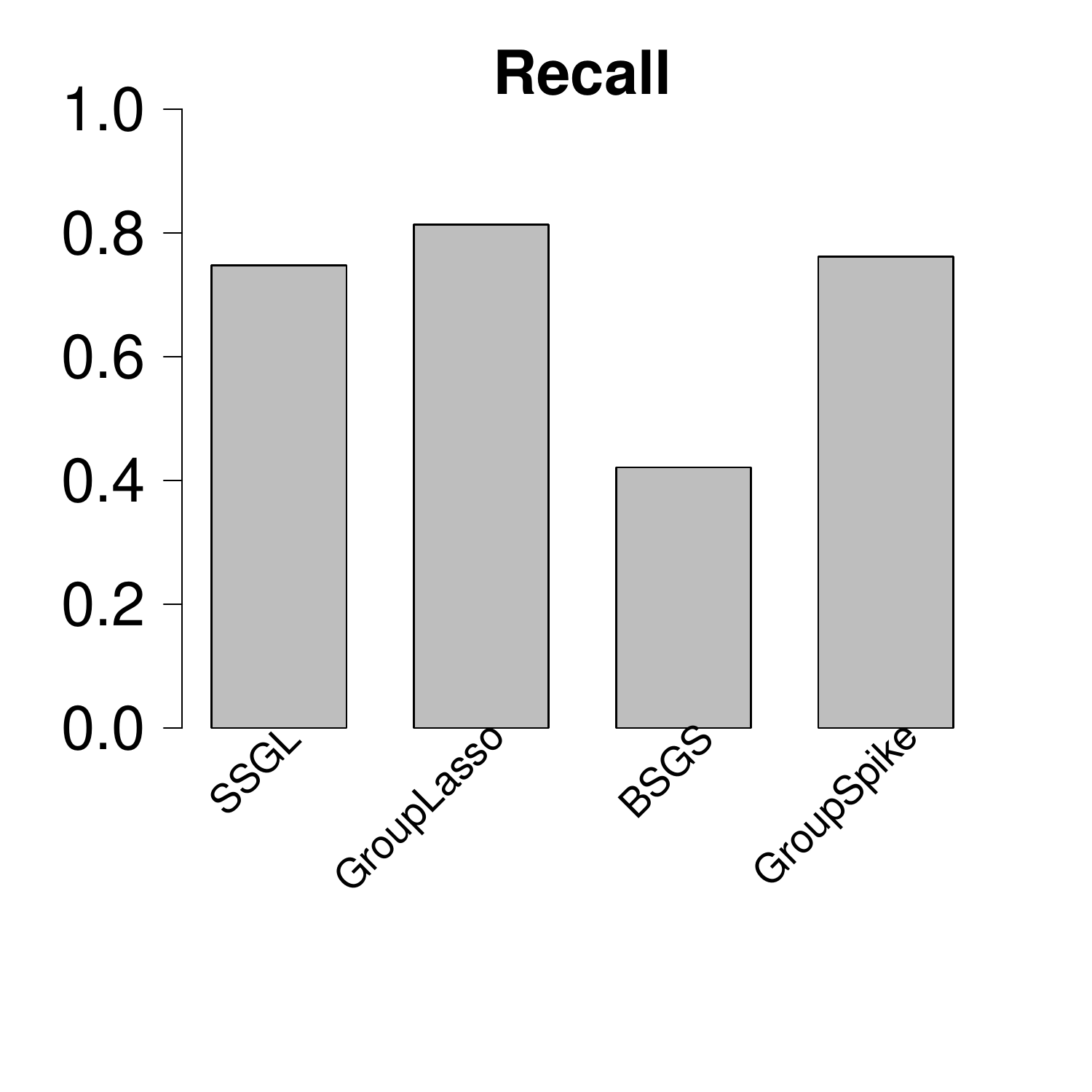}\\
		\caption{Simulation results from the many groups setting with $G=2000$. The left panel presents the out-of-sample mean squared error, the middle panel shows the precision score, and the right panel shows the recall score.}
		\label{fig:AppendixSimBigG}
	\end{figure}
	
	We now generate data with $n=200$ and $G=2000$, where each group contains three predictors. We generate data from the following model:
	\begin{align*}
	\mathbb{E} ( \bm{Y} \vert \boldsymbol{X}) = \sum_{g=1}^G \boldsymbol{X}_g \boldsymbol{\beta}_g,
	\end{align*}
	where we set $\boldsymbol{\beta}_g = \boldsymbol{0}$ for $g=1, \dots 1996$. For the final four groups, we draw individual coefficient values from independent normal distributions with mean 0 and standard deviation 0.4. These coefficients are redrawn for each data set in the simulation study, and therefore, the results are averaging over many possible combinations of magnitudes for the true nonzero coefficients. We see that the best performing approach in this scenario is the GroupSpike approach, followed by the SSGL approach. The SSGL approach outperforms group lasso in terms of MSE and precision, while group lasso has a slightly higher recall score. 
	
	\subsection{Computation Time} \label{CPUExperiment} \label{App:B5}
	
	In this study, we evaluate the computational speed of the SSGL procedure in comparison with the fully Bayesian GroupSpike approach that places point-mass spike-and-slab priors on groups of coefficients. We fix $n=300$ and vary the number of groups $G \in \{ 100, 200, \dots, 2000 \}$, with two elements per group. For the SSGL approach, we keep track of the computation time for estimating the model for $\lambda_0 = 20$. For large values of $\lambda_0$, it typically takes 100 or fewer iterations for the SSGL method to converge. For the GroupSpike approach, we keep track of the computation time required to run 100 MCMC iterations. Both SSGL and GroupSpike models were run on an Intel E5-2698v3 processor.
	
	In any given data set, the computation time will be higher than the numbers presented here because the SSGL procedure typically requires fitting the model for multiple values of $\lambda_0$, while the GroupSpike approach will likely take far more than 100 MCMC iterations to converge, especially in higher dimensions. Nonetheless, this should provide a comparison of the relative computation speeds for each approach. 
	
	The average CPU time in seconds can be found in Figure \ref{fig:CPU}. We see that the SSGL approach is much faster as it is able to estimate all the model parameters for a chosen $\lambda_0$ in just a couple of seconds, even for $G=2000$ (or $p=4000$). When $p=2000$, the SSGL returned a final solution in roughly three seconds on average, whereas GroupSpike required over two minutes to run 100 iterations (and would most likely require many more iterations to converge). This is to be expected as the GroupSpike approach relies on MCMC. Figure \ref{fig:CPU} shows the large computational gains that can be achieved using our MAP finding algorithm.
	
	\begin{figure}[t!]
		\centering
		\includegraphics[width=0.48\linewidth]{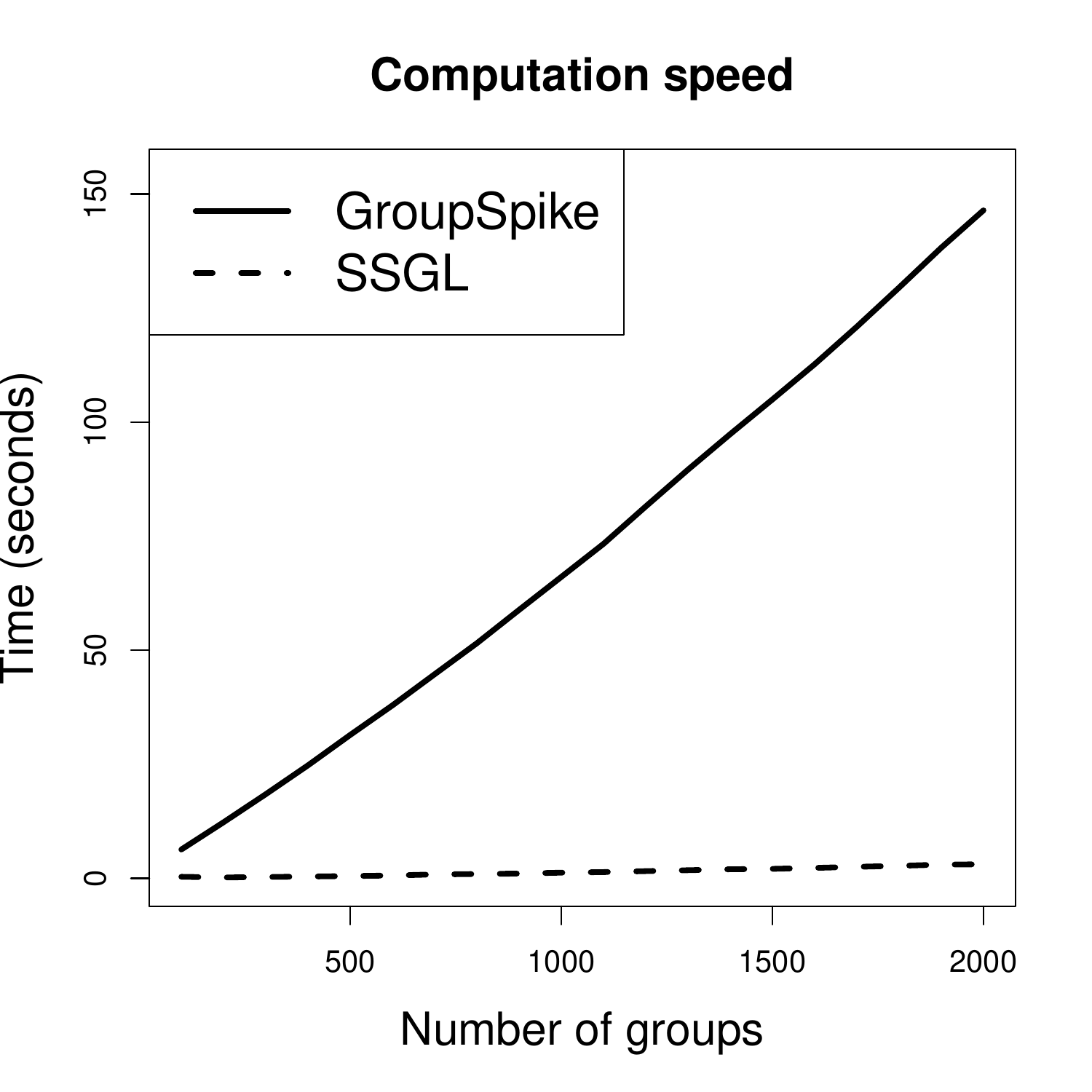}
		\caption{CPU time for the SSGL and GroupSpike approaches averaged across 1000 replications for fixed $n=300$ and different group sizes $G$.}
		\label{fig:CPU}
	\end{figure}
	
	\section{Additional Results and Discussion for Real Data Examples} \label{App:C}
	
	In this section, we perform additional data analysis of the SSGL method on benchmark datasets where $p<n$ to demonstrate that the SSGL model also works well in low-dimensional settings. We also provide additional analyses and discussion of the two real data examples analyzed in Section \ref{dataanalysis}. 
	
	\subsection{Testing Predictive Performance of the SSGL on Datasets Where $p < n$} \label{predictiveperformance} 
	We first look at three data sets which have been analyzed in a number of manuscripts, most recently in \cite{linero2018bayesian}. The tecator data set is available in the \texttt{caret} package in \textsf{R} \citep{kuhn2008building} and has three different outcomes $\Yb$ to analyze. Specifically, this data set looks at using 100 infrared absorbance spectra to predict three different features of chopped meat with a sample size of 215. The Blood-Brain data is also available in the \texttt{caret} package and aims to predict levels of a particular compound in the brain given 134 molecular descriptors with a sample size of 208. Lastly, the Wipp data set contains 300 samples with 30 features from a computer model used to estimate two-phase fluid flow \citep{storlie2011surface}. For each of these data sets, we hold out 20 of the subjects in the data as a validation sample and see how well the model predicts the outcome in the held-out data. We repeat this 1000 times and compute the root mean squared error (RMSE) for prediction. 
	
	\begin{table}[t!]
		\resizebox{.98\textwidth}{!}{
			\centering
			\begin{tabular}{lrrrrrrr}
				\hline
				Data & SSGL & GroupLasso & RandomForest & SoftBart & SuperLearner & BSGS & GroupSpike \\ 
				\hline
				Tecator 1 & 1.41 & 1.57 & 2.75 & 1.93 & 1.00 & 5.16 & 1.67 \\ 
				Tecator 2 & 1.25 & 1.58 & 2.91 & 1.97 & 1.00 & 6.77 & 1.41 \\ 
				Tecator 3 & 1.14 & 1.38 & 1.94 & 1.81 & 1.10 & 3.31 & 1.00 \\ 
				BloodBrain & 1.10 & 1.04 & 1.00 & 1.01 & 1.00 & 1.24 & 1.13 \\ 
				Wipp & 1.44 & 1.30 & 1.46 & 1.00 & 1.17 & 4.68 & 1.30 \\  
				\hline
			\end{tabular}
		}
		\caption{Standardized out-of-sample root mean squared prediction error averaged across 1000 replications for the data sets in Section \ref{predictiveperformance}. An RMSE of 1 indicates the best performance within a data set.}
		\label{tab:pred}
	\end{table}
	
	Table \ref{tab:pred} shows the results for each of the methods considered in the simulation study. The results are standardized so that for each data set, the RMSE is divided by the minimum RMSE for that data set. This means that the model with an RMSE of 1 had the best predictive performance, and all others should be greater than 1, with the magnitude indicating how poor the performance was. We see that the top performer across the data sets was SuperLearner, which is not surprising given that SuperLearner is quite flexible and averages over many different prediction models. Our simulation studies showed that SuperLearner may not work as well when $p > n$. However, the data sets considered here all have $p < n$, which could explain its improved performance here. Among the other approaches, SSGL performs quite well as it has RMSE's relatively close to 1 for all the data sets considered.
	
	\subsection{Additional Details for Bardet-Biedl Analysis}
	
	Here we present additional results for the Bardet-Biedl Syndrome gene expression analysis conducted in Section \ref{bb_subsection}. Table \ref{bb_gene_table} displays the 12 probes found by SSGL. Table \ref{gene_ontology_results} displays the terms for which SSGL was enriched in a gene ontology enrichment analysis. 
	
	\begin{table}[ht]
\centering
\begin{tabular}{llrr}
  \hline
Probe ID & Gene Symbol & SSGL Norm & Group Lasso Norm \\ 
  \hline
1374131\_at &  & 0.034 &  \\ 
  1383749\_at & Phospho1 & 0.067 & 0.088 \\ 
  1393735\_at &  & 0.033 & 0.002 \\ 
  1379029\_at & Zfp62 & 0.074 &  \\ 
  1383110\_at & Klhl24 & 0.246 &  \\ 
  1384470\_at & Maneal & 0.087 & 0.005 \\ 
  1395284\_at &  & 0.014 &  \\ 
  1383673\_at & Nap1l2 & 0.045 &  \\ 
  1379971\_at & Zc3h6 & 0.162 &  \\ 
  1384860\_at & Zfp84 & 0.008 &  \\ 
  1376747\_at &  & 0.489 & 0.002 \\ 
  1379094\_at &  & 0.220 &  \\ 
   \hline
\end{tabular}
\caption{Probes found by SSGL on the Bardet-Biedl syndrome gene expression data set. The probes which were also found by the Group Lasso have nonzero group norm values.}
\label{bb_gene_table}
\end{table}
	
	\begin{table}[ht]
\centering
\begin{tabular}{{ | m{1em} | m{11cm} | }}
  \hline
 & SSGL: enriched terms in gene ontology enrichment analysis \\ 
  \hline
1 & alpha-mannosidase activity \\ 
2 & RNA polymerase II intronic transcription regulatory region sequence-specific DNA binding \\ 
3 & mannosidase activity \\ 
4 & intronic transcription regulatory region sequence-specific DNA binding \\ 
5 & intronic transcription regulatory region DNA binding \\ 
   \hline
\end{tabular}
\caption{Table displays the terms for which SSGL was found to be enriched in a gene ontology enrichment analysis, ordered by statistical significance.}\label{gene_ontology_results}
\end{table}
	
\subsection{Additional Details for Analysis of NHANES Data}

Here we will present additional results from the NHANES data analysis in Section \ref{NHANES}. Here, the aim is to identify environmental exposures that are associated with leukocyte telomere length. In the NHANES data, we have measurements from 18 persistent organic pollutants. Persistent organic pollutants are toxic chemicals that have potential to adversely affect health. They are known to remain in the environment for long periods of time and can travel through wind, water, or even the food chain.  Our data set consists of 11 polychlorinated biphenyls (PCBs), three Dioxins, and four Furans. We want to understand the impact that these can have on telomere length, and to understand if any of these pollutants interact in their effect on humans. 

The data also contains additional covariates that we will adjust for such as age, a squared term for age, gender, BMI, education status, race, lymphocyte count, monocyte count, cotinine level, basophil count, eosinophil count, and neutrophil count. To better understand the data set, we have shown the correlation matrix between all organic pollutants and covariates in Figure \ref{fig:CorrNHANES}. We can see that the environmental exposures are all fairly positively correlated with each other. In particular, the PCBs are highly correlated among themselves. The correlation across chemical types, such as the correlation between PCBs and Dioxins or Furans are lower, though still positively correlated. The correlation between the covariates that we place into our model and the exposures is generally extremely low, and the correlation among the individual covariates is also low, with the exception of a few blood cell types as seen in the upper right of Figure \ref{fig:CorrNHANES}.

\begin{figure}[t!]
		\centering
		\includegraphics[width=0.77\linewidth]{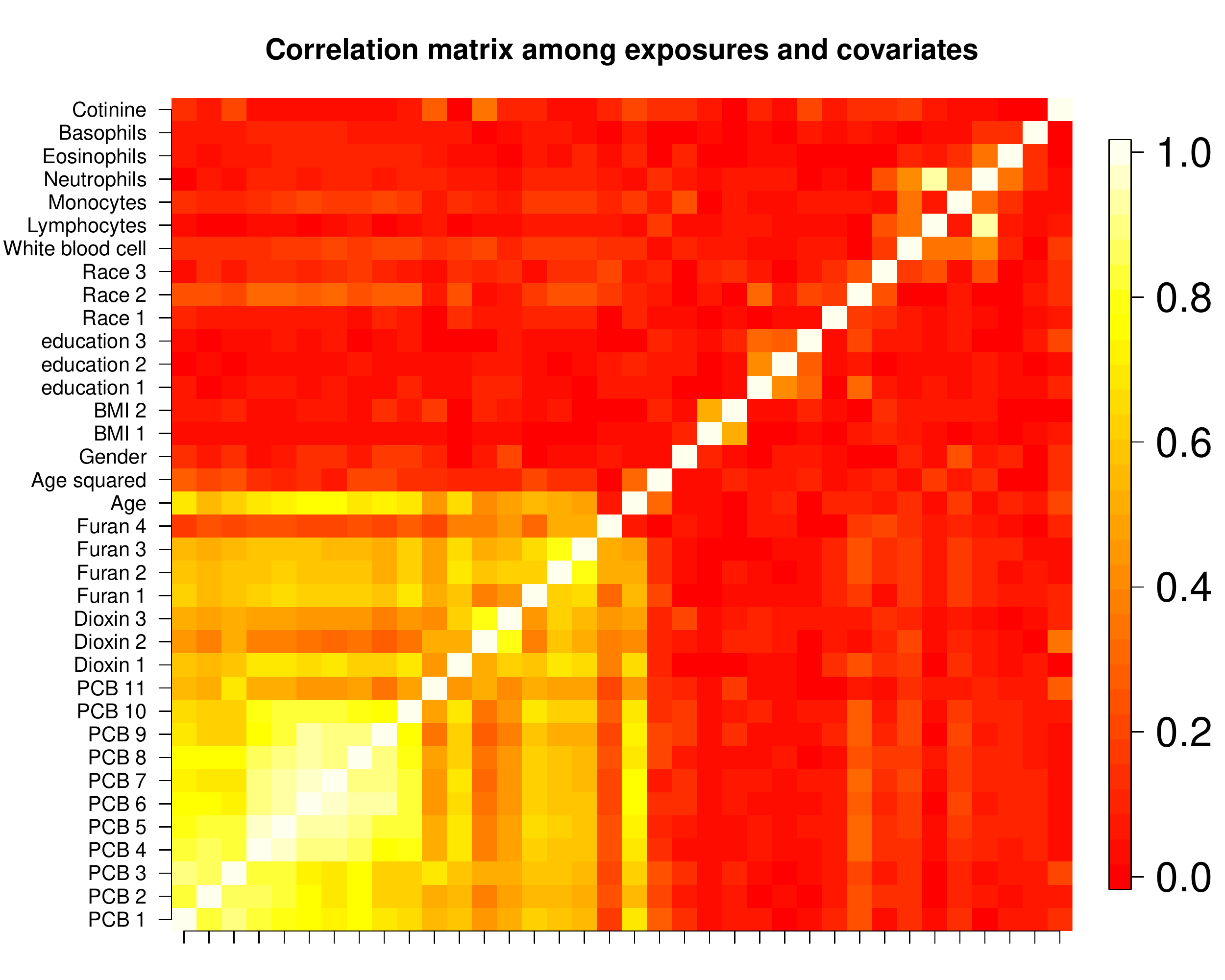}
		\caption{Correlation matrix among the 18 exposures and 18 demographic variables used in the final analysis for the NHANES study.}
		\label{fig:CorrNHANES}
\end{figure}


As discussed in Section \ref{NHANES}, when we fit the SSGL model to this data set, we identified four main effects (plotted in Figure \ref{fig:MainEffectNHANES}). Our model also identified six interactions as having nonzero parameter estimates. The identified interactions are PCB 10 - PCB 7, Dioxin 1 - PCB 11, Dioxin 2 - PCB 2, Dioxin 2 - Dioxin 1, Furan 1 - PCB 10, and Furan 4 - Furan 3. We see that there are interactions both within a certain type of pollutant (Dioxin and Dioxin, etc.) and across pollutant types (Furan and PCB).

Lastly, looking at Figure \ref{fig:MainEffectNHANES}, we can see that the exposure response curves for the four identified main effects are relatively linear, particularly for PCB 11 and Furan 1. With this in mind, we also ran our SSGL model with one degree of freedom splines for each main effect. Note that this does not require a model that handles grouped covariate structures as the main effects and interactions in this case are both single parameters. Cross-validated error from the model with one degree of freedom is nearly identical to the model with two degrees of freedom, though the linear model selects far more terms. The linear model selects six main effect terms and 20 interaction terms. As the two models provide similar predictive performance but the model with two degrees of freedom is far more parsimonious, we elect to focus on the model with two degrees of freedom. 


	\section{Proofs of Main Results} \label{App:D}
	
	\subsection{Preliminary Lemmas} \label{App:D1}
	Before proving the main results in the paper, we first prove the following lemmas.
	
	\begin{lemma} \label{auxlemma1}
		Suppose that $\betab_g \in \R^{m_g}$ follows a group lasso density indexed by $\lambda$, i.e. $\betab_g \sim \bm{\Psi} ( \betab_g \vert \lambda )$. Then
		\begin{equation*} 
		\mathbb{E}( \lVert \bm{\beta}_g \rVert_2^2 ) = \frac{m_g (m_g+1)}{\lambda^2}.
		\end{equation*}
	\end{lemma}
	\begin{proof}
		The group lasso density, $\bm{\Psi}(\bm{\beta}_g \vert \lambda )$, is the marginal density of a scale mixture,
		\begin{equation*}
		\bm{\beta}_g \sim \N_{m_g} ( \zerob, \tau \bm{I}_{m_g} ), \hspace{.3cm} \tau \sim \mathcal{G} \left( \frac{m_g+1}{2}, \frac{\lambda^2}{2} \right).
		\end{equation*}
		Therefore, using iterated expectations, we have
		\begin{align*}
		\mathbb{E}( \lVert \betab_g \rVert_2^2 ) & = \mathbb{E}\left[ \mathbb{E}( \lVert \betab_g \rVert_2^2 \hspace{.1cm} \vert \hspace{.1cm} \tau ) \right] \\
		& = m_g \mathbb{E}( \tau ) \\
		& = \frac{m_g (m_g+1) }{\lambda^2}.
		\end{align*}
	\end{proof}
	
	\begin{lemma} \label{auxlemma2}
		Suppose $\sigma^2 > 0, \sigma_0^2 > 0$. Then for any $\epsilon_n \in (0, 1)$ such that $\epsilon_n \rightarrow 0$ as $n \rightarrow \infty$, we have for sufficiently large $n$,
		\begin{align*}
		\left\{ \lvert \sigma^2 - \sigma_0^2 \rvert \geq 4 \sigma_0^2 \epsilon_n \right\} \subseteq \left\{ \frac{\sigma^2}{\sigma_0^2} > \frac{1 - \epsilon_n}{1 - \epsilon_n} \textrm{ or } \frac{\sigma^2}{\sigma_0^2} < \frac{1 - \epsilon_n}{1 + \epsilon_n} \right\}.
		\end{align*}
	\end{lemma}
	\begin{proof}
		For large $n$, $\epsilon_n < 1/2$, so $2 \epsilon_n / (1-\epsilon_n) < 4 \epsilon_n, - 2 \epsilon_n / (1 + \epsilon_n) > -4 \epsilon_n$, and thus,
		\begin{align*}
		& \lvert \sigma^2 - \sigma_0^2 \rvert \geq 4 \sigma_0^2 \epsilon_n \Rightarrow ( \sigma^2 - \sigma_0^2)/ \sigma_0^2 \geq 4 \epsilon_n \textrm{ or } (\sigma^2 - \sigma_0^2)/\sigma_0^2 \leq -4 \epsilon_n \\
		& \qquad \qquad \Rightarrow \frac{\sigma^2}{\sigma_0^2} - 1 > \frac{2 \epsilon_n}{1 - \epsilon_n} \textrm{ or } \frac{\sigma^2}{\sigma_0^2} - 1 < - \frac{2 \epsilon_n}{1 + \epsilon_n} \\
		& \qquad \qquad \Rightarrow \frac{\sigma^2}{\sigma_0^2} > \frac{1 + \epsilon_n}{1 - \epsilon_n} \textrm{ or } \frac{\sigma^2}{\sigma_0^2} < \frac{1 - \epsilon_n}{1 + \epsilon_n},
		\end{align*}
		and hence,
		\begin{align*}
		\lvert \sigma^2 - \sigma_0^2 \rvert \geq 4 \sigma_0^2 \epsilon_n \hspace{.2cm} \Rightarrow \hspace{.2cm}  \frac{\sigma^2}{\sigma_0^2} > \frac{1 + \epsilon_n}{1 - \epsilon_n} \textrm{ or } \frac{\sigma^2}{\sigma_0^2} < \frac{1 - \epsilon_n}{1 + \epsilon_n}.
		\end{align*}.
	\end{proof}
	
	\begin{lemma} \label{auxlemma3}
		Suppose that a vector $\bm{z} \in \R^m$ can be decomposed into subvectors, $\bm{z} = [ \bm{z}_1', \ldots, \bm{z}_d' ]$, where $\sum_{i=1}^{d} \lvert \bm{z}_i \rvert = m$ and $\lvert \bm{z}_i \rvert$ denotes the length of $\bm{z}_i$. Then $\lVert \bm{z} \rVert_2 \leq \sum_{i=1}^{d} \lVert \bm{z}_i \rVert_2$.
		\begin{proof}
			We have
			\begin{align*}
			\lVert \bm{z} \rVert_2 & = \sqrt{ z_{11}^2 + \ldots + z_{1 \lvert z_1 \rvert}^2 + \ldots + z_{d1}^2 + \ldots + z_{d \lvert z_d \rvert}  } \\
			& \leq \sqrt{ z_{11}^2 + \ldots + z_{1 \lvert z_1 \rvert}^2 } + \ldots + \sqrt{ z_{d1}^2 + \ldots + z_{d \lvert z_d \rvert} } \\
			& = \lVert \bm{z}_1 \rVert_2 + \ldots + \lVert \bm{z}_d \rVert_2.
			\end{align*}
		\end{proof}
	\end{lemma}
	
	\subsection{Proofs for Section 3} \label{App:D2}
	\begin{proof}[Proof of Proposition \ref{globalmodeseparable}]
		This result follows from an adaptation of the arguments of \citet{ZZ12}. The group-specific optimization problem is:
		\begin{align}
		\widehat{\betab}_g = \argmax_{\betab_g}\left\{ -\frac{1}{2}\lVert \zb_g-\betab_g\rVert_2^2+ \sigma^2pen_S(\betab|\theta) \right\}. \label{groupwise}
		\end{align}
		We first note that the optimization problem \eqref{groupwise} is equivalent to maximizing the objective
		\begin{align}
		L(\betab_g) &= -\frac{1}{2}\lVert \zb_g - \betab_g\rVert_2^2+ \sigma^2pen_S(\betab|\theta) + \frac{1}{2} \lVert \zb_g\rVert_2^2 \\
		&= \lVert \betab_g\rVert_2\left[\frac{\betab_g^T\zb_g}{\lVert \betab\rVert_2} - \left(\frac{\lVert \betab_g\rVert_2}{2} - \frac{\sigma^2pen_S(\betab|\theta)}{\lVert \betab_g\rVert_2}\right) \right] \\
		&= \lVert \betab_g\rVert_2\left[\lVert\zb_g\rVert_2\cos\varphi - \left(\frac{\lVert \betab_g\rVert_2}{2} - \frac{\sigma^2pen_S(\betab|\theta)}{\lVert \betab_g \rVert_2}\right) \right] \label{delta_proof}
		\end{align}
		where $\varphi$ is the angle between $\zb_g$ and $\betab_g$. Then, when $\lVert \zb_g\rVert_2 < \Delta$, the second factorized term of \eqref{delta_proof} is always less than zero, and so $\widehat{\betab}_g = \zerob_{m_g}$ must be the global maximizer of $L$. On the other hand, when the global maximizer $\widehat{\betab}_g = \zerob_{m_g}$, then the second factorized term must always be less than zero, otherwise $\widehat{\betab}_g = \zerob_{m_g}$ would no longer be the global maximizer and so $\lVert \zb_g\rVert_2 < \Delta$. 
	\end{proof}
	
	\begin{proof}[Proof of Lemma \ref{theta_mean_lemma}]
		We have
		\begin{align}
		\mathbb{E}[\theta |\widehat{\betab}] = \frac{\int_0^1 \theta^a (1-\theta)^{b-1}(1-\theta z)^{G-\widehat{q} }\prod_{g=1}^{\widehat{q}} (1-\theta x_g) d\theta}{ \int_0^1 \theta^{a-1} (1-\theta)^{b-1}(1-\theta z)^{G-\widehat{q} }\prod_{g=1}^{\widehat{q}} (1-\theta x_g)d\theta }. \label{thetaexpectation}
		\end{align}
		When $\lambda_0 \to \infty$, we have $z \to 1$ and $x_g \to -\infty$ for all $g = 1,\dots, \widehat{q}$. Hence,
		\begin{align}
		\lim_{\lambda_0\to \infty} \mathbb{E}[\theta|\widehat{\betab}] &= \lim_{z\to 1}\lim_{x_g\to -\infty}\frac{\int_0^1 \theta^a (1-\theta)^{b + G-\widehat{q}-1}\prod_{g=1}^{\widehat{q}} (1-\theta x_g)}{ \int_0^1 \theta^{a-1} (1-\theta)^{b-1}(1-\theta z)^{G-\widehat{q} }\prod_{g=1}^{\widehat{q}} (1-\theta x_g) } \\
		&=\frac{ \int_0^1 \theta^{a + \widehat{q}}(1-\theta)^{b + G - \widehat{q}-1}d\theta }{\int_0^1 \theta^{a + \widehat{q} - 1}(1-\theta)^{b + G - \widehat{q} - 1}d\theta} \\
		&= \frac{a + \widehat{q}}{a + b+ G }.
		\end{align}
		
	\end{proof}

	\subsection{Proofs for Section 6} \label{App:D3}
	In this section, we use proof techniques from \cite{NingGhosal2018, SongLiang2017, WeiReichHoppinGhosal2018} rather than the ones in \cite{RockovaGeorge2018}. However, none of these other papers considers \textit{both} continuous spike-and-slab priors for groups of regression coefficients \textit{and} an independent prior on the unknown variance. 
	
	\begin{proof}[Proof of Theorem \ref{posteriorcontractiongroupedregression}]
		Our proof is based on first principles of verifying Kullback-Leibler (KL) and testing conditions (see e.g., \cite{GhosalGhoshVanDerVaart2000}). We first prove (\ref{l2contraction}) and (\ref{varianceconsistency}).
		\vspace{.3cm}
		
		\noindent \textbf{Part I: Kullback-Leibler conditions.} Let $f \sim \N_n ( \Xb \betab, \sigma^2 \bm{I}_n) $ and $f_0 \sim \N_n (\Xb \betab_0, \sigma_0^2 \bm{I}_n )$, and let $\Pi(\cdot)$ denote the prior (\ref{hiermodel}). We first show that for our choice of $\epsilon_n = \sqrt{s_0 \log G / n}$,
		\begin{equation} \label{KullbackLeiblercond}
		\Pi \left( K (f_0, f) \leq n \epsilon_n^2, V(f_0, f) \leq n \epsilon_n^2 \right) \geq \exp (-C_1 n \epsilon_n^2),
		\end{equation}
		for some constant $C_1 > 0$, where $K(\cdot, \cdot)$ denotes the KL divergence and $V(\cdot, \cdot)$ denotes the KL variation. The KL divergence between $f_0$ and $f$ is
		\begin{equation} \label{KLdiv}
		K(f_0, f) = \frac{1}{2} \left[ n \left( \frac{\sigma_0^2}{\sigma^2} \right) - n - n \log \left( \frac{\sigma_0^2}{\sigma^2} \right) + \frac{ \lVert \Xb ( \betab - \betab_0 ) \rVert_2^2}{\sigma^2}   \right],
		\end{equation}
		and the KL variation between $f_0$ and $f$ is
		\begin{equation} \label{KLvar}
		V(f_0, f) = \frac{1}{2} \left[ n \left( \frac{\sigma_0^2}{\sigma^2} \right)^2 - 2n \left( \frac{\sigma_0^2}{\sigma^2} \right) + n  \right] + \frac{\sigma_0^2}{(\sigma^2)^{2}} \lVert \Xb ( \betab - \betab_0 ) \rVert_2^2.
		\end{equation}
		Define the two events $\mathcal{A}_1$ and $\mathcal{A}_2$ as follows:
		\begin{equation} \label{eventA1}
		\mathcal{A}_1 = \left\{ \sigma^2: n \left( \frac{\sigma_0^2}{\sigma^2} \right) - n - n \log \left( \frac{\sigma_0^2}{\sigma^2} \right) \leq n \epsilon_n^2, \right. \\
		\left. n \left( \frac{\sigma_0^2}{\sigma^2} \right)^2 - 2n \left( \frac{\sigma_0^2}{\sigma^2} \right) + n \leq n \epsilon_n^2  \right\}
		\end{equation}
		and
		\begin{equation} \label{eventA2}
		\mathcal{A}_2 = \left\{ (\betab, \sigma^2): \frac{ \lVert \bm{X} ( \bm{\beta} - \bm{\beta}_0 ) \rVert_2^2}{\sigma^2} \leq n \epsilon_n^2, \right. \\
		\left. \frac{\sigma_0^2}{(\sigma^2)^{2}} \lVert \Xb ( \betab - \betab_0 ) \rVert_2^2 \leq n \epsilon_n^2/2 \right\}.
		\end{equation}
		Following from (\ref{KullbackLeiblercond})-(\ref{eventA2}), we may write $\Pi ( K(f_0, f) \leq \epsilon_n^2, V(f_0, f) \leq \epsilon_n^2 ) = \Pi ( \mathcal{A}_2 \vert \mathcal{A}_1 ) \Pi ( \mathcal{A}_1)$. We derive lower bounds for $\Pi(\mathcal{A}_1)$ and $\Pi (\mathcal{A}_2 \vert \mathcal{A}_1)$ separately. Noting that we may rewrite $\mathcal{A}_1$ as 
		\begin{align*}
		\mathcal{A}_1 = \left\{ \sigma^2: \frac{\sigma_0^2}{\sigma^2} - 1 - \log \left( \frac{\sigma_0^2}{\sigma^2} \right)  \leq \epsilon_n^2, \hspace{.3cm}  \left( \frac{\sigma_0^2}{\sigma^2} - 1 \right)^2  \leq \epsilon_n^2 \right\},
		\end{align*}
		and expanding $\log(\sigma_0^2 / \sigma^2)$ in the powers of $1-\sigma_0^2/\sigma^2$ to get $\sigma_0^2 / \sigma^2-1-\log(\sigma_0^2/\sigma^2) \sim (1-\sigma_0^2/\sigma^2)^2/2$, it is clear that $\mathcal{A}_1 \supset \mathcal{A}_1^{\star}$, where $\mathcal{A}_1^{\star} = \{ \sigma^2: \sigma_0^2 / ( \epsilon_n + 1) \leq \sigma^2 \leq \sigma_0^2 \}$. Thus, since $\sigma^2 \sim \mathcal{IG}(c_0, d_0)$, we have for sufficiently large $n$,
		\begin{align*} \label{A1lowerbound}
		\Pi(\mathcal{A}_1) \geq \Pi(\mathcal{A}_1^{\star}) & \asymp \displaystyle \int_{\sigma_0^2/( \epsilon_n + 1)}^{\sigma_0^2} (\sigma^2)^{-c_0-1} e^{-d_0/ \sigma^2} d \sigma^2 \\
		& \geq  (\sigma_0^2)^{-c_0-1} e^{-d_0 (\epsilon_n + 1) / \sigma_0^2} . \numbereqn
		\end{align*}
		Thus, from (\ref{A1lowerbound}), we have
		\begin{equation} \label{neglogA1upper}
		- \log \Pi (\mathcal{A}_1) \lesssim \epsilon_n + 1 \lesssim n \epsilon_n^2,
		\end{equation}
		since $n\epsilon_n^2 \rightarrow \infty$. Next, we consider $\Pi (\mathcal{A}_2 \vert \mathcal{A}_1)$. We have
		\begin{align*}
		\frac{\sigma_0^2}{(\sigma^2)^2} \lVert \Xb ( \betab - \betab_0) \rVert_2^2 & = \bigg| \bigg| \frac{ \Xb ( \betab - \betab_0 )}{\sigma} \bigg| \bigg|_2^2 \left( \frac{\sigma_0^2}{\sigma^2} - 1 \right) +  \bigg| \bigg| \frac{ \Xb ( \betab - \betab_0 )}{\sigma} \bigg| \bigg|_2^2,
		\end{align*}
		and conditional on $\mathcal{A}_1$, we have that the previous display is bounded above by
		\begin{align*}
		\bigg| \bigg| \frac{ \Xb ( \betab - \betab_0 )}{\sigma} \bigg| \bigg|_2^2 \left( \epsilon_n + 1 \right) < \frac{2}{\sigma^2} \lVert \Xb ( \betab - \betab_0 ) \rVert_2^2 ,
		\end{align*}
		for large $n$ (since $\epsilon_n < 1$ when $n$ is large). Since $\mathcal{A}_1 \supset \mathcal{A}_1^{\star}$, where $\mathcal{A}_1^{\star}$ was defined earlier, the left-hand side of both expressions in \eqref{eventA2} can be bounded above by a constant multiple of $\lVert \Xb (\betab-\betab_0) \rVert_2^2$, conditional on $\mathcal{A}_1$. Therefore, for some constant $b_1 > 0$, $\Pi( \mathcal{A}_2 \vert \mathcal{A}_1)$ is bounded below by
		\begin{align*} \label{A2givenA1LowerBound1}
		& \Pi (\mathcal{A}_2 \vert \mathcal{A}_1 ) \geq \Pi \left( \lVert \Xb (\betab - \betab_0 ) \rVert_2^2  \leq \frac{b_1^2 n \epsilon_n^2}{2} \right) \\
		& \geq \Pi \left( \lVert \betab - \betab_0 \rVert_2^2 \leq \frac{b_1^2 \epsilon_n^2}{2 k n^{\alpha-1} } \right) \\
		& \geq \int_{0}^{1} \left \{ \Pi_{S_0} \left( \lVert \betab_{S_0} - \betab_{0 S_0} \rVert_2^2 \leq \frac{ b_1^2 \epsilon_n^2}{ 4 k n^{\alpha}  } \bigg| \theta \right) \right\} \left\{ \Pi_{S_0^c} \left( \lVert \bm{\beta}_{S_0^c} \rVert_2^2 \leq \frac{ b_1^2 \epsilon_n^2}{ 4 k n^{\alpha} } \bigg| \theta \right) \right\} d \pi(\theta), \numbereqn
		\end{align*}
		where we used Assumption \ref{A3} in the second inequality, and in the third inequality, we used the fact that conditional on $\theta$, the SSGL prior is separable, so $\pi(\betab | \theta )= \pi_{S_0}(\betab | \theta ) \pi_{S_0^c}(\betab | \theta )$. We proceed to lower-bound each bracketed integrand term in (\ref{A2givenA1LowerBound1}) separately. Changing the variable $\betab - \betab_0 \rightarrow \bm{b}$ and using the fact that $\pi_{S_0} (\betab | \theta ) > \theta^{s_0} \prod_{g \in S_0} \left[ C_g \lambda_1^{m_g} e^{-\lambda_1 \lVert \bm{\beta}_g \rVert_2} \right]$ and $\lVert \bm{z} \rVert_2 \leq \lVert \bm{z} \rVert_1$ for any vector $\bm{z}$, we have as a lower bound for the first term in (\ref{A2givenA1LowerBound1}),
		\begin{align*} \label{A2givenA1LowerBound2}
		&  \theta^{s_0}  e^{-\lambda_1 \lVert \bm{\beta}_{S_0} \rVert_2 } \prod_{g \in S_0} C_g \left\{ \displaystyle \int_{ \lVert \bm{b}_{g} \rVert_1 \leq \frac{b_1 \epsilon_n }{ 2 s_0 \sqrt{k n^{\alpha} } } } \lambda_1^{m_g} e^{-\lambda_1 \lVert \bm{b}_{g} \rVert_1 } d \bm{b}_g \right\}. \numbereqn
		\end{align*}
		Each of the integral terms in (\ref{A2givenA1LowerBound2}) is the probability of the first $m_g$ events of a Poisson process happening before time $b_1 \epsilon_n / 2 s_0 \sqrt{k n^{\alpha}}$. Using similar arguments as those in the proof of Lemma 5.1 of \cite{NingGhosal2018}, we obtain as a lower bound for the product of integrals in (\ref{A2givenA1LowerBound2}),
		\begin{align*} \label{A2givenA1LowerBound3}
		& \displaystyle \prod_{g \in S_0} C_g \left\{ \displaystyle \int_{ \lVert \bm{b}_{g} \rVert_1 \leq \frac{b_1 \epsilon_n }{ 2 s_0 \sqrt{ k n^{\alpha} } } }   \lambda_1^{m_g} e^{-\lambda_1 \lVert \bm{b}_{g} \rVert_1 } d \bm{b}_g \right\} \\
		& \qquad \geq \displaystyle \prod_{g \in S_0} C_g e^{-\lambda_1 b_1 \epsilon_n / 2 s_0 \sqrt{ k n^{\alpha}}} \frac{1}{m_g !} \left( \frac{  \lambda_1 b_1 \epsilon_n}{ s_0 \sqrt{  k n^{\alpha}}} \right)^{m_g} \\
		& \qquad = e^{- \lambda_1 b_1 \epsilon_n / 2 \sqrt{ k n^{\alpha}}} \displaystyle \prod_{g \in S_0} \frac{C_g}{m_g !} \left( \frac{  \lambda_1 b_1 \epsilon_n}{ s_0  \sqrt{  k n^{\alpha} }} \right)^{m_g}. \numbereqn
		\end{align*}
		Combining (\ref{A2givenA1LowerBound2})-(\ref{A2givenA1LowerBound3}), we have the following lower bound for the first bracketed term in (\ref{A2givenA1LowerBound1}):
		\begin{align} \label{A2givenA1LowerBound5}
		\theta^{s_0} e^{-\lambda_1 \lVert \bm{\beta}_{S_0} \rVert_2 } e^{- \lambda_1 b_1 \epsilon_n / 2 \sqrt{ k n^{\alpha}}} \displaystyle \prod_{g \in S_0} \frac{C_g}{m_g !} \left( \frac{ \lambda_1 b_1 \epsilon_n}{ s_0 \sqrt{ k n^{\alpha}}} \right)^{m_g}.
		\end{align}
		Now, noting that $\pi_{S_0^c} ( \bm{\beta} | \theta ) > (1-\theta)^{G-s_0} \prod_{g \in S_0^c} \left[ C_g \lambda_0^{m_g} e^{-\lambda_0 \lVert \betab_g \rVert_2} \right]$, we further bound the second bracketed term in (\ref{A2givenA1LowerBound1}) from below. Let $\Check{\pi} ( \cdot )$ denote the density, $\Check{\pi} (\betab_g) =  C_g \lambda_0^{m_g} e^{-\lambda_0 \lVert \betab_g \rVert_2}$. We have
		\begin{align*} \label{A2givenA1LowerBound6}
		&\Pi_{S_0^c} \left( \lVert \bm{\beta}_{S_0^c} \rVert_2^2 \leq \frac{ b_1^2 \epsilon_n^2}{ 4 k n^{\alpha} } \bigg| \theta \right) > (1 - \theta)^{G-s_0} \displaystyle \prod_{g \in S_0^c} \Check{\pi} \left(  \lVert \betab_g \rVert_2^2 \leq \frac{ b_1^2 \epsilon_n^2}{4 k n^{\alpha}  (G-s_0)} \right) \\
		& \qquad \qquad \geq (1 - \theta)^{G-s_0} \displaystyle \prod_{g \in S_0^c} \left[ 1 - \frac{4 k n^{\alpha}  (G- s_0) \mathbb{E}_{\Check{\Pi}} \left( \lVert \bm{\beta}_g \rVert_2^2 \right) }{ b_1^2 \epsilon_n^2} \right] \\
		& \qquad \qquad = (1-\theta)^{G-s_0} \displaystyle \prod_{g \in S_0^c} \left[ 1 - \frac{4 k n^{\alpha} (G- s_0)  m_g (m_g+1) }{ \lambda_0^2 b_1^2 \epsilon_n^2 } \right] \\
		& \qquad \qquad \geq (1-\theta)^{G-s_0} \left[ 1 - \frac{4 k n^{\alpha} G  m_{\max} (m_{\max} + 1)}{ \lambda_0^2 b_1 \epsilon_n^2 } \right]^{G-s_0}, \numbereqn 
		\end{align*}
		where we used an application of the Markov inequality and Lemma \ref{auxlemma1} in the second line. Combining \eqref{A2givenA1LowerBound5}-\eqref{A2givenA1LowerBound6} gives as a lower-bound for \eqref{A2givenA1LowerBound1},
		\begin{align*} \label{A2givenA1LowerBound7}
		 \Pi ( \mathcal{A}_2 | \mathcal{A}_1 ) & \geq \left\{ e^{-\lambda_1 \lVert \bm{\beta}_{S_0} \rVert_2} e^{-\lambda_1 b_1 \epsilon_n / 2 \sqrt{k n^{\alpha}} } \prod_{g \in S_0} \frac{C_g}{m_g!} \left( \frac{\lambda_1 b_1 \epsilon_n}{s_0 \sqrt{k n^{\alpha}}} \right)^{m_g} \right\} \\
		 & \times \left\{ \int_{0}^{1} \theta^{s_0} (1-\theta)^{G-s_0} \left[ 1 - \frac{4 k n^{\alpha} G m_{\max} (m_{\max}+1)}{\lambda_0^2 b_1 \epsilon_n^2}  \right]^{G-s_0} d \pi(\theta)  \right\} \numbereqn
		\end{align*}
		Let us consider the second bracketed term in \eqref{A2givenA1LowerBound7} first. By assumption, $\lambda_0 = (1-\theta) / \theta$. Further, $\lambda_0^2 = (1-\theta)^2 / \theta^2$ is monotonically decreasing in $\theta$ for $\theta \in (0, 1)$. Hence, for constant $c > 2$ in the $\mathcal{B} ( 1, G^c)$ prior on $\theta$, a lower bound for the second bracketed term in \eqref{A2givenA1LowerBound7} is
		\begin{align*} \label{A2givenA1LowerBound8}
		    & \int_{1/(2G^c+1)}^{1/(G^c+1)} \theta^{s_0} (1-\theta)^{G-s_0} \left[ 1 - \frac{4 k n^{\alpha} G m_{\max} (m_{\max} +1 )}{\lambda_0^2 b_1 \epsilon_n^2} \right]^{G-s_0} d \pi(\theta) \\
		    & \geq (2G^c + 1)^{-s_0} \left[ 1 - \frac{ 4 k n^{\alpha} G m_{\max} (m_{\max} + 1)}{G^{2c} b_1 \epsilon_n^2} \right]^{G-s_0} \int_{1/(2G^c+1)}^{1/(G^c+1)} (1-\theta)^{G-s_0} d \pi (\theta) \\
		    & \gtrsim (2 G^c+1)^{-s_0} \left[ 1 - \frac{1}{G-s_0} \right]^{G-s_0} \int_{1/(2G^c+1)}^{1/(G^c+1)} (1-\theta)^{G- s_0} d \pi(\theta) \\
		    & \asymp (2G^c + 1)^{-s_0} G^{-c} \int_{1/(2G^c+1)}^{1/(G^c+1)} (1-\theta)^{G^c+G-s_0-1} d \theta \\
		    & = (2 G^c + 1)^{-s_0} G^{-c} (G^c + G - s_0)^{-1} \\
		    & \qquad \qquad \times\left[ \left( 1 - \frac{1}{2G^c +1} \right)^{G^c + G - s_0} - \left( 1 - \frac{1}{G^c + 1} \right)^{G^c + G - s_0} \right] \\
		    & \gtrsim (2G^{c} + 1)^{-s_0} G^{-c} (G^c + G - s_0)^{-1}, \numbereqn
		\end{align*}
		where in the third line, we used our assumptions about the growth rates for $m_{\max}$, $G$, and $s_0$ in Assumptions \ref{A1}-\ref{A2} and the fact that $c > 2$. In the fourth line, we used the fact that $(1- 1/x)^x \rightarrow e^{-1}$ as $x \rightarrow \infty$ and $\theta \sim \mathcal{B}(1, G^{c})$. In the sixth line, we used the fact that the bracketed term in the fifth line can be bounded below by $e^{-1} - e^{-2}$ for sufficiently large $n$. 
		
		Combining \eqref{A2givenA1LowerBound7}-\eqref{A2givenA1LowerBound8}, we obtain for sufficiently large $n$,
		\begin{align*} \label{neglogpiA2givenA1}
		- \log \Pi ( \mathcal{A}_2 \vert \mathcal{A}_1 ) \lesssim  & \hspace{.2cm} \lambda_1 \lVert \bm{\beta}_{0 S_0} \rVert_2 + \frac{ \lambda_1 b_1 \epsilon_n }{2 \sqrt{ k n^{\alpha}}} + \displaystyle \sum_{g \in S_0} \log (m_g !) - \displaystyle \sum_{g \in S_0} \log C_g  \\
		& + \displaystyle \sum_{g \in S_0} m_g  \log \left( \frac{s_0 \sqrt{ k n^{\alpha}}}{\lambda_1 b_1 \epsilon_n} \right) + s_0 \log(2G^{c} + 1) + c \log G \\
		& + \log (G^c + G - s_0 )  \numbereqn
		\end{align*}
		We examine each of the terms in (\ref{neglogpiA2givenA1}) separately. By Assumptions \ref{A1} and \ref{A5} and the fact that $\lambda_1 \asymp 1/n$, we have
		\begin{align*}
		\lambda_1 \lVert \bm{\beta}_{0 S_0} \rVert_2 \leq \lambda_1 \sqrt{s_0 m_{\max}} \lVert \bm{\beta}_{0 S_0} \rVert_{\infty} \lesssim  s_0  \log G \lesssim n \epsilon_n^2,
		\end{align*}
		and
		\begin{align*}
		\frac{\lambda_1 b_1 \epsilon_n}{2 \sqrt{ k n^{\alpha}}} \lesssim \epsilon_n \lesssim n \epsilon_n^2.
		\end{align*}
		Next, using the facts that $x! \leq x^x$ for $x \in \mathbb{N}$ and Assumption \ref{A1}, we have
		\begin{align*}
		\displaystyle \sum_{g \in S_0} \log (m_g !) \leq s_0 \log(m_{\max} !) \leq s_0 m_{\max} \log (m_{\max}) \leq s_0 m_{\max} \log n \lesssim n \epsilon_n^2.
		\end{align*}
		Using the fact that the normalizing constant, $C_g = 2^{-m_g} \pi^{-(m_g - 1)/2} [ \Gamma ((m_g+1)/2) ]^{-1}$, we also have
		\begin{align*}
		\displaystyle \sum_{g \in S_0} - \log C_g = & \displaystyle \sum_{g \in S_0} \left\{ m_g \log 2 + \left( \frac{m_g - 1}{2} \right) \log \pi + \log \left[ \Gamma \left( \frac{m_g+1}{2} \right) \right] \right\} \\
		& \leq s_0 m_{\max}( \log 2 + \log \pi ) + \displaystyle \sum_{g \in S_0}  \log ( m_g ! ) \\
		& \lesssim s_0 m_{\max}( \log 2 + \log \pi ) + s_0 m_{\max} \log n \\
		&\lesssim s_0 \log G \\
		& \lesssim n \epsilon_n^2,
		\end{align*}
		where we used the fact that $\Gamma ( (m_g+1)/2 ) \leq \Gamma(m_g + 1) = m_g !$. Finally, since $\lambda_1 \asymp 1/n$ and using Assumption \ref{A1} that $m_{\max} = O(\log G / \log n)$, we have
		\begin{align*}
		\displaystyle \sum_{g \in S_0} m_g  \log \left( \frac{s_0 \sqrt{ k n^{\alpha} }}{\lambda_1 b_1 \epsilon_n} \right) & \lesssim \displaystyle s_0 m _{\max}  \log \left( \frac{s_0 n^{ \alpha /2+1} \sqrt{k}}{ b_1 \epsilon_n^2} \right) \\
		& = s_0 m_{\max} \log \left( \frac{n^{\alpha/2+2} \sqrt{k}}{b_1 \log G} \right) \\
		& \lesssim s_0 m_{\max} \log n \\
		& \lesssim s_0 \log G \\
		& \lesssim n \epsilon_n^2.
		\end{align*}
		Finally, it is clear by the definition of $n \epsilon_n^2$ and the fact that $c > 2$ is a constant that
		\begin{align*}
		   s_0 \log (2G^c + 1) + c \log G + \log (G^c + G - s_0 ) \asymp s_0 \log G = n \epsilon_n^2.
		\end{align*}
		Combining all of the above, together with (\ref{neglogpiA2givenA1}), we have 
		\begin{equation} \label{neglogpiA2givenA1pt2}
		- \log \Pi ( \mathcal{A}_2 \vert \mathcal{A}_1 ) \lesssim n \epsilon_n^2.
		\end{equation}
		By (\ref{neglogA1upper}) and (\ref{neglogpiA2givenA1pt2}), we may choose a large constant $C_1 > 0$, so that 
		\begin{equation*}
		\Pi( \mathcal{A}_2 \vert \mathcal{A}_1 ) \Pi (\mathcal{A}_1) \gtrsim \exp(- C_1 n \epsilon_n^2 / 2) \exp(- C_1 n \epsilon_n^2 / 2) = \exp(-C_1 n \epsilon_n^2),
		\end{equation*}
		so the Kullback-Leibler condition (\ref{KullbackLeiblercond}) holds.
		\vspace{.5cm}
		
		\noindent \textbf{Part II: Testing conditions.} To complete the proof, we show the existence of a sieve $\mathcal{F}_n$ such that
		\begin{equation} \label{testingcond1}
		\Pi( \mathcal{F}_n^c) \leq \exp(-C_2 n \epsilon_n^2),
		\end{equation}
		for positive constant $C_2 > C_1+2$, where $C_1$ is the constant from (\ref{KullbackLeiblercond}), and a sequence of test functions $\phi_n \in [0,1]$ such that
		\begin{equation} \label{testingcond2}
		\mathbb{E}_{f_0} \phi_n \leq e^{-C_4 n \epsilon_n^2},
		\end{equation}
		and 
		\begin{equation} \label{testingcond3}
		\displaystyle \sup_{ \begin{array}{rl} f \in \mathcal{F}_n: & \lVert \betab - \betab_0 \rVert_2 \geq (3+\sqrt{\nu_1}) \sigma_0 \epsilon_n, \\ & \textrm{ or } \lvert \sigma^2 - \sigma_0^2 \rvert \geq 4 \sigma_0^2 \epsilon_n \end{array} } \mathbb{E}_f (1 - \phi_n) \leq e^{-C_4 n \epsilon_n^2},
		\end{equation}
		for some $C_4 > 0$, where $\nu$ is from Assumption \ref{A4}. Recall that $\omega_g \equiv \omega_g(\lambda_0, \lambda_1, \theta) = \frac{1}{\lambda_0 - \lambda_1} \log \left[ \frac{1-\theta}{\theta} \frac{\lambda_0^{m_g}}{\lambda_1^{m_g}} \right]$. Choose $C_3 \geq C_1+2+\log 3$, and consider the sieve,
		\begin{equation} \label{sievedef}
		\mathcal{F}_n = \left\{ f: \lvert \bm{\gamma} (\betab) \rvert \leq C_3 s_0, 0 < \sigma^2 \leq G^{C_3 s_0 / c_0} \right\},
		\end{equation}
		where $c_0$ is from $\mathcal{IG}(c_0, d_0)$ prior on $\sigma^2$ and $\lvert \bm{\gamma} (\betab) \rvert$ denotes the generalized dimensionality (\ref{generalizeddimensionality}).
		
		We first verify (\ref{testingcond1}). We have
		\begin{equation} \label{sievecomplement1}
		\Pi (\mathcal{F}_n^c)  \leq \Pi \left( \lvert \bm{\gamma} (\betab) \rvert > C_3 s_0 \right) + \Pi \left( \sigma^2 > G^{C_3 s_0 /c_0} \right).
		\end{equation}
		We focus on bounding each of the terms in (\ref{sievecomplement1}) separately. First, let $\theta_0 = C_3 s_0 \log G / G^c$, where $c > 2$ is the constant in the $\mathcal{B}(1, G^c)$ prior on $\theta$. Similarly as in the proof of Theorem 6.3 in \citet{RockovaGeorge2018}, we have $\pi(\betab_g | \theta) < 2 \theta C_g \lambda_1^{m_g} e^{- \lambda_1 \lVert \betab_g \rVert_2}$ for all $\lVert \betab_g \rVert_2 > \omega_g$. We have for any $\theta \leq \theta_0$ that
		\begin{align*} \label{sievecomplement2}
		\Pi ( | \bm{\gamma} ( \betab ) | > C_3 s_0 | \theta ) & \leq  \displaystyle \sum_{S: |S| > C_3 s_0 } 2^{|S|}  \theta_0^{|S|}  \displaystyle \int_{ \lVert \betab_g \rVert_2 > \omega_g; g\in S} C_g \lambda_1^{m_g} e^{-\lambda_1 \lVert \betab_g \rVert_2} d \betab_S \\
		& \qquad \times \displaystyle \int_{ \lVert \betab_g \rVert_2 \leq \omega_g; g \in S^c} \Pi_{S^c} ( \betab ) d \betab_{S^c} \\
		& \lesssim \displaystyle \sum_{S: |S| > C_3 s_0 } \theta_0^{|S|}, \numbereqn
		\end{align*}
		where we used the assumption that $\lambda_1 \asymp 1/n$, the definition of $\omega_g$, and the fact that $\theta \leq \theta_0$ to bound the first integral term from above by $ \prod_{g \in S} (1/n)^{m_g} \leq n^{-|S|}$, and we bounded the second integral term above by one. We then have
		\begin{align*} \label{sievecomplement2-pt2}
		& \Pi ( | \bm{\gamma} (\bm{\beta} ) | > C_3 s_0 ) = \int_{0}^{1} \Pi ( | \bm{\gamma} (\bm{\beta} ) | > C_3 s_0 | \theta) d \pi (\theta) \\
		 & \qquad \leq \int_{0}^{\theta_0} \Pi ( | \bm{\gamma} ( \betab ) | > C_3 s_0 | \theta) d \pi (\theta) + \Pi ( \theta > \theta_0 ). \numbereqn
		\end{align*}
		Note that since $s_0 = o(n / \log G)$ by Assumption \ref{A1}, $G \gg n$, and $c > 2$, we have $\theta_ 0 \leq C_3 n / G^c < 1 / G^2$ for sufficiently large $n$. Following from \eqref{sievecomplement2}, we thus have that for sufficiently large $n$,
		\begin{align*} \label{sievecomplement2-pt3}
		 & \int_{0}^{\theta_0} \Pi ( | \bm{\gamma} ( \betab ) | > C_3 s_0 | \theta) d \pi (\theta) \leq \sum_{S: |S| > C_3 s_0} \theta_0^{|S|} \\
		 & \qquad \leq \sum_{k = \lfloor C_3 s_0 \rfloor + 1}^{G} { G \choose k} \left( \frac{1}{G^2} \right)^{k} \\
		 & \qquad \leq \sum_{k= \lfloor C_3 s_0 \rfloor + 1}^{G} \left( \frac{e}{k G} \right)^{k} \\
		 & \qquad < \displaystyle \sum_{k = \lfloor{C_3 s_0} \rfloor+1}^{G} \left( \frac{e}{G (\lfloor{C_3 s_0} \rfloor +1)} \right)^{k} \\
		& \qquad = \frac {\left( \frac{e}{G(\lfloor{C_3 s_0} \rfloor+1)} \right)^{\lfloor{C_3 s_0} \rfloor+1} - \left( \frac{e}{ G (\lfloor{C_3 s_0} \rfloor +1)} \right)^{G+1} }{ 1 - \frac{e}{G(\lfloor{C_3 s_0} \rfloor+1)} } \\
		& \qquad \lesssim G^{- ( \lfloor C_3 s_0 \rfloor+1 )}  \\
		& \qquad \lesssim \exp \left( -C_3 n \epsilon_n^2 \right). \numbereqn
		\end{align*} 
		where we used the inequality ${ G \choose k } \leq (e G / k)^k$ in the third line of the display.
		
		Next, since $\theta \sim \mathcal{B}(1, G^{c})$, we have
		\begin{align*} \label{sievecomplement2-pt4}
		\Pi ( \theta > \theta_0 ) & = (1 - \theta_0)^{G^c} \\
		& = \left( 1 - \frac{ C_3 s_0 \log G}{G^c} \right)^{G^c} \\
		& \leq e^{- C_3 s_0 \log G} \\
		& = e^{-C_3 n \epsilon_n^2}. \numbereqn
		\end{align*}
		Combining \eqref{sievecomplement2-pt2}-\eqref{sievecomplement2-pt4}, we have that
		\begin{align} \label{sievecomplement3}
		   \Pi ( | \bm{\gamma} (\bm{\beta} ) | > C_3 s_0 ) \leq 2 e^{-C_3 n \epsilon_n^2}.
		\end{align}
		Finally, we have as a bound for the second term in \eqref{sievecomplement1},
		\begin{align*} \label{sievecomplement4}
		\Pi \left( \sigma^2 > G^{C_3 s_0/c_0} \right) & = \displaystyle \int_{G^{C_3 s_0 / c_0}}^{\infty} \frac{d_0^{c_0}}{\Gamma(c_0)} (\sigma^2)^{-c_0-1} e^{-d_0 / \sigma^2} d \sigma^2 \\
		& \lesssim \int_{G^{C_3 s_0/c_0}}^{\infty} (\sigma^2)^{-c_0-1} \\
		& \asymp G^{-C_3 s_0} \\
		& \lesssim \exp(-C_3 n \epsilon_n^2). \numbereqn
		\end{align*}
		Combining (\ref{sievecomplement1})-(\ref{sievecomplement4}), we have
		\begin{align*}
		\Pi (\mathcal{F}_n^c) \leq 3 \exp \left( - C_3 n \epsilon_n^2 \right) = \exp \left( -C_3 n \epsilon_n^2 + \log 3 \right),
		\end{align*}
		and so given our choice of $C_3$, (\ref{sievecomplement1}) is asymptotically bounded from above by $\exp(-C_2 n \epsilon_n^2)$ for some $C_2 \geq C_1+2$. This proves (\ref{testingcond1}).
		
		We now proceed to prove (\ref{testingcond2}). Our proof is based on the technique used in \citet{SongLiang2017} with suitable modifications. For $\xi \subset \{1, \ldots, G \}$, let $\Xb_{\xi}$ denote the submatrix of $\Xb$ with submatrices indexed by $\xi$, where $\lvert \xi \rvert \leq \bar{p}$ and $\bar{p}$ is from Assumption \ref{A4}. Let $\widehat{\betab}_{\xi} = ( \Xb_{\xi}^T \Xb_{\xi})^{-1} \Xb_{\xi}^T \yb$ and $\betab_{0 \xi}$ denote the subvector of $\betab_0$ with groups indexed by $\xi$. Let $m_{\xi} = \sum_{g \in \xi} m_g$, and let $\widehat{\sigma}_{\xi}^2 = \lVert \yb - \bm{X}_\xi \widehat{\betab}_{\xi} \rVert_2^2 / (n - m_{\xi} )$. Note that $\widehat{\betab}_{\xi}$ and $\widehat{\sigma}_{\xi}^2$ both exist and are unique because of Assumptions \ref{A1}, \ref{A2}, and \ref{A4} (which combined, gives us that $m_{\xi} = o(n)$). 
		
		Let $\widetilde{p}$ be an integer satisfying $\widetilde{p} \asymp s_0$ and $\widetilde{p} \leq \bar{p} - s_0$, where $\bar{p}$ is from Assumption \ref{A4}, and the specific choice of $\widetilde{p}$ will be given below. Recall that $S_0$ is the set of true nonzero groups with cardinality $s_0 = \lvert S_0 \rvert$. Similar to \cite{SongLiang2017}, we consider the test function $\phi_n = \max \{ \phi_n', \tilde{\phi}_n \}$, where
		\begin{equation} \label{testfunction}
		\begin{array}{ll}
		\phi_n' = \displaystyle \max_{\xi \supset S_0, \lvert \xi \rvert \leq \widetilde{p}+s_0} 1 \left\{ \lvert \widehat{\sigma}_{\xi}^2 - \sigma_0^2 \rvert \geq \sigma_0^2 \epsilon_n \right\}, & \textrm{ and } \\
		\tilde{\phi}_n = \displaystyle \max_{\xi \supset S_0, \lvert \xi \rvert \leq \widetilde{p}+s_0} 1 \left\{ \lVert \widehat{\betab}_{\xi} - \betab_{0 \xi} \rVert_2 \geq \sigma_0 \epsilon_n  \right\}. & 
		\end{array}
		\end{equation}
		Because of Assumption \ref{A4}, we have $\widetilde{p} \prec n$ and $\widetilde{p} \prec n \epsilon_n^2$. Additionally, since $\epsilon_n = o(1)$, we can use almost identical arguments as those used to establish (A.5)-(A.6) in the proof of Theorem A.1 of \cite{SongLiang2017} to show that for any $\xi$ satisfying $\xi \supset S_0, | \xi | \leq \widetilde{p}$,
		\begin{equation*} 
		\mathbb{E}_{( \betab_0, \sigma_0^2 )} 1 \left\{ \lvert \widehat{\sigma}_{\xi}^2 - \sigma_0^2 \rvert \geq \sigma_0^2 \epsilon_n \right\} \leq \exp(- c_4' n \epsilon_n^2), 
		\end{equation*}
		for some constant $\hat{c}_4 > 0$, and for any $\xi$ satisfying $\xi \supset S_0, | \xi | \leq \widetilde{p}$,
		\begin{equation*} 
		\mathbb{E}_{( \betab_0, \sigma_0^2 )} 1 \left\{ \lVert \widehat{\betab}_{\xi} - \betab_{0 \xi} \rVert_2 \geq \sigma_0 \epsilon_n  \right\} \leq \exp(-\tilde{c}_4 n \epsilon_n^2), 
		\end{equation*}
		for some $\tilde{c}_4 > 0$. Using the proof of Theorem A.1 in \cite{SongLiang2017}, we may then choose $\widetilde{p} = \lfloor \min \{ c_4', \tilde{c}_4 \} n \epsilon_n^2 / (2 \log G) \rfloor$, and then 
		\begin{equation} \label{upperboundfirsttesting}
		\mathbb{E}_{f_0} \phi_n \leq \exp ( - \check{c}_4 n \epsilon_n^2 ),
		\end{equation} 
		for some $\check{c}_4 > 0$. Next, define the set,
		\begin{align*}
		\begin{array}{ll}
		\mathcal{C} & = \left\{ \lVert \betab - \betab_0 \rVert_2 \geq (3+\sqrt{\nu_1}) \sigma_0 \epsilon_n \textrm{ or } \sigma^2 / \sigma_0^2 > (1+\epsilon_n)/(1-\epsilon_n) \right. \\
		& \qquad \left.\textrm{ or } \sigma^2 / \sigma_0^2 < (1-\epsilon_n)/(1+\epsilon_n ) \right\}.
		\end{array}
		\end{align*}
		By Lemma \ref{auxlemma2}, we have
		\begin{equation} \label{upperboundsecondtesting1}
		\displaystyle \sup_{ \begin{array}{rl} f \in \mathcal{F}_n: & \lVert \betab - \betab_0 \rVert_2 \geq  (3+\sqrt{\nu_1}) \sigma_0 \epsilon_n , \\ & \textrm{ or } \lvert \sigma^2 - \sigma_0^2 \rvert \geq 4 \sigma_0^2 \epsilon_n \end{array}} \mathbb{E}_f (1-\phi_n) \leq \displaystyle \sup_{ f \in \mathcal{F}_n: (\betab, \sigma^2) \in \mathcal{C}} \mathbb{E}_f (1 - \phi_n). \numbereqn
		\end{equation}
		Similar to \cite{SongLiang2017}, we consider $\mathcal{C} \subset \widehat{\mathcal{C}} \cup \widetilde{\mathcal{C}}$, where
		\begin{align*}
		& \widehat{\mathcal{C}} = \{ \sigma^2/\sigma_0^2 > (1+\epsilon_n)/(1-\epsilon_n) \textrm{ or } \sigma^2 / \sigma_0^2 < (1-\epsilon_n)/(1+\epsilon_n) \}, \\
		& \tilde{\mathcal{C}} = \{  \lVert \betab - \betab_0 \rVert \geq (3 + \sqrt{\nu_1}) \sigma_0 \epsilon_n \textrm{ and } \sigma^2 = \sigma_0^2 \},
		\end{align*}
		and so an upper bound for (\ref{upperboundsecondtesting1}) is
		\begin{align*} \label{upperboundsecondtesting2}
		& \displaystyle \sup_{f\in \mathcal{F}_n: (\betab, \sigma^2) \in \mathcal{C}} \mathbb{E}_f (1-\phi_n) = \displaystyle \sup_{f \in \mathcal{F}_n: (\betab, \sigma^2) \in \mathcal{C}} \mathbb{E}_f \min \{ 1-\phi_n', 1-\tilde{\phi}_n \} \\
		& \qquad \leq \max \left\{ \displaystyle \sup_{f \in \mathcal{F}_n: (\betab, \sigma^2) \in \hat{\mathcal{C}}} \mathbb{E}_f (1-\phi_n'), \displaystyle \sup_{f \in \mathcal{F}_n: (\betab, \sigma^2) \in \tilde{\mathcal{C}}} \mathbb{E}_f (1-\tilde{\phi}_n) \right\}. \numbereqn
		\end{align*}
		Let $\tilde{\xi} = \{ g: \lVert \betab_g \rVert_2 > \omega_g \} \cup S_0$, $m_{\tilde{\xi}} = \sum_{g \in \tilde{\xi}} m_g$, and $\tilde{\xi}^c = \{1, \ldots, G \} \setminus \tilde{\xi}$. For any $f \in \mathcal{F}_n$ such that $(\betab, \sigma^2) \in \hat{\mathcal{C}} \cup \tilde{\mathcal{C}}$, we must have then that $\lvert \tilde{\xi} \rvert \leq C_3 s_0 + s_0 \leq \bar{p}$, by Assumption \ref{A4}. By \eqref{sievecomplement2-pt4}, the prior puts exponentially vanishing probability on values of $\theta > \theta_0$ where $\theta_0 = C_3 s_0 \log G / G^c < 1/(G^2+1)$ for large $G$. Since $\lambda_0 = (1-\theta )/ \theta$ is monotonic decreasing in $\theta$, we have that with probability greater than $1 - e^{-C_3 n \epsilon_n^2}$, $\lambda_0 \geq G^2$. Combining this fact with Assumption \ref{A3} and using $\mathcal{F}_n$ in \eqref{sievedef}, we have that for any $f \in \mathcal{F}_n, (\betab, \sigma^2) \in \hat{C} \cup \tilde{C}$ and sufficiently large $n$,
		\begin{align*} \label{upperboundsecondtesting3}
		\lVert \Xb_{\tilde{\xi}^c} \betab_{\tilde{\xi}^c} \rVert_2 
		& \leq \sqrt{k n^{\alpha}} \lVert \betab_{\tilde{\xi}^c} \rVert_2 \\
		& \leq \sqrt{k n^{\alpha} } \left[ (G - \lvert \tilde{\xi} \rvert ) \displaystyle \max_{g \in \tilde{\xi}^c} \omega_g \right] \\
		& \leq  \sqrt{ k n^{\alpha} }  \left\{ \frac{G}{\lambda_0-\lambda_1} \log \left[ \frac{1-\theta}{\theta} \left( \frac{\lambda_0}{\lambda_1} \right)^{m_{\max}} \right] \right\} \\
		& \lesssim \min \{ \sqrt{k}, 1 \} \times \sqrt{\nu_1}\sqrt{n} \sigma_0 \epsilon_n, \numbereqn
		\end{align*}
		where $\nu$ is from Assumption \ref{A4}. In the above display, we used Lemma \ref{auxlemma3} in the second inequality, while the last inequality follows from our assumptions on $(\theta, \lambda_0, \lambda_1)$ and $m_{\max}$, so one can show that the bracketed term in the third line is asymptotically bounded above by $D \sqrt{\nu_1 } \sqrt{n^{1 -\alpha}} \sigma_0 \epsilon_n$ for large $n$ and any constant $D > 0$. Thus, using nearly identical arguments as those used to prove Part I of Theorem A.1 in \cite{SongLiang2017}, we have
		\begin{align*} \label{upperboundsecondtesting4}
		& \displaystyle \sup_{f \in \mathcal{F}_n: (\betab, \sigma^2) \in \hat{\mathcal{C}}} \mathbb{E}_f (1-\phi_n') \\
		& \qquad\leq \sup_{f \in \mathcal{F}_n: (\betab, \sigma^2) \in \hat{\mathcal{C}}} \Pr \left( \lvert \chi^2_{n-m_{\tilde{\xi}}} (\zeta) - (n - m_{\tilde{\xi}} ) \rvert \geq (n- m_{\tilde{\xi}} ) \epsilon_n \right) \\
		& \qquad \leq \exp( - \hat{c}_4 n \epsilon_n^2), \numbereqn
		\end{align*}
		where the noncentrality parameter $\zeta$ satisfies $\zeta \leq n \epsilon_n^2 \nu_1 \sigma_0^2  / 16 \sigma^2 $, and the last inequality follows from the fact that the noncentral $\chi^2$ distribution is subexponential and Bernstein's inequality (see Lemmas A.1 and A.2 in \cite{SongLiang2017}).
		
		Using the arguments in Part I of the proof of Theorem A.1 in \cite{SongLiang2017}, we also have that for large $n$, 
		\begin{align*}\label{upperboundsecondtesting5}
		& \displaystyle \sup_{f \in \mathcal{F}_n: (\betab, \sigma^2) \in \tilde{\mathcal{C}}} \mathbb{E}_f (1-\tilde{\phi}_n) \\ 
		& \qquad \leq \displaystyle \sup_{f \in \mathcal{F}_n: (\betab, \sigma^2) \in \hat{\mathcal{C}}} \Pr \left( \lVert ( \Xb_{\tilde{\xi}}^T \Xb_{\tilde{\xi}})^{-1} \Xb_{\tilde{\xi}}^T \bm{\varepsilon} \rVert_2 \geq \left[ \lVert \betab_{\tilde{\xi}} - \betab_{0 \tilde{\xi}} \lVert_2 - \sigma_0 \epsilon_n - \right. \right. \\
		& \left. \left.  \qquad \qquad \qquad \qquad \qquad  \lVert ( \Xb_{\tilde{\xi}}^T \Xb_{\tilde{\xi}} )^{-1} \Xb_{\tilde{\xi}}^T \Xb_{\tilde{\xi}^c} \betab_{\tilde{\xi}^c} \rVert_2 \right] / \sigma \right) \\
		& \qquad \leq  \displaystyle \sup_{f \in \mathcal{F}_n: (\betab, \sigma^2) \in \hat{\mathcal{C}}} \Pr \left( \lVert ( \Xb_{\tilde{\xi}}^T \Xb_{\tilde{\xi}})^{-1} \Xb_{\tilde{\xi}}^T \bm{\varepsilon} \rVert_2 \geq  \epsilon_n  \right) \\
		& \qquad \leq \sup_{f \in \mathcal{F}_n: (\betab, \sigma^2) \in \hat{\mathcal{C}}} \Pr ( \chi_{\lvert \widetilde{\xi} \rvert}^2 \geq n \nu_1 \epsilon_n^2 ) \\
		& \qquad \leq \exp ( - \tilde{c}_4 n \epsilon_n^2),  \numbereqn
		\end{align*}
		where the second inequality in the above display holds since $\lVert \betab_{\tilde{\xi}} - \betab_{0 \tilde{\xi}} \rVert_2 \geq \lVert \betab - \betab_0 \rVert_2 - \lVert \betab_{ \tilde{\xi}^c} \rVert_2 $, and since (\ref{upperboundsecondtesting3}) can be further bounded from above by $\sqrt{k n^{\alpha}} \sqrt{\nu_1} \sigma_0 \epsilon_n$ and thus $\lVert \betab_{\tilde{\xi}^c} \rVert \leq \sqrt{ \nu_1} \sigma_0 \epsilon_n$. Therefore, we have for $f \in \mathcal{F}_n, (\betab, \sigma^2) \in \tilde{C}$,
		\begin{align*}
		\lVert \betab_{\tilde{\xi}} - \betab_{0 \tilde{\xi}} \rVert_2  \geq (3 + \sqrt{\nu_1}) \sigma_0 \epsilon_n - \sqrt{\nu_1} \sigma_0 \epsilon_n = 3 \sigma_0 \epsilon_n,
		\end{align*}
		while by Assumption \ref{A4} and (\ref{upperboundsecondtesting3}), we also have
		\begin{align*}
		& \lVert ( \Xb_{\tilde{\xi}}^T \Xb_{\tilde{\xi}})^{-1} \Xb_{\tilde{\xi}}^T \Xb_{\tilde{\xi}^c} \betab_{\tilde{\xi}^c} \rVert_2 \leq \sqrt{ \lambda_{\max} \left( (\Xb_{\tilde{\xi}}^T \Xb_{\tilde{\xi}})^{-1} \right) } \lVert \Xb_{\tilde{\xi}^c} \betab_{\tilde{\xi}^c} \rVert_2 \\
		& \qquad \leq \left( \sqrt{1/n \nu_1} \right) \left( \sqrt{n \nu_1} \sigma_0 \epsilon_n \right) =  \sigma_0 \epsilon_n,
		\end{align*}
		and then we used the fact that on the set $\tilde{C}$, $\sigma = \sigma_0$. The last three inequalities in (\ref{upperboundsecondtesting5}) follow from Assumption \ref{A4}, the fact that $\lvert \widetilde{\xi} \rvert \leq \bar{p} \prec n \epsilon_n^2$, and the fact that for all $m>0$, $\Pr(\chi^2_m \geq x) \leq \exp(-x/4)$ whenever $x \geq 8m$. Altogether, combining (\ref{upperboundsecondtesting1})-(\ref{upperboundsecondtesting5}), we have that
		\begin{equation} \label{upperboundsecondtesting6}
		\displaystyle \sup_{ \begin{array}{rl} f \in \mathcal{F}_n: & \lVert \betab - \betab_0 \rVert_2 \geq (3+\sqrt{\nu}) \sigma_0 \epsilon_n, \\ & \textrm{ or } \lvert \sigma^2 - \sigma_0^2 \rvert \geq 4 \sigma_0^2 \epsilon_n \end{array} } \mathbb{E}_f (1 - \phi_n) \leq \exp \left( - \min \{ \hat{c}_4, \tilde{c}_4 \} n \epsilon_n^2 \right),
		\end{equation}
		where $\hat{c}_4 > 0$ and $\tilde{c}_4 > 0$ are the constants from (\ref{upperboundsecondtesting4}) and (\ref{upperboundsecondtesting5}). 
		
		Now set $C_4 =  \min \{ \hat{c}_4, \tilde{c}_4, \check{c}_4 \}$, where $\check{c}_4$ is the constant from (\ref{upperboundfirsttesting}). By and (\ref{upperboundfirsttesting}) and (\ref{upperboundsecondtesting6}), this choice of $C_4$ will satisfy both testing conditions (\ref{testingcond2}) and (\ref{testingcond3}).
		
		Since we have verified (\ref{KullbackLeiblercond}) and (\ref{testingcond1})-(\ref{testingcond3}) for $\epsilon_n = \sqrt{s_0 \log G / n}$, we have
		\begin{align*}
		\Pi \left( \betab : \lVert \betab - \betab_0 \rVert_2 \geq (3+\sqrt{\nu}) \sigma_0 \epsilon_n \bigg| \yb \right) \rightarrow 0 \textrm{ a.s. } \mathbb{P}_0 \textrm{ as } n, G \rightarrow \infty,
		\end{align*}
		and 
		\begin{align*} 
		\Pi \left( \sigma^2: \lvert \sigma^2 - \sigma_0^2 \rvert \geq 4 \sigma_0^2 \epsilon_n  \vert \yb \right) \rightarrow 0 \textrm{ as } n \rightarrow \infty, \textrm{ a.s. } \mathbb{P}_0 \textrm{ as } n, G \rightarrow \infty,
		\end{align*}
		i.e. we have proven (\ref{l2contraction}) and (\ref{varianceconsistency}).
		\vspace{.4cm}
		
		\noindent \textbf{Part III. Posterior contraction under prediction error loss.}
		The proof is very similar to the proof of (\ref{l2contraction}). The only difference is the testing conditions. We use the same sieve $\mathcal{F}_n$ as that in (\ref{sievedef}) so that (\ref{testingcond1}) holds, but now, we need to show the existence of a different sequence of test functions $\tau_n \in [0,1]$ such that
		\begin{equation} \label{testingcond2prediction}
		\mathbb{E}_{f_0} \tau_n \leq e^{-C_4 n \epsilon_n^2},
		\end{equation}
		and 
		\begin{equation} \label{testingcond3prediction}
		\displaystyle \sup_{ \begin{array}{rl} f \in \mathcal{F}_n: & \lVert \Xb \betab - \Xb \betab_0 \rVert_2 \geq M_2 \sigma_0 \sqrt{n} \epsilon_n, \\ & \textrm{ or } \lvert \sigma^2 - \sigma_0^2 \rvert \geq 4 \sigma_0^2 \epsilon_n \end{array} } \mathbb{E}_f (1 - \tau_n) \leq e^{-C_4 n \epsilon_n^2}.
		\end{equation}
		Let $\widetilde{p}$ be the same integer from (\ref{testfunction}) and consider the test function $\tau_n = \max \{ \tau_n', \tilde{\tau}_n \}$, where
		\begin{equation} \label{testfunctionprediction}
		\begin{array}{ll}
		\tau_n' = \displaystyle \max_{\xi \supset S_0, \lvert \xi \rvert \leq \widetilde{p}+s_0} 1 \left\{ \lvert \widehat{\sigma}_{\xi}^2 - \sigma_0^2 \rvert \geq \sigma_0^2 \epsilon_n \right\}, & \textrm{ and } \\
		\tilde{\tau}_n = \displaystyle \max_{\xi \supset S_0, \lvert \xi \rvert \leq \widetilde{p}+s_0} 1 \left\{ \lVert \Xb_{\xi} \widehat{\betab}_{\xi} - \Xb_{\xi} \betab_{0 \xi} \rVert_2 \geq \sigma_0 \sqrt{n} \epsilon_n  \right\}. & 
		\end{array}
		\end{equation}
		Using Assumption \ref{A4} that for any $\xi \subset \{1, \ldots, G \}$ such that $\lvert \xi \rvert \leq \bar{p}$, $\lambda_{\max} ( \Xb_{\xi}^T \Xb_{\xi}) \leq n \nu_2$ for some $\nu_2 > 0$, we have that
		\begin{equation*}
		\lVert \Xb_{\xi} \widehat{\betab}_{\xi} - \Xb_{\xi} \betab_{0 \xi} \lVert_2 \leq \sqrt{n \nu_2 }  \lVert \widehat{\betab}_{\xi} - \betab_{0 \xi} \lVert_2,
		\end{equation*}
		and so
		\begin{equation*}
		\Pr \left(  \lVert \Xb_{\xi} \widehat{\betab}_{\xi} - \Xb_{\xi} \betab_{0 \xi} \lVert_2 \geq \sigma_0 \sqrt{n} \epsilon_n \right) \leq \Pr \left( \lVert \widehat{\betab}_{\xi} - \betab_{0 \xi} \lVert_2 \geq \nu_2^{-1/2} \sigma_0 \epsilon_n  \right).
		\end{equation*}
		Therefore, using similar steps as those in Part II of the proof, we can show that our chosen sequence of tests $\tau_n$ satisfies (\ref{testingcond2prediction}) and (\ref{testingcond3prediction}). We thus arrive at
		\begin{align*}
		\Pi \left( \betab : \lVert \betab - \betab_0 \rVert_2 \geq M_2 \sigma_0 \epsilon_n \bigg| \yb \right) \rightarrow 0 \textrm{ a.s. } \mathbb{P}_0 \textrm{ as } n, G \rightarrow \infty,
		\end{align*}
		i.e. we have proven (\ref{predictioncontraction}).
	\end{proof}
	
	\begin{proof}[Proof of Theorem \ref{dimensionalitygroupedregression}]
		According to Part I of the proof of Theorem \ref{posteriorcontractiongroupedregression}, we have that for $\epsilon_n = \sqrt{s_0 \log G / n}$,
		\begin{align*}
		\Pi \left( K(f_0, f) \leq n \epsilon_n^2, V(f_0, f) \leq n \epsilon_n^2 \right) \geq \exp \left( -C n \epsilon_n^2 \right) 
		\end{align*}
		for some $C>0$. Thus, by Lemma 8.10 of \cite{GhosalVanDerVaart2017}, there exist positive constants $C_1$ and $C_2$ such that the event,
		\begin{equation} \label{eventEn}
		E_n = \left\{ \displaystyle \int \displaystyle \int \frac{f (\yb) }{f_0 (\yb)} d \Pi (\betab) d \Pi(\sigma^2) \geq e^{-C_1 n \epsilon_n^2} \right\},
		\end{equation}
		satisfies 
		\begin{equation} \label{probEncomplement}
		\mathbb{P}_0 ( E_n^c) \leq e^{-(1+C_2) n \epsilon_n^2}.
		\end{equation}
		Define the set $\mathcal{T} = \{ \betab : \lvert \bm{\gamma} (\betab) \rvert \leq C_3 s_0 \}$, where we choose $C_3 > 1+C_2$. We must show that $\mathbb{E}_0 \Pi ( \mathcal{T}^c \vert \yb ) \rightarrow 0$ as $n \rightarrow \infty.$ The posterior probability $\Pi ( \mathcal{T}^c \vert \yb)$ is given by
		\begin{equation} \label{posteriorprobTcomplement}
		\Pi ( \mathcal{T}^c \vert \yb ) = \frac{ \displaystyle \int \int_{\mathcal{T}^c} \frac{f (\yb)}{f_0 (\yb)} d \Pi (\betab) d \Pi ( \sigma^2)}{ \displaystyle \int \int \frac{f (\yb) }{f_0 (\yb)} d \Pi (\betab) d \Pi (\sigma^2)}.
		\end{equation}
		By (\ref{probEncomplement}), the denominator of (\ref{posteriorprobTcomplement}) is bounded below by $e^{-(1+C_2) n \epsilon_n^2}$. For the numerator of (\ref{posteriorprobTcomplement}), we have as an upper bound,
		\begin{equation} \label{numeratorupperbound}
		\mathbb{E}_0 \left( \displaystyle \int \int_{\mathcal{T}^c} \frac{f (\yb)}{f_0 (\yb)}  d \Pi (\betab) \Pi (\sigma^2) \right) \leq \displaystyle \int_{\mathcal{T}^c} d \Pi (\betab) = \Pi \left( \lvert \bm{\gamma} ( \betab ) \rvert > C_3 s_0 \right). 
		\end{equation}
		Using the same arguments as (\ref{sievecomplement2})-(\ref{sievecomplement3}) in the proof of Theorem \ref{posteriorcontractiongroupedregression}, we can show that 
		\begin{equation} \label{numeratorupperbound2}
		\Pi \left( \lvert \bm{\gamma} ( \betab ) \rvert > C_3 s_0 \right) \prec e^{-C_3 n \epsilon_n^2}.
		\end{equation}
		Combining (\ref{eventEn})-(\ref{numeratorupperbound}), we have that
		\begin{align*}
		\mathbb{E}_0 \Pi \left( \mathcal{T}^c \vert \yb \right) &  \leq \mathbb{E}_0 \Pi ( \mathcal{T}^c \vert \yb ) 1_{E_n} + \mathbb{P}_0 (E_n^c) \\
		& < \exp \left( (1+C_2) n \epsilon_n^2 - C_3 n \epsilon_n^2 \right) + o(1) \\
		& \rightarrow 0 \textrm{ as } n, G \rightarrow \infty,
		\end{align*}
		since $C_3 > 1+C_2$. This proves (\ref{posteriorcompressibility}).
	\end{proof}
	
	\begin{proof}[Proof of Theorem \ref{contractionGAMs}]
		Let $f_{0j}(\bm{X}_j)$ be an $n \times 1$ vector with $i$th entry equal to $f_{0j}(X_{ij})$. Note that proving posterior contraction with respect to the empirical norm \eqref{empiricalcontractionGAM} is equivalent to proving that
		\begin{equation} \label{predictioncontractionGAM}
		\Pi \left( \bm{\beta}: \lVert \widetilde{\bm{X}} \bm{\beta} - \sum_{j=1}^{p} f_{0j} (\bm{X}_j) \rVert_2 \geq \widetilde{M}_1 \sqrt{n} \epsilon_n \bigg| \yb \right) \rightarrow 0 \textrm{ a.s. } \widetilde{\mathbb{P}}_0 \textrm{ as } n, p \rightarrow \infty,
		\end{equation}
		so to prove the theorem, it suffices to prove \eqref{predictioncontractionGAM}. Let $f \sim \mathcal{N}_n( \widetilde{\Xb} \betab, \sigma^2 \bm{I}_n)$ and $f_0 \sim \mathcal{N}_n ( \widetilde{\Xb} \betab_{0} + \bm{\delta}, \sigma_0^2 \bm{I}_n )$, and let $\Pi(\cdot)$ denote the prior (\ref{hiermodel}). Similar to the proof for Theorem \ref{posteriorcontractiongroupedregression}, we show that for our choice of $\epsilon_n^2 = s_0 \log p /n + s_0 n^{-2 \kappa / (2 \kappa + 1)}$ and some constant $C_1 > 0$,
		\begin{equation} \label{KullbackLeiblercondGAMs}
		\Pi \left( K (f_0, f) \leq n \epsilon_n^2, V(f_0, f) \leq n \epsilon_n^2 \right) \geq \exp (-C_1 n \epsilon_n^2),
		\end{equation}
		and the existence of a sieve $\mathcal{F}_n$ such that 
		\begin{equation} \label{testingGAMcond1}
		\Pi ( \mathcal{F}_n^c) \leq \exp(-C_2 n \epsilon_n^2),
		\end{equation}
		for positive constant $C_2 > C_1+2$, and a sequence of test functions $\phi_n \in [0,1]$ such that
		\begin{equation} \label{testingGAMcond2}
		\mathbb{E}_{f_0} \phi_n \leq e^{-C_4 n \epsilon_n^2},
		\end{equation}
		and 
		\begin{equation} \label{testingGAMcond3}
		\displaystyle \sup_{ \begin{array}{rl} f \in \mathcal{F}_n: & \lVert \widetilde{\Xb} \betab - \sum_{j=1}^{p} f_{0j} ( \Xb_{j} ) \rVert_2 \geq \tilde{c}_0 \sigma_0 \sqrt{n} \epsilon_n, \\ & \textrm{ or } \lvert \sigma^2 - \sigma_0^2 \rvert \geq 4 \sigma_0^2 \epsilon_n \end{array} } \mathbb{E}_f (1 - \phi_n) \leq e^{-C_4 n \epsilon_n^2},
		\end{equation}
		for some $C_4 > 0$ and $\tilde{c}_0 > 0$.
		
		We first verify (\ref{KullbackLeiblercondGAMs}). The KL divergence between $f_0$ and $f$ is
		\begin{equation} \label{KLdivGAM}
		K(f_0, f) = \frac{1}{2} \left[ n \left( \frac{\sigma_0^2}{\sigma^2} \right) - n - n \log \left( \frac{\sigma_0^2}{\sigma^2} \right) + \frac{ \lVert \widetilde{\Xb} ( \betab - \betab_0 ) - \bm{\delta} \rVert_2^2}{\sigma^2}   \right],
		\end{equation}
		and the KL variation between $f_0$ and $f$ is
		\begin{equation} \label{KLvarGAM}
		V(f_0, f) = \frac{1}{2} \left[ n \left( \frac{\sigma_0^2}{\sigma^2} \right)^2 - 2n \left( \frac{\sigma_0^2}{\sigma^2} \right) + n  \right] + \frac{\sigma_0^2}{(\sigma^2)^{2}} \lVert \widetilde{\Xb} ( \betab - \betab_0 ) - \bm{\delta} \rVert_2^2.
		\end{equation}
		Define the two events $\widetilde{\mathcal{A}}_1$ and $\widetilde{\mathcal{A}}_2$ as follows:
		\begin{equation} \label{eventA1GAM}
		\widetilde{\mathcal{A}}_1 = \left\{ \sigma^2: n \left( \frac{\sigma_0^2}{\sigma^2} \right) - n - n \log \left( \frac{\sigma_0^2}{\sigma^2} \right) \leq n \epsilon_n^2, \right. \\
		\left. n \left( \frac{\sigma_0^2}{\sigma^2} \right)^2 - 2n \left( \frac{\sigma_0^2}{\sigma^2} \right) + n \leq n \epsilon_n^2  \right\}
		\end{equation}
		and
		\begin{equation} \label{eventA2GAM}
		\widetilde{\mathcal{A}}_2 = \left\{ (\betab, \sigma^2): \frac{ \lVert \widetilde{\bm{X}} ( \bm{\beta} - \bm{\beta}_0 ) - \bm{\delta} \rVert_2^2}{\sigma^2} \leq n \epsilon_n^2, \right. \\
		\left. \frac{\sigma_0^2}{(\sigma^2)^{2}} \lVert \widetilde{\Xb} ( \betab - \betab_0 ) - \bm{\delta} \rVert_2^2 \leq n \epsilon_n^2/2 \right\}.
		\end{equation}
		Following from (\ref{KLdivGAM})-(\ref{eventA2GAM}), we have $\Pi ( K(f_0, f) \leq n \epsilon_n^2, V(f_0, f) \leq n \epsilon_n^2) = \Pi (\widetilde{\mathcal{A}}_2 \vert \widetilde{\mathcal{A}}_1 ) \Pi (\widetilde{\mathcal{A}}_1)$. Using the steps we used to prove (\ref{neglogA1upper}) in part I of the proof of Theorem \ref{posteriorcontractiongroupedregression}, we have
		\begin{equation} \label{A1upperGAM}
		\Pi ( \widetilde{\mathcal{A}}_1 ) \gtrsim \exp (- C_1 n \epsilon_n^2 / 2),
		\end{equation}
		for some sufficiently large $C_1 > 0$. Following similar reasoning as in the proof of Theorem \ref{posteriorcontractiongroupedregression}, we also have for some $b_2 > 0$,
		\begin{equation} \label{lowerboundA2givenA1GAM}
		\Pi \left( \widetilde{A}_2 \vert \widetilde{A}_1 \right) \geq \Pi \left( \lVert \widetilde{\Xb} (\betab - \betab_{0}) - \bm{\delta} \rVert_2^2 \leq \frac{b_2^2 n \epsilon_n^2}{2} \right). 
		\end{equation}
		Using Assumptions \ref{B3} and \ref{B6}, we then have
		\begin{align*}
		\lVert \widetilde{\Xb} ( \betab - \betab_{0}) - \bm{\delta} \rVert_2^2 & \leq \left( \lVert \widetilde{\Xb} ( \betab - \betab_0) \rVert_2 + \lVert \bm{\delta} \rVert_2 \right)^2 \\
		& \leq 2 \lVert \widetilde{\Xb} ( \betab - \betab_0 )  \rVert_2^2 + 2 \lVert \bm{\delta} \rVert_2^2 \\
		& \lesssim 2 \left( n k_1 \lVert \betab - \betab_0 \rVert_2^2 + \frac{ k_1 b_2^2 n s_0 d^{-2 \kappa}}{4} \right) \\ 
		& \asymp 2n \left( \lVert \betab - \betab_0 \rVert_2^2 + \frac{ b_2^2 s_0 d^{-2 \kappa}}{4} \right),
		\end{align*}
		and so (\ref{lowerboundA2givenA1GAM}) can be asymptotically lower bounded by
		\begin{align*} \label{lowerboundA2givenA1GAM2}
		& \Pi \left(  \lVert \betab - \betab_0 \rVert_2^2 + \frac{b_2^2 s_0 d^{-2 \kappa}}{4} \leq \frac{b_2^2 \epsilon_n^2}{4 } \right) \\
		& = \Pi \left( \lVert \betab - \betab_0 \rVert_2^2 \leq \frac{b_2^2}{4} \left( \epsilon_n^2 - s_0 n^{- 2 \kappa / (2 \kappa + 1)} \right) \right),
		\end{align*}
		where we used Assumption \ref{B1} that $d \asymp n^{1 / (2 \kappa + 1)}$. Using very similar arguments as those used to prove (\ref{neglogpiA2givenA1pt2}), this term can also be lower bounded by  $\exp (- C_1 n \epsilon_n^2 /2 )$. Altogether, we have
		\begin{equation} \label{lowerboundA2givenA1GAM3}
		\Pi( \widetilde{A}_2 \vert \widetilde{A}_1 ) \gtrsim \exp ( -  C_1 \epsilon_n^2 / 2).
		\end{equation}
		Combining (\ref{A1upperGAM}) and (\ref{lowerboundA2givenA1GAM3}), we have that (\ref{KullbackLeiblercondGAMs}) holds. To verify (\ref{testingGAMcond1}), we choose $C_3 \geq C_1 + 2 + \log 3$ and use the same sieve $\mathcal{F}_n$ as the one we employed in the proof of Theorem \ref{posteriorcontractiongroupedregression} (eq. (\ref{sievedef})), and then (\ref{testingGAMcond1}) holds for our choice of $\mathcal{F}_n$.
		
		Finally, we follow the recipe of  \citet{WeiReichHoppinGhosal2018} and \citet{SongLiang2017} to construct our test function $\phi_n$ which will satisfy both (\ref{testingGAMcond2}) and (\ref{testingGAMcond3}). For $\xi \subset \{1, \ldots, p \}$, let $\widetilde{\Xb}_{\xi}$ denote the submatrix of $\widetilde{\Xb}$ with submatrices indexed by $\xi$, where $\lvert \xi \rvert \leq \bar{p}$ and $\bar{p}$ is from Assumption \ref{B4}. Let $\widehat{\betab}_{\xi} = ( \widetilde{\Xb}_{\xi}^T \widetilde{\Xb}_{\xi})^{-1} \widetilde{\Xb}_{\xi}^T \yb$ and $\betab_{0 \xi}$ denote the subvector of $\betab_0$ with basis coefficients appearing in $\xi$. Then the total number of elements in $\widehat{\betab}_{\xi}$ is $d \lvert \xi \rvert $. Finally, let $\widehat{\sigma}_{\xi}^2 = \yb^T ( \bm{I}_n - \bm{H}_{\xi} ) \yb / (n - d \lvert \xi \rvert  )$, where $\bm{H}_{\xi} = \widetilde{\Xb}_{\xi} ( \widetilde{\Xb}_{\xi}^T \widetilde{\Xb}_{\xi} )^{-1} \widetilde{\Xb}_{\xi}^T$ is the hat matrix for the subgroup $\xi$. 
		
		Let $\widetilde{p}$ be an integer satisfying $\widetilde{p} \asymp s_0$ and $\widetilde{p} \leq \bar{p} - s_0$, where $\bar{p}$ is from Assumption \ref{B4} and the specific choice for $\widetilde{p}$ will be given later. Recall that $S_0$ is the set of true nonzero groups with cardinality $s_0 = \lvert S_0 \rvert$. Similar to \cite{WeiReichHoppinGhosal2018}, we consider the test function, $\phi_n = \max \{ \phi_n', \tilde{\phi}_n \}$, where
		\begin{equation} \label{testfunctionGAM}
		\begin{array}{ll}
		\phi_n' = \displaystyle \max_{\xi \supset S_0, \lvert \xi \rvert \leq \widetilde{p}+s_0} 1 \left\{ \lvert \widehat{\sigma}_{\xi}^2 - \sigma_0^2 \rvert \geq c_0 ' \sigma_0^2 \epsilon_n \right\}, & \textrm{ and } \\
		\tilde{\phi}_n = \displaystyle \max_{\xi \supset S_0, \lvert \xi \rvert \leq \widetilde{p}+s_0} 1 \left\{ \bigg| \bigg| \widetilde{\Xb} \widehat{\betab}_{\xi} - \displaystyle \sum_{j \in \xi} f_{0j} ( \Xb_{j} ) \bigg| \bigg|_2 \geq \tilde{c}_0 \sigma_0 \sqrt{n} \epsilon_n  \right\}, & 
		\end{array}
		\end{equation}
		for some positive constants $c_0 '$ and $\tilde{c}_0$. Using Assumptions \ref{B1} and \ref{B4}, we have that for any $\xi$ in our test $\phi_n$, $d \lvert \xi \rvert \leq d (\widetilde{p}+s_0) \leq d \bar{p} \prec n \epsilon_n^2$. Using essentially the same arguments as those in the proof for Theorem 4.1 in \cite{WeiReichHoppinGhosal2018}, we have that for any $\xi$ which satisfies $\xi \supset S_0$ so that $\lvert \xi \rvert \leq \widetilde{p} + s_0$,
		\begin{equation} \label{upperboundfirsttestingGAM1}
		\mathbb{E}_{(\betab_0, \sigma_0^2)} 1 \left\{ \lvert \widehat{\sigma}_{\xi}^2 - \sigma_0^2 \rvert \geq c_0 ' \epsilon_n \right\} \leq \exp(-c_4' n \epsilon_n^2),
		\end{equation}
		for some $c_0 '' > 0$. By Assumption \ref{B3}, we also have
		\begin{align*}
		\bigg| \bigg| \widetilde{\Xb} \widehat{\betab} - \displaystyle \sum_{j=1}^{p} f_{0j} ( \Xb_{j} ) \bigg| \bigg|_2 & = \lVert \widetilde{\Xb} (\widehat{\betab} - \betab_0 ) - \bm{\delta} \rVert_2 \\
		& \leq \sqrt{n k_1} \lVert \widehat{\betab} - \betab_0 \rVert_2 + \lVert \bm{\delta} \rVert_2,
		\end{align*}
		and using the fact that $ \lVert \bm{\delta} \rVert_2 \lesssim \sqrt{n s_0} d^{-\kappa} \lesssim \tilde{c}_0 \sigma_0 \sqrt{n} \epsilon_n / 2$ (by Assumptions \ref{B1} and \ref{B6}), we have that for any $\xi$ such that $\xi \supset S_0, \lvert \xi \rvert \leq \widetilde{p} + s_0$,
		\begin{align*}
		& \mathbb{E}_{(\betab_0, \sigma_0^2)} 1 \left\{ \bigg| \bigg| \widetilde{\Xb} \widehat{\betab} - \displaystyle \sum_{j=1}^{p} f_{0j} ( \Xb_j ) \bigg| \bigg|_2 \geq \tilde{c}_0 \sigma_0 \sqrt{n} \epsilon_n \right\} \\
		& \qquad \leq \mathbb{E}_{(\betab_0, \sigma_0^2)} \left\{ \lVert \widehat{\betab} - \betab_0 \rVert_2 \geq \tilde{c}_0 \sigma_0 \epsilon_n / 2 \sqrt{k_1} \right\} \\
		& \qquad \leq \exp ( -\tilde{c}_4 n \epsilon_n^2), 
		\end{align*}
		for some $\tilde{c}_4 > 0$, where we used the proof of Theorem A.1 in \cite{SongLiang2017} to arrive at the final inequality. Again, as in the proof of Theorem A.1 of \cite{SongLiang2017}, we choose $\widetilde{p} = \lfloor \min \{ c_4', \tilde{c}_4 \} n \epsilon_n^2 / (2 \log p) \rfloor$, and then
		\begin{equation} \label{upperboundfirsttestingGAM2}
		\mathbb{E}_{f_0} \phi_n \leq \exp( - \check{c}_4 n \epsilon_n^2),
		\end{equation}
		for some $\check{c}_4 > 0$. Next, we define the set,
		\begin{align*}
		\begin{array}{ll}
		\mathcal{C} & = \left\{ \lVert \widetilde{\Xb} \betab - \sum_{j=1}^{p} f_{0j} ( \Xb_j ) \rVert_2 \geq \tilde{c}_0 \sigma_0 \sqrt{n} \epsilon_n \textrm{ or } \sigma^2 / \sigma_0^2 > (1+\epsilon_n)/(1-\epsilon_n) \right. \\
		& \qquad \left.\textrm{ or } \sigma^2 / \sigma_0^2 < (1-\epsilon_n)/(1+\epsilon_n ) \right\}
		\end{array}.
		\end{align*}
		By Lemma \ref{auxlemma2}, we have
		\begin{align*} \label{upperboundsecondtesting1GAM}
		& \displaystyle \sup_{ \begin{array}{rl} f \in \mathcal{F}_n: & \lVert \widetilde{\Xb} \betab - \sum_{j=1}^{p} f_{0j} ( \Xb_j ) \rVert_2 \geq \tilde{c}_0 \sigma_0 \sqrt{n} \epsilon_n , \\ & \textrm{ or } \lvert \sigma^2 - \sigma_0^2 \rvert \geq 4 \sigma_0^2 \epsilon_n \end{array}} \mathbb{E}_f (1-\phi_n) \\
		&\qquad \qquad \leq \displaystyle \sup_{ f \in \mathcal{F}_n: (\betab, \sigma^2) \in \mathcal{C}} \mathbb{E}_f (1 - \phi_n). \numbereqn
		\end{align*}
		Similar to \cite{SongLiang2017}, we consider $\mathcal{C} \subset \widehat{\mathcal{C}} \cup \widetilde{\mathcal{C}}$, where
		\begin{align*}
		& \widehat{\mathcal{C}} = \{ \sigma^2/\sigma_0^2 > (1+\epsilon_n)/(1-\epsilon_n) \textrm{ or } \sigma^2 / \sigma_0^2 < (1-\epsilon_n)/(1+\epsilon_n) \}, \\
		& \tilde{\mathcal{C}} = \{  \lVert \widetilde{\Xb} \betab - \sum_{j=1}^{p} f_{0j} ( \Xb_j ) \rVert_2 \geq \tilde{c}_0 \sigma_0 \epsilon_n \textrm{ and } \sigma^2 = \sigma_0^2 \},
		\end{align*}
		and so an upper bound for (\ref{upperboundsecondtesting1GAM}) is
		\begin{align*} \label{upperboundsecondtesting2GAM}
		& \displaystyle \sup_{f\in \mathcal{F}_n: (\betab, \sigma^2) \in \mathcal{C}} \mathbb{E}_f (1-\phi_n) = \displaystyle \sup_{f \in \mathcal{F}_n: (\betab, \sigma^2) \in \mathcal{C}} \mathbb{E}_f \min \{ 1-\phi_n', 1-\tilde{\phi}_n \} \\
		& \qquad \leq \max \left\{ \displaystyle \sup_{f \in \mathcal{F}_n: (\betab, \sigma^2) \in \hat{\mathcal{C}}} \mathbb{E}_f (1-\phi_n'), \displaystyle \sup_{f \in \mathcal{F}_n: (\betab, \sigma^2) \in \tilde{\mathcal{C}}} \mathbb{E}_f (1-\tilde{\phi}_n) \right\}. \numbereqn
		\end{align*}
		Using very similar arguments as those used to prove (\ref{upperboundsecondtesting6}) in Theorem \ref{posteriorcontractiongroupedregression} and using Assumptions \ref{B1} and \ref{B6}, so that the bias $\lVert \bm{\delta} \rVert_2^2 \lesssim n s_0 d^{-2 \kappa} \lesssim n \epsilon_n^2 $, we can show that  (\ref{upperboundsecondtesting2GAM}) can be further bounded from above as
		\begin{align*}  \label{upperboundsecondtesting3GAM}
		& \displaystyle \sup_{ \begin{array}{rl} f \in \mathcal{F}_n: & \lVert \widetilde{\Xb} \betab - \sum_{j=1}^{p} f_{0j} ( \Xb_j ) \rVert_2 \geq \tilde{c}_0 \sigma_0 \sqrt{n} \epsilon_n , \\ & \textrm{ or } \lvert \sigma^2 - \sigma_0^2 \rvert \geq 4 \sigma_0^2 \epsilon_n \end{array}} \mathbb{E}_f (1-\phi_n) \\
		& \qquad \qquad \leq \exp \left( - \min \{ \hat{c}_4, \tilde{c}_4 \} n \epsilon_n^2 \right), \numbereqn
		\end{align*}
		where $\hat{c}_4 > 0$ and $\tilde{c}_4 > 0$ are the constants from (\ref{upperboundfirsttestingGAM1}) and (\ref{upperboundfirsttestingGAM2}). 
		
		Choose $C_4 = \min \{ \check{c}_4, \hat{c}_4, \tilde{c}_4 \}$, and we have from (\ref{upperboundfirsttestingGAM2}) and (\ref{upperboundsecondtesting3GAM}) that (\ref{testingGAMcond2}) and (\ref{testingGAMcond3}) both hold. 
		
		Since we have verified (\ref{KullbackLeiblercondGAMs}) and (\ref{testingGAMcond1})-(\ref{testingGAMcond3})  for our choice of $\epsilon_n^2 = s_0 \log p / n + s_0 n^{-2 \kappa / (2 \kappa + 1)}$, it follows that
		\begin{align*}
		\Pi \left( \betab : \bigg| \bigg| \widetilde{\Xb} \betab - \displaystyle \sum_{j=1}^{p} f_{0j} ( \bm{X}_j ) \bigg| \bigg|_2 \geq \tilde{c}_0 \sigma_0 \sqrt{n} \epsilon_n \vert \yb \right) \rightarrow 0 \textrm{ a.s. } \widetilde{\mathbb{P}}_0 \textrm{ as } n, p \rightarrow \infty,
		\end{align*}
		and 
		\begin{align*} 
		\Pi \left( \sigma^2: \lvert \sigma^2 - \sigma_0^2 \rvert \geq 4 \sigma_0^2 \epsilon_n  \vert \yb \right) \rightarrow 0 \textrm{ as } n \rightarrow \infty, \textrm{ a.s. } \widetilde{\mathbb{P}}_0 \textrm{ as } n, p \rightarrow \infty,
		\end{align*}
		i.e. we have proven (\ref{predictioncontractionGAM}), or equivalently, (\ref{empiricalcontractionGAM}) and (\ref{GAMvarianceconsistency}).
		
	\end{proof}
	
	\begin{proof}[Proof of Theorem \ref{dimensionalityGAM}]
		The proof is very similar to the proof of Theorem \ref{dimensionalitygroupedregression} and is thus omitted.
	\end{proof}
	
\end{appendix}

\end{document}